\documentclass[10pt,a4paper]{article}
\usepackage[utf8]{inputenc}
\usepackage[T1]{fontenc}
\usepackage{microtype}
\usepackage{fourier}

\usepackage{tikz}
\usepackage[margin=10pt,font=small,labelfont=bf,labelsep=period]{caption}
\usepackage{amsmath}
\usepackage{amsthm}
\usepackage{amssymb}
\usepackage{mathscinet}
\numberwithin{equation}{section}

\usepackage[nottoc,notlot,notlof]{tocbibind}

\usepackage{aliascnt}
\usepackage[colorlinks,linkcolor=blue,citecolor=blue,urlcolor=blue]{hyperref}

\theoremstyle{plain}

\newtheorem{theorem}{Theorem}[section]

\newaliascnt{lemma}{theorem}
\newtheorem{lemma}[lemma]{Lemma}
\aliascntresetthe{lemma}

\newaliascnt{corollary}{theorem}
\newtheorem{corollary}[corollary]{Corollary}
\aliascntresetthe{corollary}

\theoremstyle{definition}

\newaliascnt{definition}{theorem}
\newtheorem{definition}[definition]{Definition}
\aliascntresetthe{definition}

\newaliascnt{example}{theorem}
\newtheorem{example}[example]{Example}
\aliascntresetthe{example}

\newaliascnt{remark}{theorem}
\newtheorem{remark}[remark]{Remark}
\aliascntresetthe{remark}

\newcommand{\R}{\mathbf{R}}

\renewcommand{\epsilon}{\varepsilon}

\newcommand{\abs}[1]{\left\lvert #1 \right\rvert}
\newcommand{\avg}[1]{\bigl\langle #1 \bigr\rangle}
\newcommand{\jump}[1]{\bigl[ #1 \bigr]}
\newcommand{\twin}{\widetilde}
\newcommand{\da}{d\alpha(x)}
\newcommand{\db}{d\beta(y)}
\newcommand{\dA}[1]{d\alpha^{#1}(x)}
\newcommand{\dB}[1]{d\beta^{#1}(y)}
\newcommand{\heineintegral}{\mathcal{J}}
\newcommand{\piecewiseclass}{PC^{\infty}}
\newcommand{\spaceDprime}{\mathcal{D}'(\R)}
\DeclareMathOperator{\sgn}{sgn}

\title{Dynamics of interlacing peakons (and shockpeakons) in the Geng--Xue equation}

\author{%
  Hans Lundmark\thanks{Department of Mathematics, Linköping University, SE-581\,83 Linköping, Sweden; hans.lundmark@liu.se}
  \and
  Jacek Szmigielski\thanks{Department of Mathematics and Statistics, University of Saskatchewan, 106 Wiggins Road, Saskatoon, Saskatchewan, S7N\,5E6, Canada; szmigiel@math.usask.ca}
}

\date{December 8, 2016}

\makeatletter
\hypersetup{%
  pdfauthor={Hans Lundmark, Jacek Szmigielski},
  pdftitle={\@title},
}
\makeatother

\begin{document}

\maketitle

\begin{abstract}
  We consider multi\-peakon solutions,
  and to some extent also multi\-shock\-peakon solutions,
  of a coupled two-component
  integrable PDE found by Geng and Xue
  as a generalization of Novikov's cubically non\-linear Camassa--Holm type
  equation.
  In order to make sense of such solutions,
  we find it necessary to assume that there are no overlaps,
  meaning that a peakon or shockpeakon in one component is not allowed to occupy
  the same position as a peakon or shockpeakon in the other component.
  Therefore one can distinguish many inequivalent configurations,
  depending on the order in which the peakons or shockpeakons
  in the two components appear relative to each other.
  Here we are in particular interested in the case of
  \emph{interlacing} peakon solutions,
  where the peakons alternatingly occur in one component and in the other.
  Based on explicit expressions for these
  solutions in terms of elementary functions,
  we describe the general features of the dynamics,
  and in particular the asymptotic large-time behaviour
  (assuming that there are no anti\-peakons, so that the solutions
  are globally defined).
  As far as the positions are concerned,
  interlacing Geng--Xue peakons display the usual scattering phenomenon
  where the peakons asymptotically
  travel with constant velocities, which are all distinct, except that the
  two fastest peakons (the fastest one in each component)
  will have the same velocity.
  However, in contrast to many other peakon equations,
  the amplitudes of the peakons will not in general tend to
  constant values; instead they grow or decay exponentially.
  Thus the logarithms of the amplitudes (as functions of time)
  will asymptotically behave like straight lines,
  and comparing these lines for large positive and negative times,
  one observes phase shifts similar to those seen for the positions
  of the peakons (and also for the positions of solitons in many other contexts).
  In addition to these $K+K$ interlacing pure peakon solutions,
  we also investigate $1+1$ shockpeakon solutions, 
  and collisions leading to shock formation
  in a $2+2$ peakon--antipeakon solution.
\end{abstract}

\tableofcontents{}

\section{Introduction}
\label{sec:intro}

The \textbf{Geng--Xue equation} \cite{geng-xue:cubic-nonlinearity},
also known as the \textbf{two-component Novikov equation},
is the following integrable two-component PDE in $1+1$ dimensions:
\begin{equation}
  \label{eq:GX}
  \begin{gathered}
    m_t + (m_xu + 3mu_x)v = 0
    , \\
    n_t + (n_xv + 3nv_x)u = 0
    ,
  \end{gathered}
\end{equation}
where $m = u - u_{xx}$ and $n = v - v_{xx}$
are auxiliary quantitites associated to the two unknown functions $u(x,t)$ and $v(x,t)$,
and where subscripts denote partial derivatives, as usual.
It is a close mathematical relative of the
Degasperis--Procesi and Novikov equations,
which are in turn offspring of the Camassa--Holm shallow water wave
equation; see \autoref{sec:history}.

\subsection{Peakon solutions of the Geng--Xue equation}
\label{sec:intro-peakons}

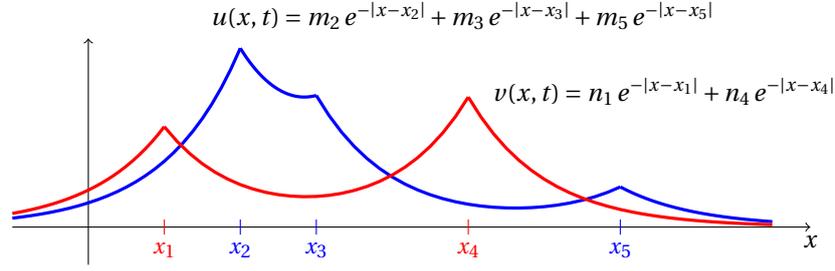
\begin{figure}
  \centering
  \begin{tikzpicture}
    \draw[->] (-1, 0) -- (9.5, 0) node[below] {$x$};
    \draw[->] (0, -0.5) -- (0, 2.5);

    \def\uformula{plot (\noexpand\x,{2*exp(-abs(\noexpand\x-2))+1*exp(-abs(\noexpand\x-3))+0.5*exp(-abs(\noexpand\x-7))})}
    \def\vformula{plot (\noexpand\x,{1.3*exp(-abs(\noexpand\x-1))+1.7*exp(-abs(\noexpand\x-5))})}

    \draw[red]  (1, 0.1) -- +(0, -0.2) node[below] {\small $x_1$};
    \draw[blue] (2, 0.1) -- +(0, -0.2) node[below] {\small $x_2$};
    \draw[blue] (3, 0.1) -- +(0, -0.2) node[below] {\small $x_3$};
    \draw[red]  (5, 0.1) -- +(0, -0.2) node[below] {\small $x_4$};
    \draw[blue] (7, 0.1) -- +(0, -0.2) node[below] {\small $x_5$};

    \begin{scope}[blue,very thick]
      \draw [domain = -1 : 2] \uformula;
      \draw [domain = 2 : 3] \uformula;
      \draw [domain = 3 : 7] \uformula;
      \draw [domain = 7 : 9] \uformula;
    \end{scope}
    \begin{scope}[red,very thick]
      \draw [domain = -1 : 1] \vformula;
      \draw [domain = 1 : 5] \vformula;
      \draw [domain = 5 : 9] \vformula;
    \end{scope}
    \draw (1.5,2.8) node[right] {$u(x,t) = m_2 \, e^{-\abs{x-x_2}} + m_3 \, e^{-\abs{x-x_3}} + m_5 \, e^{-\abs{x-x_5}}$};
    \draw (5.2,1.8) node[right] {$v(x,t) = n_1 \, e^{-\abs{x-x_1}} + n_4 \, e^{-\abs{x-x_4}}$};
  \end{tikzpicture}
  \caption{A \textbf{non-overlapping peakon}
    configuration~\eqref{eq:GXpeakons},
    meaning that the peakons in $u(x,t)$ and $v(x,t)$ occupy different sites:
    since $n_1$ is nonzero, $m_1$ must be zero,
    and since $m_2$ is nonzero, $n_2$ must be zero, and so on.}
  \label{fig:nonoverlapping-peakons}
\end{figure}

Our primary subject in this article is a particular class of solutions of
the Geng--Xue equation~\eqref{eq:GX} known as \emph{peakons} (peaked solitons) --
weak solutions formed by nonlinear superposition of $e^{-\abs{x}}$-shaped waves:
\begin{equation}
  \label{eq:GXpeakons}
  \begin{split}
    u(x,t) &= \sum_{k=1}^N m_k(t) \, e^{-\abs{x - x_k(t)}}
    , \\
    v(x,t) &= \sum_{k=1}^N n_k(t) \, e^{-\abs{x - x_k(t)}}
    .
  \end{split}
\end{equation}
We will always label the variables in increasing order,
\begin{equation*}
  x_1 < x_2 < \dots < x_N
  .
\end{equation*}
As will be explained in \autoref{sec:weak-solutions-GX} below,
the ansatz \eqref{eq:GXpeakons} satisfies the PDE \eqref{eq:GX} in
a certain distributional sense if and only if
the functions $x_k(t)$, $m_k(t)$ and $n_k(t)$ satisfy the ODEs
\begin{equation}
  \label{eq:GX-peakon-ode}
  \begin{aligned}
    \dot x_k &= u(x_k) \, v(x_k)
    ,\\
    \dot m_k &= m_k \bigl(u(x_k) \, v_x(x_k) - 2 u_x(x_k) v(x_k) \bigr)
    ,\\
    \dot n_k &= n_k \bigl(u_x(x_k) \, v(x_k) - 2 u(x_k) v_x(x_k) \bigr)
    ,
  \end{aligned}
\end{equation}
for $k = 1,2,\dots,N$,
with the additional constraint that the two types of peakons (those belonging to~$u$ and to~$v$)
occupy \emph{different} sites~$x_k$;
in other words, for each~$k$, either
peakon number~$k$ belongs to~$u$,
\begin{equation*}
  m_k \neq 0
  \quad\text{and}\quad
  n_k = 0
  ,
\end{equation*}
or else it belongs to~$v$,
\begin{equation*}
  m_k = 0
  \quad\text{and}\quad
  n_k \neq 0
  .
\end{equation*}
(Note that these conditions are preserved by the ODEs.
Note also that we may assume that $m_k=n_k=0$ doesn't happen,
since in that case we can just discard the $k$th terms in the ansatz.)
See \autoref{fig:nonoverlapping-peakons}.

The ODEs~\eqref{eq:GX-peakon-ode}
have been written using a convenient shorthand notation,
where
$u(x_k)$, $u_x(x_k)$, $v(x_k)$ and $v_x(x_k)$
are nothing but abbreviations defined as follows:
\begin{equation}
  \label{eq:uv-shorthand}
  \begin{aligned}
    u(x_k) &:= \sum_{i=1}^N m_i \, e^{-\abs{x_k-x_i}}
    ,\\
    u_x(x_k) &:= -\sum_{i=1}^N m_i \, \sgn(x_k-x_i) \, e^{-\abs{x_k-x_i}}
    ,\\
    v(x_k) &:= \sum_{i=1}^N n_i \, e^{-\abs{x_k-x_i}}
    ,\\
    v_x(x_k) &:= -\sum_{i=1}^N n_i \, \sgn(x_k-x_i) \, e^{-\abs{x_k-x_i}}
    ,
  \end{aligned}
\end{equation}
where the convention $\sgn 0 = 0$ is understood.
The formulas for $u(x_k)$ and~$u_x(x_k)$ in \eqref{eq:uv-shorthand}
result from formally substituting $x=x_k$
into the ansatz $u(x) = \sum_{i=1}^N m_i \, e^{-\abs{x-x_i}}$
and its derivative
$u_x(x) = -\sum_{i=1}^N m_i \, \sgn(x-x_i) \, e^{-\abs{x-x_i}}$,
respectively.
Note, however, that $u_x(x_k)$ does not exist in the ordinary sense
if $m_k \neq 0$,
so what we have denoted by $u_x(x_k)$ is actually the average of the
the left and right limits of~$u_x$ at $x=x_k$:
\begin{equation*}
  u_x(x_k)
  = \avg{u_x}(x_k)
  = \frac12 \Biggl( \lim_{x \to x_k^-} u_x(x) + \lim_{x \to x_k^+} u_x(x) \Biggr)
  .
\end{equation*}

The Geng--Xue peakon ODEs~\eqref{eq:GX-peakon-ode}
constitute a Lax integrable system, as we will explain in
detail later in this article,
and the general solution can be written down explicitly in terms
of elementary functions.

\begin{remark}
  Our detailed study
  of the Geng--Xue equation is partly motivated by our interest in
  understanding the structure of peakon equations from some unifying
  principle. One of the most delicate issues with Lax integrable
  peakon equations is that the Lax equations in the peakon sector are
  distributional equations and as such require special care when
  dealing with nonlinear operations involving distributions with
  singular support. This makes peakons an intriguing ``borderline''
  case of Lax integrability. In this article we take a conservative
  approach: the Lax pair with $m=u-u_{xx}$ and~$n=v-v_{xx}$ being
  distributions with non-overlapping singular supports is a
  well-defined distributional pair requiring only the standard
  operation of multiplying measures by continuous functions.
  This, along with the distributional compatibility condition, dictates a
  unique way of defining a distributional Geng--Xue equation.
\end{remark}

\subsection{Shockpeakon solutions of the Geng--Xue equation}
\label{sec:intro-shockpeakons}

\begin{figure}
  \centering
  \begin{tikzpicture}
    \draw[->] (-1, 0) -- (9.5, 0) node[below] {$x$};
    \draw[->] (0, -0.5) -- (0, 2.5);

    % \def\uformula{plot (\noexpand\x,{(2-0.2*sign(\noexpand\x-2))*exp(-abs(\noexpand\x-2))+1*exp(-abs(\noexpand\x-3))+(0.5-1.2*sign(\noexpand\x-7))*exp(-abs(\noexpand\x-7))})}
    % \def\vformula{plot (\noexpand\x,{1.3*exp(-abs(\noexpand\x-1))+(1.7-0.5*sign(\noexpand\x-5))*exp(-abs(\noexpand\x-5))})}

    % Newer versions of pgf/TikZ contain a builtin sign function
    % (used in the commented-out lines above),
    % but not the version currently used by arXiv, so we have
    % to make our own...
    \tikzset{declare function={sgn(\x) = (and(\x<0, 1) * -1) + (and(\x>0, 1) * 1);}}
    \def\uformula{plot (\noexpand\x,{(2-0.2*sgn(\noexpand\x-2))*exp(-abs(\noexpand\x-2))+1*exp(-abs(\noexpand\x-3))+(0.5-1.2*sgn(\noexpand\x-7))*exp(-abs(\noexpand\x-7))})}
    \def\vformula{plot (\noexpand\x,{1.3*exp(-abs(\noexpand\x-1))+(1.7-0.5*sgn(\noexpand\x-5))*exp(-abs(\noexpand\x-5))})}

    \draw[red]  (1, 0.1) -- +(0, -0.2) node[below] {\small $x_1$};
    \draw[blue] (2, 0.1) -- +(0, -0.2) node[below] {\small $x_2$};
    \draw[blue] (3, 0.1) -- +(0, -0.2) node[below] {\small $x_3$};
    \draw[red]  (5, 0.1) -- +(0, -0.2) node[below] {\small $x_4$};
    \draw[blue] (7, 0.1) -- +(0, -0.2) node[below] {\small $x_5$};

    \begin{scope}[blue,very thick]
      \draw [domain = -1 : 1.999] \uformula;
      \draw [domain = 2.001 : 3] \uformula;
      \draw [domain = 3 : 6.999] \uformula;
      \draw [domain = 7.001 : 9] \uformula;
    \end{scope}
    \begin{scope}[red,very thick]
      \draw [domain = -1 : 1] \vformula;
      \draw [domain = 1 : 4.999] \vformula;
      \draw [domain = 5.001 : 9] \vformula;
    \end{scope}
    \draw (1.5,3.0) node[right] {$u(x,t)$};
    \draw (4.5,2.7) node[right] {$v(x,t)$};
  \end{tikzpicture}
  \caption{A \textbf{non-overlapping shockpeakon}
    configuration~\eqref{eq:GXshockpeakons},
    meaning that the shockpeakons (or ordinary peakons as a special case)
    in $u(x,t)$ and $v(x,t)$ occupy different sites.
    Compared to \autoref{fig:nonoverlapping-peakons},
    shocks $s_2>0$ and $s_5>0$ have been added to~$u(x,t)$
    (but $s_3=0$),
    and a shock $r_4>0$ has been added to~$v(x,t)$
    (but $r_1=0$).}
  \label{fig:nonoverlapping-shockpeakons}
\end{figure}
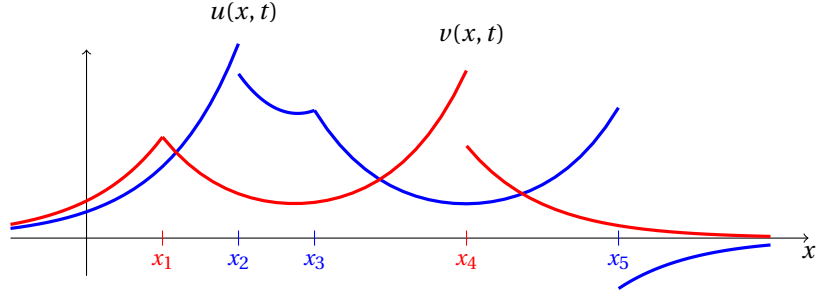

Remarkably, the Geng--Xue equation also admits a weaker type
of solution, \emph{shockpeakons},
given by the more general ansatz
\begin{equation}
  \label{eq:GXshockpeakons}
  \begin{split}
    u(x,t) &= \sum_{k=1}^N \Bigl( m_k(t) - s_k(t) \, \sgn\bigl(x-x_k(t)\bigr) \Bigr) \, e^{-\abs{x - x_k(t)}}
    , \\
    v(x,t) &= \sum_{k=1}^N \Bigl( n_k(t) - r_k(t) \, \sgn\bigl(x-x_k(t)\bigr) \Bigr) \, e^{-\abs{x - x_k(t)}}
    .
  \end{split}
\end{equation}
See \autoref{fig:nonoverlapping-shockpeakons}.
If $s_k(t) \neq 0$, then the function $x \mapsto u(x,t)$ has a
jump of size $-2s_k(t)$ the point $x=x_k(t)$,
and similarly for $x \mapsto v(x,t)$ if $r_k(t) \neq 0$,
so shockpeakon solutions are only piecewise continuous.
The ansatz‌~\eqref{eq:GXshockpeakons}
satisfies the Geng--Xue equation (again in a distributional
sense explained in \autoref{sec:weak-solutions-GX})
if and only if
the functions $x_k(t)$, $m_k(t)$, $n_k(t)$,
$s_k(t)$ and $r_k(t)$
satisfy the ODEs
\begin{equation}
  \label{eq:GX-shockpeakon-ode}
  \begin{aligned}
    \dot x_k &= u(x_k) \, v(x_k)
    ,\\
    \dot m_k &
    = m_k \bigl( u(x_k) \, v_x(x_k) - 2 u_x(x_k) v(x_k) \bigr)
    + s_k \bigl( u(x_k) \, v(x_k) + u_x(x_k) v_x(x_k) \bigr)
    ,\\
    \dot s_k &= s_k \bigl( 2 u(x_k) \, v_x(x_k) - u_x(x_k) v(x_k) \bigr)
    ,\\
    \dot n_k &
    = n_k \bigl( u_x(x_k) \, v(x_k) - 2 u(x_k) v_x(x_k) \bigr)
    + r_k \bigl( u(x_k) \, v(x_k) + u_x(x_k) v_x(x_k) \bigr)
    ,\\
    \dot r_k &= r_k \bigl( 2 u_x(x_k) \, v(x_k) - u(x_k) v_x(x_k) \bigr)
    ,
  \end{aligned}
\end{equation}
for $k = 1,2,\dots,N$,
with the non-overlapping constraint that, for each~$k$, either
\begin{equation*}
  (m_k \neq 0
  \quad\text{or}\quad
  s_k \neq 0)
  \quad\text{and}\quad
  n_k = r_k = 0
  ,
\end{equation*}
or
\begin{equation*}
  m_k = s_k = 0
  \quad\text{and}\quad
  (n_k \neq 0
  \quad\text{or}\quad
  r_k \neq 0)
  .
\end{equation*}
(See \autoref{thm:GX-shockpeakons}.)
In~\eqref{eq:GX-shockpeakon-ode}, the previous shorthand
notation~\eqref{eq:uv-shorthand}
has been extended as follows:
\begin{equation}
  \label{eq:uv-shorthand-shock}
  \begin{aligned}
    u(x_k) &:= \sum_{i=1}^N \bigl( m_i - s_i \, \sgn(x_k - x_i) \bigr) \, e^{-\abs{x_k-x_i}}
    ,\\
    u_x(x_k) &:= \sum_{i=1}^N \bigl( s_i - m_i \, \sgn(x_k - x_i) \bigr) \, e^{-\abs{x_k-x_i}}
    ,\\
    v(x_k) &:= \sum_{i=1}^N \bigl( n_i - r_i \, \sgn(x_k - x_i) \bigr) \, e^{-\abs{x_k-x_i}}
    ,\\
    v_x(x_k) &:= \sum_{i=1}^N \bigl( r_i - n_i \, \sgn(x_k - x_i) \bigr) \, e^{-\abs{x_k-x_i}}
  .
  \end{aligned}
\end{equation}
We emphasize that these formulas are nothing but abbreviations
which are convenient when writing down the shockpeakon ODEs~\eqref{eq:GX-shockpeakon-ode};
when interpreting the solutions distributionally,
the precise value assigned to $u$ at a jump discontinuity is irrelevant,
and the derivative $u_x$ doesn't exist at such a point.
However, we do have
\begin{equation}
  \label{eq:u-ux-are-averages}
  u(x_k)
  = \frac12 \left( \lim_{x \to x_k^-} u(x) + \lim_{x \to x_k^+} u(x) \right)
  ,\quad
  u_x(x_k)
  = \frac12 \left( \lim_{x \to x_k^-} u_x(x) + \lim_{x \to x_k^+} u_x(x) \right)
  ,
\end{equation}
and similarly for $v(x_k)$ and~$v_x(x_k)$.

\begin{remark}
  Of course,
  if $s_k = r_k = 0$ for all~$k$,
  then the shockpeakon ansatz~\eqref{eq:GXshockpeakons},
  the ODEs~\eqref{eq:GX-shockpeakon-ode}
  and the shorthand notation~\eqref{eq:uv-shorthand-shock}
  reduce to their ordinary peakon counterparts~\eqref{eq:GXpeakons},
  \eqref{eq:GX-peakon-ode}, \eqref{eq:uv-shorthand}, respectively.
\end{remark}

\begin{remark}
  \label{rem:entropy-condition}
  Shockpeakons were first introduced
  for the Degasperis--Procesi equation
  by Lundmark in~\cite{lundmark:shockpeakons},
  where it was found that it is necessary to have $s_i \ge 0$
  in order to obtain an \emph{entropy solution} as defined by
  Coclite and Karlsen
  \cite{coclite-karlsen:DPwellposedness,
    coclite-karlsen:DPuniqueness};
  i.e., the jump at each shock must go \emph{downwards},
  from high on the left to low on the right.
  This condition will automatically be satisfied whenever a shockpeakon
  forms at a peakon--antipeakon collision in the Degasperis--Procesi equation.

  Most likely, it will turn out natural to impose
  the corresponding restriction
  $s_i \ge 0$ and $r_i \ge 0$
  on Geng--Xue shockpeakons as well,
  although we will not attempt here to define entropy solutions.
\end{remark}

\begin{example}
  \label{ex:shockpeakons-1+1}
  In the $1+1$ case, with one shockpeakon in~$u(x,t)$ at $x=x_1(t)$,
  and another in~$v(x,t)$ at $x=x_2(t)$, where $x_1<x_2$,
  the governing ODEs~\eqref{eq:GX-shockpeakon-ode} take the
  following form:
  \begin{equation}
    \label{eq:GX-shockpeakon-ode-K1}
    \begin{aligned}
      \dot x_1 &= m_1 (n_2+r_2) E_{12}
      ,\\
      \dot x_2 &= (m_1-s_1) n_2 E_{12}
      ,\\
      \dot m_1 &= (m_1^2 - m_1 s_1 + s_1^2) (n_2+r_2) E_{12}
      ,\\
      \dot n_2 &= -(m_1-s_1) (n_2^2 + n_2 r_2 + r_2^2) E_{12}
      ,\\
      \dot s_1 &= s_1 (2m_1-s_1) (n_2+r_2) E_{12}
      ,\\
      \dot r_2 &= -r_2 (m_1-s_1) (2n_2+r_2) E_{12}
      ,
    \end{aligned}
  \end{equation}
  where $E_{12} = e^{-\abs{x_1-x_2}} = e^{x_1-x_2}$.
  We have managed to integrate these equations explicitly
  (see \autoref{sec:dynamics-shock-1}),
  but we don't know how to solve
  the shockpeakon ODEs~\eqref{eq:GX-shockpeakon-ode}
  for $K \ge 2$.
  In fact, even for $K=1$ it is not currently clear to us
  whether one can interpret the solution of~\eqref{eq:GX-shockpeakon-ode-K1}
  in terms of Lax integrability.
\end{example}

\subsection{Interlacing peakon configurations}
\label{sec:intro-interlacing}

\begin{figure}
  \centering
  \begin{tikzpicture}
    \draw[->] (-1, 0) -- (9.5, 0) node[below] {$x$};
    \draw[->] (0, -0.5) -- (0, 2.5);

    \def\uformula{plot (\noexpand\x,{2.5*exp(-abs(\noexpand\x-1))+1*exp(-abs(\noexpand\x-3))+0.5*exp(-abs(\noexpand\x-7))})}
    \def\vformula{plot (\noexpand\x,{1.7*exp(-abs(\noexpand\x-2))+1.3*exp(-abs(\noexpand\x-5))+1.2*exp(-abs(\noexpand\x-8))})}

    \draw[blue] (1, 0.1) -- +(0, -0.2) node[below] {\small $x_1$};
    \draw[red]  (2, 0.1) -- +(0, -0.2) node[below] {\small $x_2$};
    \draw[blue] (3, 0.1) -- +(0, -0.2) node[below] {\small $x_3$};
    \draw[red]  (5, 0.1) -- +(0, -0.2) node[below] {\small $x_4$};
    \draw[blue] (7, 0.1) -- +(0, -0.2) node[below] {\small $x_5$};
    \draw[red]  (8, 0.1) -- +(0, -0.2) node[below] {\small $x_6$};

    \begin{scope}[blue,very thick]
      \draw [domain = -1 : 1] \uformula;
      \draw [domain = 1 : 3] \uformula;
      \draw [domain = 3 : 7] \uformula;
      \draw [domain = 7 : 9] \uformula;
    \end{scope}
    \begin{scope}[red,very thick]
      \draw [domain = -1 : 2] \vformula;
      \draw [domain = 2 : 5] \vformula;
      \draw [domain = 5 : 8] \vformula;
      \draw [domain = 8 : 9] \vformula;
    \end{scope}
    \draw (0.3,3) node[right] {$u(x,t) = m_1 \, e^{-\abs{x-x_1}} + m_3 \, e^{-\abs{x-x_3}} + m_5 \, e^{-\abs{x-x_5}}$};
    \draw (3.8,1.8) node[right] {$v(x,t) = n_2 \, e^{-\abs{x-x_2}} + n_4 \, e^{-\abs{x-x_4}} + n_6 \, e^{-\abs{x-x_6}}$};
  \end{tikzpicture}
  \caption{An \textbf{interlacing peakon}
    configuration~\eqref{eq:GX-interlacing-peakons}
    with~$K=3$, meaning that there are three peakons in each component,
    with the first peakon belonging to~$u(x,t)$,
    the second to~$v(x,t)$,
    and so on, alternatingly.}
  \label{fig:interlacing-peakons}
\end{figure}
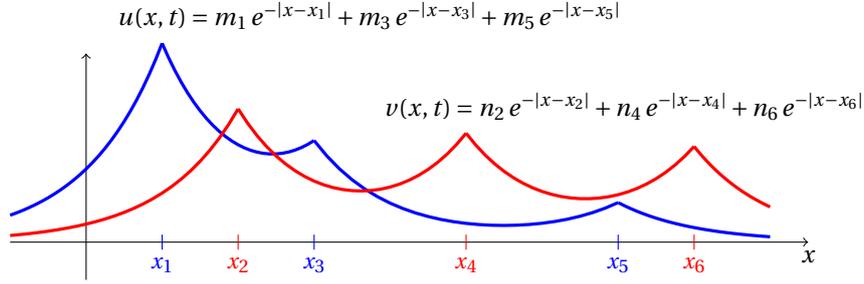

In this article our main focus will be peakons rather than shockpeakons,
and more specifically the particular case of peakon solutions
which are \emph{interlacing} in the sense that there are $N=2K$ sites
\begin{equation*}
  x_1 <  x_2 < \dots < x_{2K},
\end{equation*}
with $u$ and~$v$ containing $K$~peakons each,
located at the odd-numbered sites~$x_{2a-1}$ in the case of~$u$,
and at the even-numbered sites~$x_{2a}$ in the case of~$v$;
see \autoref{fig:interlacing-peakons}.
Moreover, unless stated otherwise,
we will assume that the nonzero amplitudes are
\emph{positive} (i.e., that there are no \emph{antipeakons} with negative
amplitude).
In other words, we assume
\begin{equation}
  \label{eq:interlacing-restriction}
  \begin{aligned}
    m_{2a-1} &> 0
    ,
    &
    m_{2a} & = 0
    ,
    \\
    n_{2a-1} &= 0
    ,
    &
    n_{2a} &> 0
    ,
  \end{aligned}
\end{equation}
for $1 \le a \le K$, so that
\begin{equation}
  \label{eq:GX-interlacing-peakons}
  \begin{split}
    u(x,t) &= \sum_{a=1}^K m_{2a-1}(t) \, e^{-\abs{x - x_{2a-1}(t)}}
    , \\
    v(x,t) &= \sum_{a=1}^K n_{2a}(t) \, e^{-\abs{x - x_{2a}(t)}}
    .
  \end{split}
\end{equation}
Drawing heavily on the groundwork from our previous article \cite{lundmark-szmigielski:GX-inverse-problem},
we will use inverse spectral methods to derive explicit formulas for these
interlacing $(K+K)$-peakon solutions, and then explore their dynamical properties.

\begin{remark}
  Explicit solution formulas for the general non-interlacing case,
  where the number of peakons in $u$ and~$v$ need not be equal,
  and where they may appear in arbitrary order,
  can be obtained by
  starting from a larger interlacing case and performing certain limiting
  procedures in order to drive selected amplitudes to zero.
  However, the details are rather technical,
  and will be saved for a separate article.
\end{remark}

\begin{example}
  The case $K=1$ is exceptional, and also very simple.
  There is one peakon in~$u$ and one in~$v$:
  \begin{equation*}
    u(x,t) = m_1(t) \, e^{-\abs{x-x_1(t)}}
    ,\qquad
    v(x,t) = n_2(t) \, e^{-\abs{x-x_2(t)}}
    .
  \end{equation*}
  Details will be given in \autoref{sec:dynamics-1}.
  With initial data $m_1(0) > 0$, $n_2(0) > 0$ and $x_1(0) < x_2(0)$,
  the solution is
  \begin{equation}
    \label{eq:solution-K1-preview}
    \begin{aligned}
      x_1(t) &= x_1(0) + ct
      ,\\
      x_2(t) &= x_2(0) + ct
      ,\\
      m_1(t) &= m_1(0) \, e^{ct}
      ,\\
      n_2(t) &= n_2(0) \, e^{-ct}
      ,
    \end{aligned}
  \end{equation}
  where $c = m_1(0) \, n_2(0) \, e^{x_1(0)-x_2(0)} > 0$.
\end{example}

\begin{figure}
  \centering
  \includegraphics{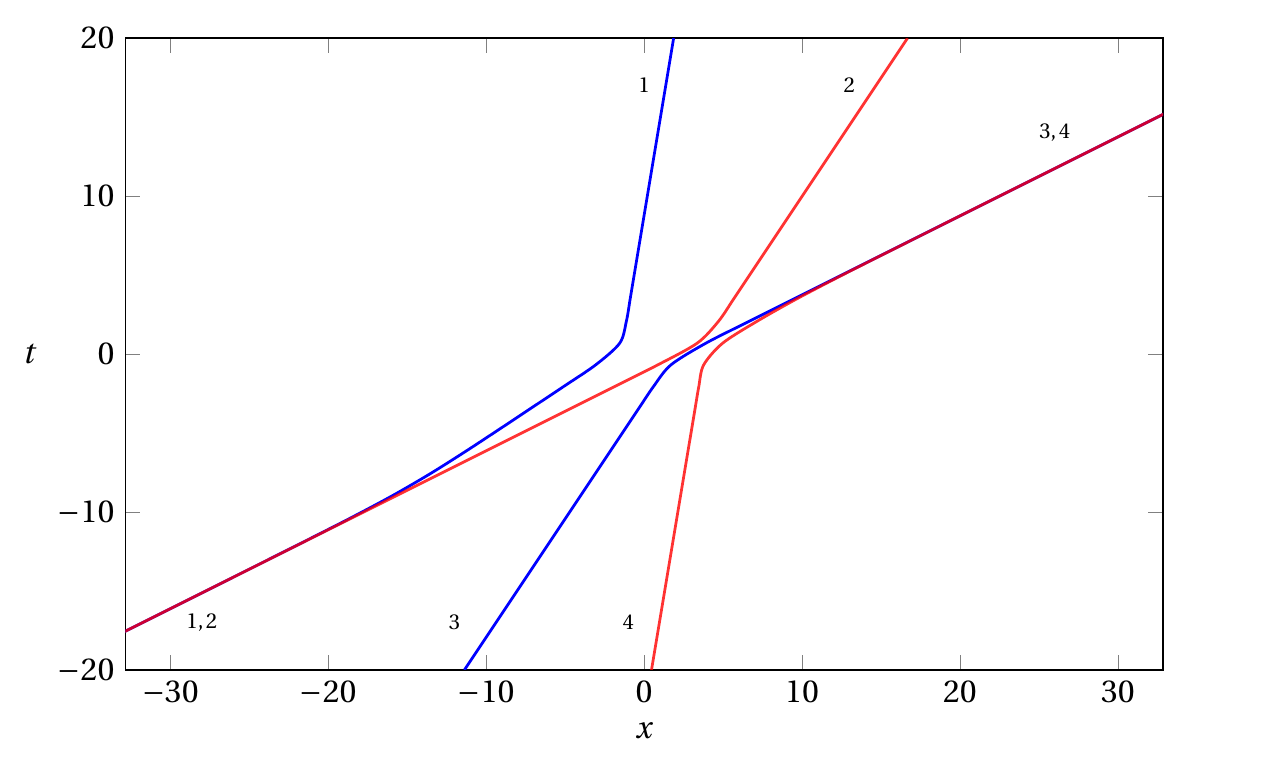}
  \caption{Spacetime plot of the peakon trajectories
    $x = x_1(t)$,
    $x = x_2(t)$,
    $x = x_3(t)$
    and
    $x = x_4(t)$
    for a $2+2$ interlacing pure peakon solution of the Geng--Xue equation:
    $u(x,t) = m_1 \, e^{-\abs{x-x_1}} + m_3 \, e^{-\abs{x-x_3}}$
    and
    $v(x,t) = n_2 \, e^{-\abs{x-x_2}} + n_4 \, e^{-\abs{x-x_4}}$,
    with $x_1 < x_2 < x_3 < x_4$
    and with
    $m_1$, $n_2$, $m_3$, $n_4$
    positive.
    The parameters used in the solution
    formulas~\eqref{eq:solution-K2}
    are given in \autoref{ex:GX-interlacing-2+2},
    together with a description of the noteworthy
    features in this picture.
    The blue curves $x = x_1(t)$ and $x = x_3(t)$ refer to the
    peakons in $u(x,t)$,
    while the red curves $x = x_2(t)$ and $x = x_4(t)$ refer to the
    peakons in $v(x,t)$.
  }
  \label{fig:GX-interlacing-2+2-X}
\end{figure}

\begin{figure}
  \centering
  \includegraphics{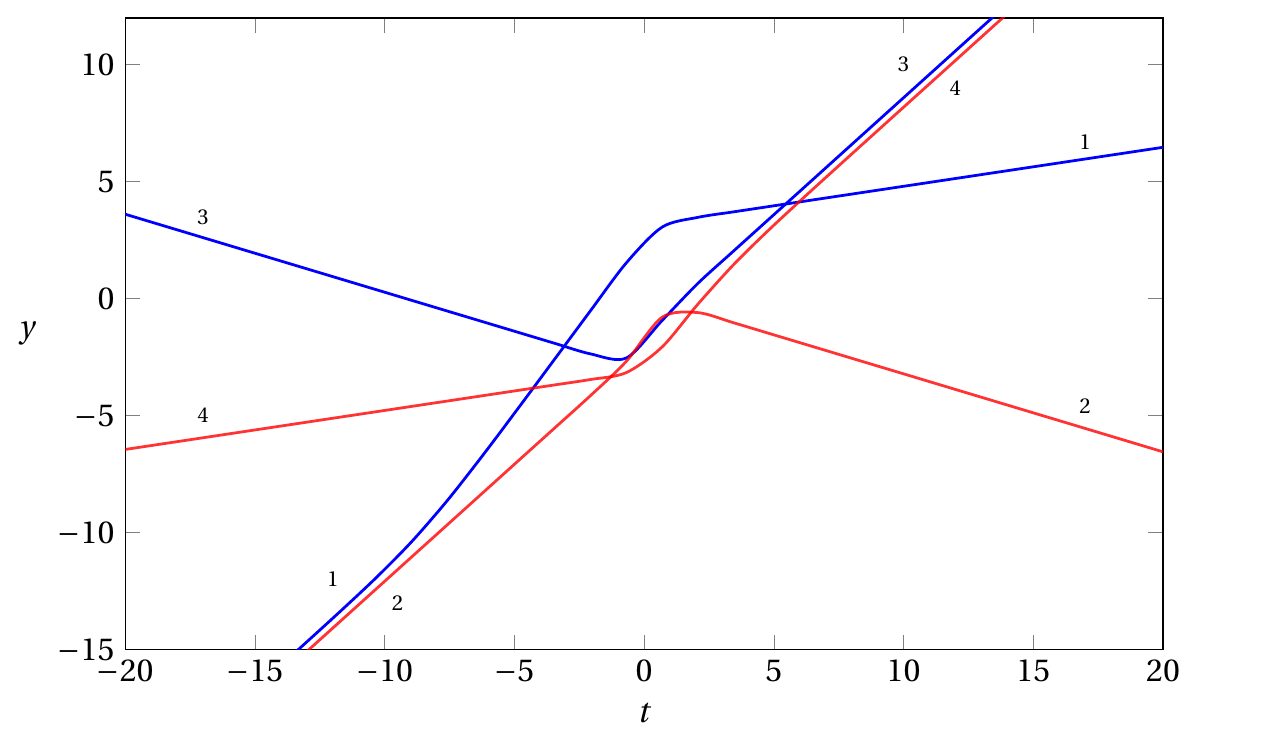}
  \caption{Plot of the curves
    $y = \ln m_1(t)$,
    $y = -\ln n_2(t)$,
    $y = \ln m_3(t)$
    and
    $y = -\ln n_4(t)$
    for the same solution as in \autoref{fig:GX-interlacing-2+2-X}.
    See \autoref{ex:GX-interlacing-2+2} for further explanation.
    The blue curves $x = \ln m_1(t)$ and $x = \ln m_3(t)$ refer to the
    peakons in $u(x,t)$,
    and the red curves $x = -\ln n_2(t)$ and $x = -\ln n_4(t)$ refer to the
    peakons in $v(x,t)$.
  }
  \label{fig:GX-interlacing-2+2-M}
\end{figure}

\begin{figure}
  \centering
  \includegraphics{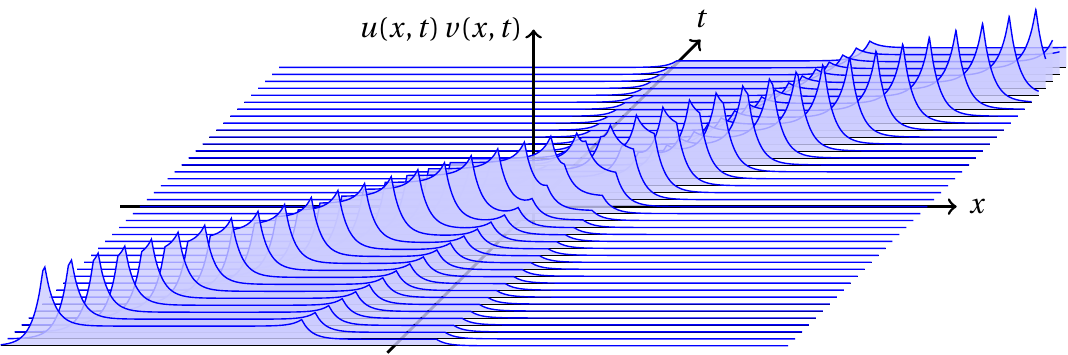}
  \caption{
    Because of the exponential growth and decay of the amplitudes
    $m_1$, $n_2$, $m_3$, $n_4$,
    it is difficult to make meaningful plots of the
    individual components $u(x,t)$ and $v(x,t)$,
    but their product $u(x,t) \, v(x,t)$ is well-behaved,
    and this product is graphed here for the same solution as in
    \autoref{fig:GX-interlacing-2+2-X}
    and
    \autoref{fig:GX-interlacing-2+2-M}.
    Note that it is the product $uv$
    which determines the velocity of the peakons,
    according to the governing ODE
    $\dot x_k = u(x_k) \, v(x_k)$.
    The domain shown is $-20 \le x \le 20$, $-10 \le t \le 10$,
    and the function is sampled at time values $1/4$ units apart.
    The projection is orthogonal, and
    the vertical scale is exaggerated by a factor of~$2$.
  }
  \label{fig:GX-interlacing-2+2-3Dplot-uv}
\end{figure}

\begin{example}
  \label{ex:GX-interlacing-2+2}
  When $K=2$, the interlacing peakon solutions take the form
  \begin{equation*}
    \begin{aligned}
      u(x,t) &
      = m_1(t) \, e^{-\abs{x-x_1(t)}} + m_3(t) \, e^{-\abs{x-x_3(t)}}
      ,\\
      v(x,t) &
      = n_2(t) \, e^{-\abs{x-x_2(t)}} + n_4(t) \, e^{-\abs{x-x_4(t)}}
      ,
    \end{aligned}
  \end{equation*}
  where it is understood that $x_1(t) < x_2(t) < x_3(t) < x_4(t)$
  so that the solution really is interlacing.
  Such solutions are studied in \autoref{sec:dynamics-2},
  mainly for pedagogical reasons.
  (All the results for $2+2$ interlacing peakon solutions
  in \autoref{sec:dynamics-2}
  are special cases of the statements for
  $K+K$ interlacing peakon solutions
  in \autoref{sec:dynamics-K},
  but the proofs for arbitrary~$K$ require
  a fair amount of additional notation.)

  The ODEs governing the dynamics of the eight variables
  \begin{equation*}
    x_1(t),\, x_2(t),\, x_3(t),\, x_4(t),\, m_1(t),\, n_2(t),\, m_3(t),\, n_4(t)
  \end{equation*}
  are given in equation~\eqref{eq:GX-peakon-ode-K2},
  and the general solution of these ODEs
  is written out in complete detail in equations
  \eqref{eq:solution-K2-position}
  and~\eqref{eq:solution-K2-amplitude}.
  These solution formulas contain five constant parameters
  \begin{equation*}
    \lambda_2 > \lambda_1 > 0
    ,\qquad
    \mu_1 > 0
    ,\qquad
    b_{\infty} > 0
    ,\qquad
    b_{\infty}^* > 0
    ,
  \end{equation*}
  and three time-dependent quantities
  \begin{equation*}
    a_1(t) = a_1(0) \, e^{t/\lambda_1}
    ,\qquad
    a_2(t) = a_2(0) \, e^{t/\lambda_2}
    ,\qquad
    b_1(t) = b_1(0) \, e^{t/\mu_1}
  \end{equation*}
  determined by their initial values
  \begin{equation*}
    a_1(0) > 0
    ,\qquad
    a_2(0) > 0
    ,\qquad
    b_1(0) > 0
    .
  \end{equation*}

  \autoref{fig:GX-interlacing-2+2-X}
  shows a plot of the peakon trajectories
  \begin{equation*}
    x = x_1(t)
    ,\quad
    x = x_2(t)
    ,\quad
    x = x_3(t)
    ,\quad
    x = x_4(t)
  \end{equation*}
  given by the formulas \eqref{eq:solution-K2-position},
  for the parameter values
  \begin{equation}
    \label{eq:interlacing-2+2-spectraldata}
    \begin{gathered}
      \lambda_1 = \frac13
      ,\quad
      \lambda_2 = 3
      ,\quad
      \mu_1 = 1
      ,\\
      a_1(0) = a_2(0) = 1
      ,\quad
      b_1(0) = 100
      ,\quad
      b_{\infty} = 1000
      ,\quad
      b_{\infty}^* = 100
      .
    \end{gathered}
  \end{equation}
  It is apparent in the picture, and will be proved in
  \autoref{sec:dynamics-2},
  that the trajectories approach certain straight lines
  asymptotically, as $t \to \pm\infty$.
  More precisely, there are three distinct asymptotic velocities
  \begin{equation*}
    \begin{aligned}
    c_1 &
    = \frac{1}{2} \left( \frac{1}{\lambda_1} + \frac{1}{\mu_1} \right)
    = \frac{1}{2} \left( \frac{1}{1/3} + \frac{1}{1} \right)
    = 2
    ,\\
    c_2 &
    = \frac{1}{2} \left( \frac{1}{\lambda_2} + \frac{1}{\mu_1} \right)
    = \frac{1}{2} \left( \frac{1}{3} + \frac{1}{1} \right)
    = \frac{2}{3}
    ,\\
    c_3 &
    = \frac{1}{2} \left( \frac{1}{\lambda_2} \right)
    = \frac{1}{2} \left( \frac{1}{3} \right)
    = \frac{1}{6}
    .
    \end{aligned}
  \end{equation*}
  As $t \to -\infty$,
  the two leftmost curves $x=x_1(t)$ and $x_2(t)$ both approach the line
  \begin{equation*}
    x = c_1 t
    + \frac12 \ln \frac{2 \, a_1(0) \, b_1(0)}{\lambda_1+\mu_1}
    + \frac12 \ln \frac{\bigl( \lambda_1-\lambda_2 \bigr)^2}{\lambda_2 (\lambda_2+\mu_1)}
    = 2t + \frac12 \ln 150 + \frac12 \ln\frac{16}{27}
    ,
  \end{equation*}
  the curve $x=x_3(t)$ approaches the line
  \begin{equation*}
    x = c_2 t
    + \frac12 \ln \frac{2 \, a_2(0) \, b_1(0)}{\lambda_2 + \mu_1}
    = \frac{2t}{3} + \frac12 \ln 50
    ,
  \end{equation*}
  and the curve $x=x_4(t)$ approaches the line
  \begin{equation*}
    x = c_3 t
    + \frac12 \ln \bigl( 2 \, a_2(0) \, b_{\infty} \bigr)
    = \frac{t}{6} + \frac12 \ln 2000
    .
  \end{equation*}
  (These formulas are taken
  from~\eqref{eq:asymptotics-K2-position-minus-infty},
  which is the special case $K=2$
  of the general formulas for the $K+K$ case given in
  \autoref{thm:asymptotics-generalK-positions}.)
  As $t \to +\infty$,
  it is instead the two rightmost curves $x=x_3(t)$ and $x_4(t)$
  that pair up; they both approach the line
  \begin{equation*}
    x = c_1 t
    + \frac12 \ln \frac{2 \, a_1(0) \, b_1(0)}{\lambda_1 + \mu_1}
    = 2t + \frac12 \ln 150
    ,
  \end{equation*}
  while the curve $x=x_2(t)$ approaches
  \begin{equation*}
    x = c_2 t
    + \frac12 \ln \frac{2 \, a_2(0) \, b_1(0)}{\lambda_2+\mu_1}
    + \frac12 \ln \frac{\bigl( \lambda_1-\lambda_2 \bigr)^2}{\lambda_1 (\lambda_1+\mu_1)}
    = \frac{2t}{3} + \frac12 \ln 50 + \frac12 \ln 16
    ,
  \end{equation*}
  and the curve $x=x_1(t)$ approaches
  \begin{equation*}
    x = c_3 t
    + \frac12 \ln \frac{\mu_1 \, a_2(0)}{\lambda_1 \lambda_2 \, b^*_{\infty}}
    + \frac12 \ln \frac{\bigl( \lambda_1-\lambda_2 \bigr)^2}{\lambda_1 (\lambda_2 + \mu_1)}
    = \frac{t}{6} + \frac12 \ln \frac{1}{100} + \frac12 \ln\frac{16}{3}
    .
  \end{equation*}
  (This is proved in~\eqref{eq:asymptotics-K2-position-plus-infty}
  and more generally in \autoref{thm:asymptotics-generalK-positions}.)
  A comparison of the two lines of the form $x = c_1 t + \text{const.}$
  shows that the second (outgoing) line
  is shifted relative to the first (incoming) one
  by the amount
  \begin{equation*}
    -\frac12 \ln \frac{\bigl( \lambda_1-\lambda_2 \bigr)^2}{\lambda_2 (\lambda_2+\mu_1)}
    = -\frac12 \ln\frac{16}{27}
  \end{equation*}
  in the $x$ direction,
  and the corresponding shifts for
  the other pairs of incoming and outgoing asymptotic lines
  are also easily computed (see \autoref{sec:phaseshift-K2-positions}
  and \autoref{cor:phaseshift-positions}).

  As for the amplitudes 
  given by the formulas~\eqref{eq:solution-K2-amplitude},
  \autoref{fig:GX-interlacing-2+2-M} shows logaritmic plots
  \begin{equation*}
    y = \ln m_1(t)
    ,\quad
    y = -\ln n_2(t)
    ,\quad
    y = \ln m_3(t)
    ,\quad
    y = -\ln n_4(t)
    ,
  \end{equation*}
  again with the same parameters~\eqref{eq:interlacing-2+2-spectraldata}.
  The reason for plotting the logarithms is that the amplitudes themselves
  grow or decay exponentially as $t \to \pm\infty$,
  so that the logarithmic plots will asymptotically approach straight lines,
  and the purpose of the extra minus signs on the even-numbered curves
  is to highlight certain relations between the slopes of these lines.
  More precisely, there are three distinct asymptotic slopes
  \begin{equation*}
    \begin{aligned}
    d_1 &
    = \frac{1}{2} \left( \frac{1}{\lambda_1} - \frac{1}{\mu_1} \right)
    = \frac{1}{2} \left( \frac{1}{1/3} - \frac{1}{1} \right)
    = 1
    ,\\
    d_2 &
    = \frac{1}{2} \left( \frac{1}{\lambda_2} - \frac{1}{\mu_1} \right)
    = \frac{1}{2} \left( \frac{1}{3} - \frac{1}{1} \right)
    = -\frac{1}{3}
    ,\\
    d_3 &
    = \frac{1}{2} \left( \frac{1}{\lambda_2} \right)
    = \frac{1}{2} \left( \frac{1}{3} \right)
    = \frac{1}{6}
    .
    \end{aligned}
  \end{equation*}
  As $t \to -\infty$,
  the four curves approach, respectively, the four lines
  \begin{equation*}
    \begin{aligned}
      y &= d_1 t
      + \ln \frac{\mu_1}{\lambda_1}
      + \frac12 \ln \frac{a_1(0)}{2 b_1(0) \, (\lambda_1 + \mu_1)}
      + \frac12 \ln \frac{\bigl( \lambda_1-\lambda_2 \bigr)^2 (\lambda_2 + \mu_1)}{\lambda_2^3}
      \\ &
      = t + \ln 3 + \frac12 \ln \frac{3}{800} + \frac12 \ln \frac{256}{243}
      ,\\
      y &= d_1 t
      + \ln \mu_1
      - \frac12 \ln \frac{b_1(0) \, (\lambda_1 + \mu_1)}{2 a_2(0)}
      + \frac12 \ln \frac{\bigl( \lambda_1-\lambda_2 \bigr)^2 (\lambda_2 + \mu_1)}{\lambda_2^3}
      \\ &
      = t - \frac12 \ln \frac{200}{3} + \frac12 \ln \frac{256}{243}
      ,\\
      y &= d_2 t
      - \ln \lambda_2
      + \frac12 \ln \frac{a_2(0) \, (\lambda_2 + \mu_1)}{2 b_1(0)}
      = - \frac{t}{3} - \ln 3 + \frac12 \ln \frac{1}{50}
      ,\\
      y &= d_3 t
      - \frac12 \ln \frac{b_{\infty}}{2 a_2(0)}
      = \frac{t}{6} - \frac12 \ln 500
      ,
    \end{aligned}
  \end{equation*}
  according to \eqref{eq:asymptotics-K2-log-amplitude-minus-infty},
  or more generally \autoref{thm:asymptotics-generalK-amplitudes}
  in the $K+K$ case.
  As $t \to +\infty$, they approach instead the lines
  \begin{equation*}
    \begin{aligned}
      y &= d_3 t
      + \frac12 \ln \frac{b^*_{\infty} a_2(0) \, \mu_1}{\lambda_1 \lambda_2}
      + \frac12 \ln \frac{\bigl( \lambda_1-\lambda_2 \bigr)^2}{\lambda_1 (\lambda_2 + \mu_1)}
      = \frac{t}{6} + \frac12 \ln 100 + \frac12 \ln \frac{16}{3}
      ,\\
      y &= d_2 t
      + \ln \mu_1
      - \frac12 \ln \frac{b_1(0) \, (\lambda_2 + \mu_1)}{2 a_2(0)}
      + \frac12 \ln \frac{\bigl( \lambda_1-\lambda_2 \bigr)^2 (\lambda_1 + \mu_1)}{\lambda_1^3}
      \\ &
      = - \frac{t}{3} - \frac12 \ln 200 + \frac12 \ln 256
      ,\\
      y &= d_1 t
      - \ln \lambda_1
      + \frac12 \ln \frac{a_1(0) \, (\lambda_1 + \mu_1)}{2 b_1(0)}
      = t - \ln \frac{1}{3} + \frac12 \ln \frac{1}{150}
      ,\\
      y &= d_1 t
      - \frac12 \ln \frac{b_1(0)}{2 a_1(0) (\lambda_1 + \mu_1)}
      = t - \frac12 \ln \frac{75}{2}
      ,
    \end{aligned}
  \end{equation*}
  according to \eqref{eq:asymptotics-K2-log-amplitude-plus-infty}
  or \autoref{thm:asymptotics-generalK-amplitudes}.
  Here, too, phase shifts between incoming and outgoing asymptotic lines
  with the same slope are easily computed;
  see \autoref{sec:phaseshift-K2-amplitudes}
  and \autoref{cor:phaseshift-amplitudes}.

  Note that even though $u(x,t)$ and~$v(x,t)$
  exhibit exponential growth,
  their product $u(x,t) \, v(x,t)$ stays bounded as $t \to \pm\infty$;
  it is this quantity which determines the velocity of the peakons,
  according to the ODEs~\eqref{eq:GX-peakon-ode}:
  \begin{equation*}
    \dot x_k = u(x_k) \, v(x_k)
    .
  \end{equation*}
  Consequently,
  \begin{equation*}
    uv\big|_{x=x_1(t)} \sim c_1
    ,\quad
    uv\big|_{x=x_2(t)} \sim c_1
    ,\quad
    uv\big|_{x=x_3(t)} \sim c_2
    ,\quad
    uv\big|_{x=x_4(t)} \sim c_3
    ,
  \end{equation*}
  as $t \to -\infty$,
  and
  \begin{equation*}
    uv\big|_{x=x_1(t)} \sim c_3
    ,\quad
    uv\big|_{x=x_2(t)} \sim c_2
    ,\quad
    uv\big|_{x=x_3(t)} \sim c_1
    ,\quad
    uv\big|_{x=x_4(t)} \sim c_1
    ,
  \end{equation*}
  as $t \to +\infty$.
  This is clearly seen in \autoref{fig:GX-interlacing-2+2-3Dplot-uv},
  which shows the graph of the function $u(x,t) \, v(x,t)$.
\end{example}

For the general $K+K$ interlacing pure peakon solution,
the solution formulas are given
in terms of abbreviated notation defined in \autoref{sec:preliminaries};
the statement is given in
\autoref{thm:inverse-spectral-map-formulas}
and \autoref{cor:peakon-solution-formulas}.
Already in the $3+3$ case,
\begin{equation*}
  \begin{aligned}
    u(x,t) &
    = m_1(t) \, e^{-\abs{x-x_1(t)}} + m_3(t) \, e^{-\abs{x-x_3(t)}} + m_5(t) \, e^{-\abs{x-x_5(t)}}
    ,\\
    v(x,t) &
    = n_2(t) \, e^{-\abs{x-x_2(t)}} + n_4(t) \, e^{-\abs{x-x_4(t)}} + n_6(t) \, e^{-\abs{x-x_6(t)}}
    ,
  \end{aligned}
\end{equation*}
the solution formulas contain so many terms that we
have chosen not write them out here in the expanded form that we
give for $K=2$
in~\eqref{eq:solution-K2}
;
however, this is partly done in~\cite[Ex.~4.11]{lundmark-szmigielski:GX-inverse-problem}.
The general solution for the $2K$ positions and the $2K$ amplitudes
depends on $4K$~parameters
whose values determine the behaviour of the solution;
$2K-1$ of them are eigenvalues of certain boundary value problems
coming from the two Lax pairs of the Geng--Xue equation,
\begin{equation*}
  0 < \lambda_1 < \lambda_2 < \dots < \lambda_K
  , \qquad
  0 < \mu_1 < \mu_2 < \dots < \mu_{K-1}
  ,
\end{equation*}
another $2K-1$ of them are residues of the associated Weyl functions,
\begin{equation*}
  a_1(0), \, a_2(0), \, \dots, a_K(0) \in \R_+
  ,\qquad
  b_1(0), \, b_2(0), \, \dots, b_{K-1}(0) \in \R_+
  ,
\end{equation*}
and there are also two additional parameters,
\begin{equation*}
  b_{\infty}, \, b^*_{\infty}  \in \R_+
  ,
\end{equation*}
where $-b_{\infty}$ is the limit at infinity of the second Weyl function,
and $-b_{\infty}^*$ is the corresponding quantity for a Weyl function
associated to an adjoint spectral problem.
Just like in \autoref{ex:GX-interlacing-2+2} above,
the solutions exhibit scattering as $t \to \pm\infty$:
the peakon trajectories $x=x_k(t)$ asymptotically approach
certain straight lines, whose slopes turn out to be determined by the
eigenvalues $\{ \lambda_i, \mu_j \}$ only.
Each amplitude grows or decays exponentially as $t \to \pm\infty$
(or tends to a constant, in borderline cases),
so the curves $y = \ln m_{2a-1}(t)$ and $y = -\ln n_{2a}(t)$ will
approach straight lines, whose slopes also are
determined by the eigenvalues only.
Proof of this asymptotic behaviour for general~$K$
is given in \autoref{sec:dynamics-K}.

\subsection{A brief history of peakons}
\label{sec:history}

A few words about the history of this problem are perhaps in order,
but since the story has been told many times, we will keep it short and
refer to our earlier articles for more details.
Peaked solitons of the form
\begin{equation}
  \label{eq:CHpeakons}
  u(x,t) = \sum_{k=1}^N m_k(t) \, e^{-\abs{x - x_k(t)}}
\end{equation}
were introduced by Camassa and Holm in 1993
\cite{camassa-holm} as solutions to their shallow water equation
\begin{equation}
  \label{eq:CH}
  m_t + m_x u + 2m u_x = 0
  ,\qquad
  m = u - u_{xx}
  .
\end{equation}
The ansatz \eqref{eq:CHpeakons} is a weak solution of the Camassa--Holm equation
\eqref{eq:CH} if and only if the positions~$x_k(t)$ and
amplitudes~$m_k(t)$ satisfy the canonical Hamiltonian system
generated by
\begin{equation*}
  H(x_1,\dots,x_n,m_1,\dots,m_n)
  = \frac12 \sum_{i,j=1}^N m_i m_j e^{-\abs{x_i-x_j}}
  ,
\end{equation*}
namely, using the shorthand notation~\eqref{eq:uv-shorthand},
\begin{equation}
  \label{eq:CH-peakon-ode}
  \dot x_k = \frac{\partial H}{\partial m_k} = u(x_k)
  ,\qquad
  \dot m_k = -\frac{\partial H}{\partial x_k} = -m_k u_x(x_k)
  .
\end{equation}
The general solution of \eqref{eq:CH-peakon-ode} (for arbitrary~$N$)
was computed by Beals, Sattinger and Szmigielski using inverse spectral methods;
it is given completely explicitly in terms of elementary functions
\cite{beals-sattinger-szmigielski:stieltjes,beals-sattinger-szmigielski:moment}.
Later, other similar integrable PDEs with explicitly
computable multipeakon solutions were discovered,
in particular the Degasperis--Procesi equation from 1998
\cite{degasperis-procesi,degasperis-holm-hone,lundmark-szmigielski:DPshort,lundmark-szmigielski:DPlong},
\begin{equation}
  \label{eq:DP}
  m_t + m_x u + 3m u_x = 0
  ,\qquad
  m = u - u_{xx}
  ,
\end{equation}
and V.~Novikov's cubically nonlinear equation from 2008
\cite{novikov:generalizations-of-CH,hone-wang:cubic-nonlinearity,hone-lundmark-szmigielski:novikov},
\begin{equation}
  \label{eq:Novikov}
  m_t + u (m_x u + 3m u_x) = 0
  ,\qquad
  m = u - u_{xx}
  ,
\end{equation}
both of which were found through mathematical (rather than physical)
considerations, namely the use of integrability tests to isolate
interesting equations similar in form to the Camassa--Holm equation.
The Degasperis--Procesi equation has later appeared in the context
of hydrodynamics
\cite{constantin-lannes:hydrodynamical-CH-DP,johnson:classical-water-waves},
but we are not aware of any
physical applications for the Novikov equation so far.
Geng and Xue~\cite{geng-xue:cubic-nonlinearity}
found their integrable two-component peakon equation~\eqref{eq:GX}
in 2009 by modifying the Lax pair for Novikov's equation
that was found by Hone and Wang~\cite{hone-wang:cubic-nonlinearity}.

The Degasperis--Procesi equation is special in that
it admits solutions where $u$ need not be continuous
\cite{coclite-karlsen:DPwellposedness,
coclite-karlsen:DPuniqueness},
in particular \emph{shockpeakons}~\cite{lundmark:shockpeakons}.
For the Camassa--Holm and Novikov equations
(as well as for other integrable peakon PDEs known so far),
the derivative $u_x$ may behave badly,
but $u$ itself must be continuous in order to make sense as a solution.
It is therefore a particularly interesting feature of the Geng--Xue
equation~\eqref{eq:GX}
that it also admits discontinuous solutions.

The literature on the Camassa--Holm equation is enormous,
and we will not attempt to survey it here.
There are also plenty of articles devoted
to the Degasperis--Procesi equation,
so we only mention a few additional references
particularly close to the topic of this article:
\cite{bertola-gekhtman-szmigielski:cubicstring,
bertola-gekhtman-szmigielski:cauchy,
kohlenberg-lundmark-szmigielski,
szmigielski-zhou:shocks-DP,
szmigielski-zhou:DP-peakon-antipeakon}.
Novikov's equation is more recent, but it is beginning to
attract attention; see for example
\cite{chen-chen-liu:novikov-global-conservative-weak,
geng-xue:cubic-nonlinearity,
grayshan:novikov-data-to-solution-map,
himonas-holliman:novikov-cauchy-problem,
hone-wang:cubic-nonlinearity,
jiang-ni:novikov-blowup,
lai-li-wu:novikov-global-solutions,
mi-mu:modified-novikov-cauchy-problem,
mohajer-szmigielski:novikov-on-real-axis,
ni-zhou:novikov,
tiglay:novikov-periodic-cauchy-problem,
wu-guo:periodic-novikov-global-wellposedness,
wu-yin:novikov-wellposedness-global-existence,
yan-li-zhang:novikov-cauchy-problem}.
Concerning the Geng--Xue equation,
we are only aware of a few studies:
\cite{%
barostichi-himonas-petronilho:ovsyannikov-theorem,
geng-xue:cubic-nonlinearity,
li-liu:GX-bihamiltonian,
li-niu:GX-reciprocal,
lundmark-szmigielski:GX-inverse-problem,
mi-mu-tao:GX-cauchy-problem,
tang-liu:GX-cauchy-problem}.
A bihamiltonian $2n$-component system which reduces to the Geng--Xue equation
when $n=1$ is constructed in~\cite{li-li-chen:multi-component-novikov}.

\subsection{Outline of the article}
\label{sec:intro-outline}

In \autoref{sec:weak-solutions-GX}, we begin our study by deriving the ODEs
for peakons and shockpeakons, and explaining the distributional sense
in which we consider them to be solution of the Geng--Xue equation.
Most of the computations are postponed to \autoref{sec:lax-pair},
where it is also verified that the Lax formulation
of the Geng--Xue equation is compatible with the peakon ODEs.
(We do not know at present whether this can be extended to cover
the shockpeakon case as well.)

\autoref{sec:preliminaries} is a review of notation and results
from our previous article~\cite{lundmark-szmigielski:GX-inverse-problem}
about an inverse spectral problem associated with the Geng--Xue
Lax pairs.
This technical foundation allows us to fairly easily derive the explicit
solution formulas for the interlacing peakon solutions
in \autoref{sec:time-spectral}.

Some readers may want to skip most of \autoref{sec:preliminaries},
since the abbreviated notation defined there
will not really be needed until we come to $K+K$ interlacing
peakon solutions for arbitrary~$K$
in \autoref{sec:dynamics-K} at the end of the article;
before that, we only deal with smaller cases where all formulas
can be written out in full detail.
Specifically,
\autoref{sec:dynamics-1} deals with the simple but somewhat
exceptional case of $1+1$ peakon solutions,
in \autoref{sec:dynamics-shock-1} we show how to integrate
the $1+1$ shockpeakon ODEs and
take a brief look at some properties of the solution,
and \autoref{sec:dynamics-2} studies the dynamics of
$2+2$ interlacing pure peakon solutions
(as already illustrated in \autoref{ex:GX-interlacing-2+2} above).

Mixed peakon--antipeakon solutions are only considered in 
\autoref{rem:1+1-peakon-antipeakon} for the case $K=1$,
where they cause no problems,
and in
\autoref{sec:2+2-peakon-antipeakon} for the case $K=2$,
where it is found in one particular example
that there is a collision after finite time where
one of the components of the solution forms a jump discontinuity,
while the other component loses a peakon at the corresponding location
(the amplitude of the peakon tends to zero).
A natural continuation past this singularity
is given by a solution with one peakon and one shockpeakon;
such a solution is a special case of the $1+1$ shockpeakon solutions
studied in~\autoref{sec:dynamics-shock-1}.

Finally, in \autoref{sec:dynamics-K} we derive the large time asymptotics
for $K+K$ interlacing pure peakon solutions.
This is somewhat more technical notation-wise, but the outcome is that
the features seen already in the case $K=2$ persist also for
$K>2$.

\autoref{sec:conclusions} rounds off the article with a summary and a few
remarks about open questions for future research.

\section{Peakons and shockpeakons as weak solutions of the Geng--Xue equation}
\label{sec:weak-solutions-GX}

Our first item of business is to explain in which sense the
shockpeakon ansatz~\eqref{eq:GXshockpeakons},
and hence also the peakon ansatz~\eqref{eq:GXpeakons},
is a solution of the Geng--Xue equation.

For smooth functions $u$ and~$v$,
the Geng--Xue equation~\eqref{eq:GX} is equivalent to
\begin{equation}
  \label{eq:GX-pre-distributional}
  \begin{gathered}
    m_t + v \cdot (4-\partial_x^2) \partial_x (\tfrac12 u^2) = 0
    , \\
    n_t + u \cdot (4-\partial_x^2) \partial_x (\tfrac12 v^2) = 0
    ,
  \end{gathered}
\end{equation}
where $m=u-u_{xx}$ and $n=v-v_{xx}$, as before.
This rewriting is inspired by the fact that
the Degasperis--Procesi equation \eqref{eq:DP} can be written as
\begin{equation}
  \label{eq:DP-pre-distributional}
  m_t + (4-\partial_x^2) \partial_x (\tfrac12 u^2) = 0
  .
\end{equation}
In~\eqref{eq:GX-pre-distributional} and~\eqref{eq:DP-pre-distributional},
the equalities can be interpreted in a distributional sense.
Assuming that
the functions $x \mapsto u(x,t)^2$ and $x \mapsto v(x,t)^2$
are locally integrable for each fixed~$t$,
we can consider them as distributions in the space~$\spaceDprime$.
Then the derivatives with respect to~$x$ in \eqref{eq:GX-pre-distributional}
can be taken in the sense of distributions,
while the time derivatives are defined as limits (in~$\spaceDprime$)
of difference quotients in the $t$~direction;
see \eqref{eq:def-Dt} in \autoref{sec:notation-distributions}.

In the notation of \autoref{sec:lax-pair},
where we use $D_x$ for the distributional derivative
and $D_t$ for the time derivative as just explained,
the interpretation that we propose is thus
\begin{equation}
  \label{eq:GX-distributional}
  \begin{gathered}
    D_t (u - D_x^2 u) + v \cdot (4-D_x^2) D_x (\tfrac12 u^2) = 0
    , \\
    D_t (v - D_x^2 v) + u \cdot (4-D_x^2) D_x (\tfrac12 v^2) = 0
    .
  \end{gathered}
\end{equation}
However, for~\eqref{eq:GX-distributional} to make sense, it is necessary that
the distribution
$(4-D_x^2) D_x (\tfrac12 u^2)$
can be multiplied by the function~$v$,
and similarly with $u$ and~$v$ interchanged.
To ensure that this is possible in the context of peakons and shockpeakons,
without having to make any \textit{ad hoc}
assignments of values at jump discontinuities
we need to impose the \emph{non-overlapping}
condition mentioned in \autoref{sec:intro-shockpeakons}:
the component~$v$ must not have a peakon or a shockpeakon at any point
where the other component~$u$ has a peakon or a shockpeakon,
and vice versa.
Then, since $u^2$ is piecewise continuous,
$(4-D_x^2) D_x (\tfrac12 u^2)$ will involve nothing
worse than Dirac deltas and their first and second derivatives,
and this will be multiplied by a function $v$ which is smooth in a neighbourhood
of the support of these singular distributions,
so the products can be evaluated using the rules
\begin{equation}
  \label{eq:delta-multiplication-rules}
  \begin{aligned}
    f(x) \, \delta_a &= f(a) \, \delta_a
    ,\\
    f(x) \, \delta'_a &= f(a) \, \delta'_a - f'(a) \, \delta_a
    ,\\
    f(x) \, \delta''_a &= f(a) \, \delta''_a - 2 f'(a) \, \delta'_a +  f''(a) \, \delta_a
    .
  \end{aligned}
\end{equation}

\begin{theorem}
  \label{thm:GX-shockpeakons}
  The shockpeakon ansatz~\eqref{eq:GXshockpeakons}
  is a solution of the Geng--Xue equation
  \eqref{eq:GX-distributional},
  in the distributional sense just described,
  if and only if
  it is non-overlapping and satisfies the ODEs~\eqref{eq:GX-shockpeakon-ode}.
  As a special case,
  the peakon ansatz~\eqref{eq:GXpeakons}
  is a solution of the Geng--Xue equation
  if and only if
  it is non-overlapping and satisfies the ODEs~\eqref{eq:GX-peakon-ode}.
\end{theorem}

\begin{proof}
  See \autoref{sec:proof-shockpeakon-odes} in the appendix.
\end{proof}

\begin{remark}
  It is understood here that the ordering assumption
  $x_1 < \dots < x_n$ must be fulfilled.
  If this condition holds at time $t=0$,
  then it will hold at least in some neighbourhood of $t=0$,
  so the ODEs always provide a local solution of the PDE.
  We will see that for pure peakon solutions,
  the ordering is automatically preserved for all~$t$,
  so that the solution is global,
  whereas for peakon--antipeakon or shockpeakon solutions
  this may not be the case.
\end{remark}

\begin{remark}
  Because of our assumption of non-overlapping we cannot
  perform the reduction $u=v$ which for smooth solutions
  turns the Geng--Xue equation into two copies of
  the Novikov equation.
  But since the Novikov equation does admit peakon solutions,
  we do not rule out the possibility that there is some way
  of defining solutions which would allow overlapping peakons
  or even shockpeakons.
  This is an interesting question and we leave it for future research.
  Let us just remark that
  in a multipeakon ansatz with overlapping,
  the distribution $D_x (4-D_x^2) (\tfrac12 u^2)$
  is a linear combination of $\delta$ and $\delta'$,
  while $v_x$ jumps at the location of those singular terms,
  so with our distributional approach
  we would need to assign some value to $v_x(x_k)$
  in the multiplication
  \begin{equation*}
    v(x) \, \delta'_{x_k} = v(x_k) \, \delta'_{x_k} - v_x(x_k) \, \delta_{x_k}
    ,
  \end{equation*}
  and of course also to $u_x(x_k)$ in the other equation.
\end{remark}

\begin{remark}
  Tang and Liu~\cite{tang-liu:GX-cauchy-problem}
  study solutions of the Geng--Xue equation
  with $u(\cdot,t)$ and~$v(\cdot,t)$ in the Besov space $B^{5/2}_{2,1}(\R)$,
  and write it as
  \begin{equation}
    \label{eq:GXweak-tang-liu}
    \begin{gathered}
      u_t + u u_x v
      + (1-\partial_x^2)^{-1} \bigl(
      3 u u_x v
      + 2 u_x^2 v_x
      + 2 u u_{xx} v_x
      + u u_x v_{xx}
      \bigr)
      = 0
      , \\
      v_t + v v_x u
      + (1-\partial_x^2)^{-1} \bigl(
      3 v v_x u
      + 2 v_x^2 u_x
      + 2 v v_{xx} u_x
      + v v_x u_{xx}
      \bigr)
      = 0
      ,
    \end{gathered}
  \end{equation}
  where $(1-\partial_x^2)^{-1}$ means convolution with $\tfrac12 e^{-\abs{x}}$.
  (See their equations (1.6)--(1.8).)
  A similar formulation was used by Mi, Mu and Tao \cite[eq.~(56)]{mi-mu-tao:GX-cauchy-problem}.
  Since these formulations require the derivatives
  $u_{xx}$ and~$v_{xx}$ to be in~$L^1$,
  they are not general enough to incorporate peakon solutions.
\end{remark}

\begin{remark}
  For the Degasperis--Procesi equation
  $D_t (u-D_x^2 u) + D_x (4-D_x^2) (\tfrac12 u^2) = 0$,
  the computation in the proof of \autoref{thm:GX-shockpeakons}
  provides a derivation of the shockpeakon ODEs
  \begin{equation}
    \label{eq:DP-shockpeakon-ode}
    \begin{aligned}
      \dot x_k &= u(x_k)
      ,\\
      \dot m_k &= -2 m_k \, u_x(x_k) + 2 s_k \, u(x_k)
      ,\\
      \dot s_k &= -s_k \, u_x(x_k)
    \end{aligned}
  \end{equation}
  which is simpler than the one originally given in~\cite{lundmark:shockpeakons}:
  just identify coefficients in \eqref{eq:DP-operator} and~\eqref{eq:m-t}.
\end{remark}

\section{Preliminaries: The map to spectral variables, and its inverse}
\label{sec:preliminaries}

The main technical work needed for analyzing the
$K+K$ interlacing peakon solutions was done in
our previous article~\cite{lundmark-szmigielski:GX-inverse-problem}.
In this section, we summarize the relevant material from that article;
it will be crucial in the following sections.

Throughout this section, we will assume implicitly that $K \ge 2$.
The case~$K=1$ is exceptional and will be treated separately
in \autoref{sec:map-K1}.

\subsection{The forward spectral map for $K \ge 2$}
\label{sec:forward-spectral-map}

First we describe the forward spectral map,
a change of variables
which takes the $4K$ ``physical'' variables
describing the positions and amplitudes of an interlacing
peakon solution,
\begin{equation}
  \label{eq:physical-parameters}
  x_1 < x_2 < \dots < x_{2K-1} < x_{2K}
  , \quad
  m_1, m_3, \dots, m_{2K-1} \in \R_+
  , \quad
  n_2, n_4, \dots, n_{2K} \in \R_+
  ,
\end{equation}
to a set of $4K$ spectral variables
\begin{equation}
  \label{eq:spectral-parameters}
  \begin{gathered}
    0 < \lambda_1 < \lambda_2 < \dots < \lambda_K
    , \qquad
    0 < \mu_1 < \mu_2 < \dots < \mu_{K-1}
    ;
    \\
    a_1,a_2,\dots,a_K \in \R_+
    ,\qquad
    b_1,b_2,\dots,b_{K-1} \in \R_+
    ,\qquad
    b_{\infty}, b^*_{\infty}  \in \R_+
    .
  \end{gathered}
\end{equation}
It was shown in \cite{lundmark-szmigielski:GX-inverse-problem}
that this map is a bijection, and
the inverse map (which is much more explicit) will be described
in \autoref{sec:inverse-spectral-map}.
Combining this with the time dependence for the spectral variables,
derived in \autoref{sec:time-spectral},
we get explicit formulas for the general interlacing
solution to the peakon ODEs~\eqref{eq:GX-peakon-ode}.
Using these formulas, the dynamics of interlacing peakons will be
analyzed in \autoref{sec:dynamics-1} (for the case $K=1$),
\autoref{sec:dynamics-2} (for $K=2$),
and \autoref{sec:dynamics-K} (for arbitrary~$K$).

As shown in \cite{geng-xue:cubic-nonlinearity},
the Geng--Xue equation \eqref{eq:GX} arises
as the compatibility condition of the Lax pair
\begin{subequations}
  \label{eq:laxI}
  \begin{equation}
    \label{eq:laxI-x}
    \frac{\partial}{\partial x}
    \begin{pmatrix} \psi_1 \\ \psi_2 \\ \psi_3 \end{pmatrix} =
    \begin{pmatrix}
      0 & zn & 1 \\
      0 & 0 & zm \\
      1 & 0 & 0
    \end{pmatrix}
    \begin{pmatrix} \psi_1 \\ \psi_2 \\ \psi_3 \end{pmatrix},
  \end{equation}
  \begin{equation}
    \label{eq:laxI-t}
    \frac{\partial}{\partial t}
    \begin{pmatrix} \psi_1 \\ \psi_2 \\ \psi_3 \end{pmatrix} =
    \begin{pmatrix}
      -v_xu & v_x z^{-1}-vunz & v_xu_x \\
      u z^{-1} & v_xu-vu_x-z^{-2} & -u_x z^{-1}-vumz \\
      -vu & v z^{-1} & vu_x
    \end{pmatrix}
    \begin{pmatrix} \psi_1 \\ \psi_2 \\ \psi_3 \end{pmatrix},
  \end{equation}
\end{subequations}
where $z$ is the spectral parameter.
However, because of the obvious symmetry in the Geng--Xue equation,
it also arises as the compatibility condition of a different
Lax pair, obtained by interchanging $u$ and~$v$:
\begin{subequations}
  \label{eq:laxII}
  \begin{equation}
    \label{eq:laxII-x}
    \frac{\partial}{\partial x}
    \begin{pmatrix} \twin\psi_1 \\ \twin\psi_2 \\ \twin\psi_3 \end{pmatrix} =
    \begin{pmatrix}
      0 & zm & 1 \\
      0 & 0 & zn \\
      1 & 0 & 0
    \end{pmatrix}
    \begin{pmatrix} \twin\psi_1 \\ \twin\psi_2 \\ \twin\psi_3 \end{pmatrix},
  \end{equation}
  \begin{equation}
    \label{eq:laxII-t}
    \frac{\partial}{\partial t}
    \begin{pmatrix} \twin\psi_1 \\ \twin\psi_2 \\ \twin\psi_3 \end{pmatrix} =
    \begin{pmatrix}
      -u_xv & u_x z^{-1}-uvmz & u_xv_x \\
      v z^{-1} & u_xv-uv_x-z^{-2} & -v_x z^{-1}-uvnz \\
      -uv & u z^{-1} & uv_x
    \end{pmatrix}
    \begin{pmatrix} \twin\psi_1 \\ \twin\psi_2 \\ \twin\psi_3 \end{pmatrix}.
  \end{equation}
\end{subequations}
(In the case $u=v$, when also $m=u-u_{xx}$ and $n=v-v_{xx}$ coincide,
these Lax pairs reduce to the one found by Hone and Wang
\cite{hone-wang:cubic-nonlinearity} for Novikov's equation
\eqref{eq:Novikov}, and the Geng--Xue equation
reduces to two copies of Novikov's equation.)

Since we are dealing with the interlacing case,
with the first (leftmost) peakon appearing in~$u$,
the second in~$v$, etc.,
the setup is not symmetric, and spectral data from
\emph{both} Lax pairs must be used in order to solve
the inverse spectral problem
which will let us
compute the peakon positions and amplitudes.

When $u$ and~$v$ are given by the interlacing peakon ansatz
\eqref{eq:GX-interlacing-peakons},
$m$ and~$n$ are discrete measures,
as explained in \autoref{sec:lax-pair}.
Interpreting the derivatives in the Lax equations in a suitable
distributional sense,
and imposing boundary conditions on
\eqref{eq:laxI-x} and \eqref{eq:laxII-x}
which are compatible with
the time evolution given by
\eqref{eq:laxI-t} and \eqref{eq:laxII-t},
we obtain finite-dimensional eigenvalue problems
which define our spectral data.
Here we will keep the exposition brief,
and merely state the resulting formulas
which are necessary for defining the spectral data.
For details, see \cite{lundmark-szmigielski:GX-inverse-problem},
in particular Appendix~B.

Consider equation \eqref{eq:laxI-x} for a fixed~$t$
(which we will suppress in the notation).
Since $m$ and $n$ are zero away from the points $x=x_k$,
it follows that $\psi_2(x;z)$ is piecewise constant, and
$\psi_1(x;z)$ and $\psi_3(x;z)$ are piecewise linear combinations
of $e^{x}$ and~$e^{-x}$.
We impose the following boundary condition on the left:
\begin{equation}
  \label{eq:psi-left}
  \begin{pmatrix} \psi_1(x;z) \\ \psi_2(x;z) \\ \psi_3(x;z) \end{pmatrix} =
  \begin{pmatrix}
    e^x \\ 0 \\ e^x
  \end{pmatrix},
  \qquad
  x < x_1
  .
\end{equation}
Then we get on the right
\begin{equation}
  \label{eq:psi-right}
    \begin{pmatrix} \psi_1(x;z) \\ \psi_2(x;z) \\ \psi_3(x;z) \end{pmatrix} =
    \begin{pmatrix}
      A(-z^2) e^x + z^2 C(-z^2) e^{-x} \\ 2z B(-z^2) \\ A(-z^2) e^x - z^2 C(-z^2) e^{-x}
    \end{pmatrix},
    \qquad
    x > x_N
    ,
\end{equation}
with polynomials $A(\lambda)$, $B(\lambda)$ and~$C(\lambda)$,
of degrees $K$, $K-1$ and~$K-1$, respectively,
defined by
\begin{equation}
  \label{eq:ABC-jump-product}
  \begin{pmatrix} A(\lambda) \\ B(\lambda) \\ C(\lambda) \end{pmatrix}
  = S_{2K}(\lambda) S_{2K-1}(\lambda) \dotsm S_{2}(\lambda) S_{1}(\lambda)
  \begin{pmatrix} 1 \\ 0 \\ 0 \end{pmatrix}
  ,
\end{equation}
where, for $a=1,\dots,K$,
\begin{equation}
  \label{eq:jumpI-even-odd}
  S_{k}(\lambda) =
  \begin{cases}
    \begin{pmatrix}
      1 & 0 & 0 \\
      m_{k} e^{x_{k}} & 1 & \lambda m_{k} e^{-x_{k}} \\
      0 & 0 & 1
    \end{pmatrix},
    & k = 2a-1
    ,
    \\[5ex]
    \begin{pmatrix}
      1 & -2\lambda n_{k} e^{-x_{k}} & 0 \\
      0 & 1 & 0 \\
      0 &  2 n_{k} e^{x_{k}} & 1
    \end{pmatrix},
    & k = 2a
    .
  \end{cases}
\end{equation}

The second Lax equation \eqref{eq:laxII-x} is similar,
with $m$ and~$n$ swapped. So with
\begin{equation}
  \label{eq:psi-twin-left}
  \begin{pmatrix} \twin\psi_1(x;z) \\ \twin\psi_2(x;z) \\ \twin\psi_3(x;z) \end{pmatrix} =
  \begin{pmatrix}
    e^x \\ 0 \\ e^x
  \end{pmatrix},
  \qquad
  x < x_1
  ,
\end{equation}
we get
\begin{equation}
  \label{eq:psi-twin-right}
    \begin{pmatrix} \twin\psi_1(x;z) \\ \twin\psi_2(x;z) \\ \twin\psi_3(x;z) \end{pmatrix} =
    \begin{pmatrix}
      \twin A(-z^2) e^x + z^2 \twin C(-z^2) e^{-x} \\ 2z \twin B(-z^2) \\ \twin A(-z^2) e^x - z^2 \twin C(-z^2) e^{-x}
    \end{pmatrix},
    \qquad
    x > x_N
    ,
\end{equation}
with polynomials
$\twin A(\lambda)$, $\twin B(\lambda)$ and~$\twin C(\lambda))$
of degrees $K-1$, $K-1$ and~$K-2$, respectively,
defined by
\begin{equation}
  \label{eq:twin-ABC-jump-product}
  \begin{pmatrix} \twin A(\lambda) \\ \twin B(\lambda) \\ \twin C(\lambda) \end{pmatrix}
  = \twin S_{2K}(\lambda) \twin S_{2K-1}(\lambda) \dotsm \twin S_{2}(\lambda) \twin S_{1}(\lambda)
  \begin{pmatrix} 1 \\ 0 \\ 0 \end{pmatrix}
  ,
\end{equation}
where
\begin{equation}
  \label{eq:jumpII-even-odd}
  \twin S_{k}(\lambda) =
  \begin{cases}
    \begin{pmatrix}
      1 & -2\lambda m_{k} e^{-x_{k}} & 0 \\
      0 & 1 & 0 \\
      0 &  2 m_{k} e^{x_{k}} & 1
    \end{pmatrix},
    & k = 2a-1,
    \\
    & \\
    \begin{pmatrix}
      1 & 0 & 0 \\
      n_{k} e^{x_{k}} & 1 & \lambda n_{k} e^{-x_{k}} \\
      0 & 0 & 1
    \end{pmatrix},
    & k = 2a.
  \end{cases}
\end{equation}

If all masses $m_{2a-1}$ and~$n_{2a}$ are positive,
then it turns out that the polynomial~$A$ has positive simple zeros
\begin{equation*}
  0 < \lambda_1 < \lambda_2 < \dots < \lambda_K
  ,
\end{equation*}
and likewise~$\twin A$ has positive simple zeros
\begin{equation*}
  0 < \mu_1 < \mu_2 < \dots < \mu_{K-1}
  ,
\end{equation*}
and we will refer to these zeros $ \{ \lambda_i, \mu_j \}$
as the \emph{eigenvalues}
of the spectral problems above.
Moreover we define \emph{residues}
\begin{equation*}
  a_1,a_2,\dots,a_K
  ;\quad
  b_1,b_2,\dots,b_{K-1}
  ;\quad
  b_{\infty}
\end{equation*}
from the partial fractions decomposition of \emph{Weyl functions}
$W$ and~$\twin W$:
\begin{equation}
  \label{eq:W-Wtwin-parfrac}
  W(\lambda)
  = - \frac{B(\lambda)}{A(\lambda)}
  = \sum_{i=1}^K \frac{a_i}{\lambda - \lambda_i}
  ,\qquad
  \twin W(\lambda)
  = - \frac{\twin B(\lambda)}{\twin A(\lambda)}
  = -b_{\infty} + \sum_{j=1}^{K-1} \frac{b_j}{\lambda - \mu_j}
  .
\end{equation}
The residues can be shown to be positive
if all masses $m_{2a-1}$ and~$n_{2a}$ are positive.

There is one final piece of spectral data, $b^*_{\infty}$,
which is needed in order to recover the mass~$m_1$
and position~$x_1$ of the leftmost peakon.
It arises in a natural way from an adjoint spectral problem,
but to avoid introducing too much notation,
we take a more direct route here and simply define it as
\begin{equation}
  \label{eq:b-infty-star-directly}
  b^*_{\infty}
  = m_1 e^{-x_1} (1 - E_{12}^2)
  ,
\end{equation}
where we use the abbreviation
$E_{ij} = e^{-\abs{x_i-x_j}} = e^{x_i-x_j}$ for $i<j$.

This concludes the description of the forward spectral map.

\begin{remark}
  Let us give some motivation for why this particular quantity
  $b^*_{\infty}$ might be of interest.
  Note from \eqref{eq:W-Wtwin-parfrac}
  that
  \begin{equation*}
    b_{\infty}
    = \lim_{\lambda \to \infty} \frac{\twin B(\lambda)}{\twin A(\lambda)}
    .
  \end{equation*}
  The coefficients of the polynomials
  $\twin A(\lambda)$ and $B(\lambda)$
  can be worked out from the defining matrix products
  \eqref{eq:twin-ABC-jump-product};
  in particular, the highest coefficients are given in equation~(B.23)
  in~\cite{lundmark-szmigielski:GX-inverse-problem},
  from which it follows that
  \begin{equation}
    \label{eq:b-infty-directly}
    b_{\infty}
    = n_{2K} e^{x_{2K}} (1 - E_{2K-1,2K}^2)
    .
  \end{equation}
  So there is a sort of duality between the roles played by
  $b_{\infty}$ and~$b^*_{\infty}$.
\end{remark}

\begin{remark}
  \label{rem:b-infty-and-star-constants-of-motion}
  It can be verified directly by differentiation
  of the expressions in \eqref{eq:b-infty-star-directly}
  and \eqref{eq:b-infty-directly}
  that both $b_{\infty}$ and~$b^*_{\infty}$
  are constants of motion for the Geng--Xue peakon ODEs~\eqref{eq:GX-peakon-ode}.
  (Cf.~\autoref{thm:spectral-variables-ode}.)
\end{remark}

\subsection{The inverse spectral map for $K \ge 2$}
\label{sec:inverse-spectral-map}

The main result that we need from our previous work is the explicit
formulas for the inverse spectral map,
taken from Corollary~4.5
in~\cite{lundmark-szmigielski:GX-inverse-problem}.
These formulas are quoted in \autoref{thm:inverse-spectral-map-formulas} below,
but first we need to define a fair amount of notation.

\begin{definition}
  Given spectral data as in \eqref{eq:spectral-parameters},
  let
  \begin{equation}
    \label{eq:discrete-measures-alpha-beta}
    \alpha = \sum_{i=1}^K a_i \delta_{\lambda_i}
    , \qquad
    \beta = \sum_{j=1}^{K-1} b_j \delta_{\mu_j}
    ,
  \end{equation}
  where $\delta$ is the Dirac delta.
  These two discrete measures on the positive real axis~$\R_+$
  will be called the \emph{spectral measures}.
\end{definition}

\begin{remark}
  For the application to Geng--Xue peakons, the measures
  \eqref{eq:discrete-measures-alpha-beta}
  are the only ones that we will have in mind,
  but \autoref{def:heine-integral} below makes sense whenever
  $\alpha$ and $\beta$ are measures on~$\R_+$
  with finite moments
  \begin{equation}
    \label{eq:alpha-beta-moments}
    \alpha_k = \int x^k \da
    < \infty
    ,
    \qquad
    \beta_k = \int y^k \db
    < \infty
    ,
  \end{equation}
  and finite bimoments with respect to the Cauchy kernel $1/(x+y)$,
  \begin{equation}
    \label{eq:alpha-beta-bimoments}
    I_{ab} = \iint \frac{x^a y^b}{x+y} \da\db
    < \infty
    .
  \end{equation}
\end{remark}

\begin{definition}
  For $x = (x_1,\dots,x_n) \in \R^n$, let $\Delta(x)$
  denote the Vandermonde determinant
  \begin{equation}
    \label{eq:def-Delta}
    \Delta(x) = \Delta(x_1,\dots,x_n) = \prod_{i < j} (x_i - x_j)
  \end{equation}
  and $\Gamma(x)$ its counterpart with only plus signs,
  \begin{equation}
    \label{eq:def-Gamma}
    \Gamma(x) = \Gamma(x_1,\dots,x_n) = \prod_{i < j} (x_i + x_j);
  \end{equation}
  in both cases the right-hand side is interpreted as~$1$
  (the empty product) if $n=0$ or $n=1$.
  Moreover, for $x \in \R^n$ and $y \in \R^m$, let
  \begin{equation}
    \label{eq:def-Gamma-mixed}
    \Gamma(x;y) = \Gamma(x_1,\dots,x_n;y_1,\dots,y_m)
    = \prod_{i=1}^{n} \prod_{j=1}^{m} (x_i + y_j)
    .
  \end{equation}
\end{definition}

\begin{remark}
  We will not really need $\Gamma(x)$ here, only $\Gamma(x;y)$.
  But we have included the definition of $\Gamma(x)$ anyway, as it 
  occurs often in the study of peakons solutions of the Degasperis--Procesi
  and Novikov equations, and also in basic identities such as
  \begin{equation*}
    \Gamma(x_1,\dots,x_n;y_1,\dots,y_n) = 
    \frac{\Gamma(x_1,\dots,x_n,y_1,\dots,y_m)}{\Gamma(x_1,\dots,x_n) \, \Gamma(y_1,\dots,y_m)}
    .
  \end{equation*}
\end{remark}

\begin{definition}
  \label{def:heine-integral}
  With $\sigma_n$ denoting the sector in $\R_+^n$
  defined by the inequalities $0 < x_1 < \dots < x_n$,
  let
  \begin{equation}
    \label{eq:heine-integral}
    \heineintegral_{nm}^{rs} =
    \int_{\sigma_{n} \times \sigma_{m}}
    \frac{\Delta(x)^2 \Delta(y)^2 \Bigl( \prod_{i=1}^{n} x_i \Bigr)^r \Bigl( \prod_{j=1}^{m} y_j \Bigr)^s}{\Gamma(x;y)}  \dA{n} \dB{m},
  \end{equation}
  for $n$ and~$m$ positive.
  We also consider the degenerate cases
  \begin{equation}
    \label{eq:heine-integral-one-index-zero}
    \begin{split}
      \heineintegral_{n0}^{rs} &= \int_{\sigma_{n}} \Delta(x)^2 \Bigl( \prod_{i=1}^{n} x_i \Bigr)^r \dA{n} \qquad (n > 0),
      \\
      \heineintegral_{0m}^{rs} &= \int_{\sigma_{m}} \Delta(y)^2 \Bigl( \prod_{j=1}^{m} y_j \Bigr)^s \dB{m} \qquad (m > 0),
      \\[1ex]
      \heineintegral_{00}^{rs} &= 1.
    \end{split}
  \end{equation}
\end{definition}

\begin{remark}
  In particular,
  $\heineintegral_{10}^{rs} = \alpha_r$ and $\heineintegral_{01}^{rs} = \beta_s$
  are the moments~\eqref{eq:alpha-beta-moments}
  of the measures $\alpha$ and~$\beta$,
  and $\heineintegral_{11}^{rs} = I_{rs}$ is the Cauchy
  bimoment~\eqref{eq:alpha-beta-bimoments}.
\end{remark}

\begin{remark}
  The integrals $\heineintegral_{nm}^{rs}$
  arise as evaluations of certain bimoment determinants
  occurring naturally in the theory of Cauchy biorthogonal polynomials;
  see comments in Appendix~A.3 of \cite{lundmark-szmigielski:GX-inverse-problem}.
  This is similar to Heine's evaluation of the Hankel determinant
  of moments from the classical theory of orthogonal polynomials,
  \begin{equation*}
    \det (\alpha_{i+j})_{i,j=0}^{n-1}
    = \frac{1}{n!} \int_{\R^n} \Delta(x)^2 \dA{n}
    .
  \end{equation*}
\end{remark}

If we now specialize to the case
when $\alpha$ and~$\beta$ are the discrete measures
\eqref{eq:discrete-measures-alpha-beta},
the moments can be written as sums instead of integrals,
\begin{equation}
  \label{eq:moment-as-sum}
  \alpha_r = \int x^r \da
  = \sum_{i=1}^K \lambda_i^r \, a_i
  ,\qquad
  \beta_s = \int y^s \db
  = \sum_{j=1}^{K-1} \mu_j^s \, b_j,
\end{equation}
and likewise for the bimoments,
\begin{equation}
  \label{eq:bimoment-as-sum}
  I_{rs} = \iint \frac{x^r y^s}{x+y} \da\db
  = \sum_{i=1}^K \sum_{j=1}^{K-1} \frac{\lambda_i^r \mu_j^s}{\lambda_i+\mu_j} \, a_i b_j.
\end{equation}
The integrals $\heineintegral_{nm}^{rs}$ also turn into sums,
\begin{equation}
  \label{eq:heine-integral-as-sum}
  \heineintegral_{nm}^{rs} =
  \sum_{I \in \binom{[K]}{n}} \sum_{J \in \binom{[K-1]}{m}}
  \Psi_{IJ} \, \lambda_I^r a_I \, \mu_J^s b_J,
\end{equation}
where we have used yet some more notation, defined as follows:

\begin{definition}
  The binomial coefficient $\binom{S}{n}$
  denotes the collection of $n$-element subsets
  of the set~$S$,
  and $[k]$ denotes the integer interval $\{ 1,2,3, \dots, k \}$.
  We always label the elements of a set
  $I \in \binom{[k]}{n}$ in increasing order:
  $I = \{ i_1 < i_2 < \dots < i_n \}$.
  Moreover,
  \begin{equation}
    \label{eq:product-notation-explanation}
    \lambda_I^r a_I \, \mu_J^s b_J =
    \Bigl( \prod_{i \in I} \lambda_i^r a_i \Bigr)
    \Bigl( \prod_{j \in J} \mu_j^s b_j \Bigr)
  \end{equation}
  and
  \begin{equation}
    \label{eq:PsiIJ}
    \Psi_{IJ} = \frac{\Delta_I^2 \twin\Delta_J^2}{\Gamma_{IJ}},
  \end{equation}
  where
  \begin{equation}
    \label{eq:DeltaI}
    \begin{split}
      \Delta_I^2 &= \Delta(\lambda_{i_1},\dots,\lambda_{i_n})^2 = \prod_{\substack{a,b \in I \\ a < b}} (\lambda_{a} - \lambda_{b})^2
      , \\
      \twin\Delta_J^2 &= \Delta(\mu_{j_1},\dots,\mu_{j_m})^2 = \prod_{\substack{a,b \in J \\ a < b}} (\mu_{a} - \mu_{b})^2
      , \\
      \Gamma_{IJ} &= \Gamma(\lambda_{i_1},\dots,\lambda_{i_n};\mu_{j_1},\dots,\mu_{j_m}) = \prod_{i \in I, \, j \in J} (\lambda_i + \mu_j)
      .
    \end{split}
  \end{equation}
  Degenerate cases: we let $[0] = \emptyset$,
  and consider empty products to be equal to~$1$,
  so that
  \begin{equation*}
    \Delta_{\emptyset}^2
    = \Delta_{\{ i \}}^2
    = \twin\Delta_{\emptyset}^2
    = \twin\Delta_{\{ j \}}^2
    = \Gamma_{I\emptyset}
    = \Gamma_{\emptyset J}
    = 1
    .
  \end{equation*}
\end{definition}

\begin{remark}
  In the discrete case \eqref{eq:heine-integral-as-sum},
  we have $\heineintegral_{nm}^{rs} > 0$ if $0 \le n \le K$
  and $0 \le m \le K-1$, otherwise
  $\heineintegral_{nm}^{rs} = 0$.
\end{remark}

\begin{theorem}[Explicit formulas for the inverse spectral map]
  \label{thm:inverse-spectral-map-formulas}
  For $K \ge 2$,
  the inverse spectral map from the spectral variables
  \eqref{eq:spectral-parameters}
  to the ``physical'' variables
  \eqref{eq:physical-parameters}
  is given by the following formulas,
  in terms of the sums
  \eqref{eq:moment-as-sum},
  \eqref{eq:bimoment-as-sum}
  and
  \eqref{eq:heine-integral-as-sum}
  above.

  The even-numbered quantities are
  \begin{align}
    x_{2(K+1-j)} &= \frac12 \ln \left( \frac{2 \, \heineintegral_{j,j-1}^{00}}{\heineintegral_{j-1,j-2}^{11}} \right),
    \\
    n_{2(K+1-j)} &=
    \frac{\heineintegral_{j-1,j-1}^{10}}{\heineintegral_{j-1,j-2}^{01} \heineintegral_{j,j-1}^{01}}
    \sqrt{\frac{\heineintegral_{j,j-1}^{00} \heineintegral_{j-1,j-2}^{11}}{2}}
    ,
  \end{align}
  for $j = 2, \dots, K$,
  together with
  \begin{align}
    x_{2K} &= \frac12 \ln 2(I_{00} + b_{\infty} \alpha_0),
    \\
    n_{2K} &= \frac{1}{\alpha_0} \sqrt{\frac{I_{00} + b_{\infty} \alpha_0}{2}}.
  \end{align}
  The odd-numbered quantities are
  \begin{align}
    x_{2(K+1-j)-1} &= \frac12 \ln \left( \frac{2 \, \heineintegral_{jj}^{00}}{\heineintegral_{j-1,j-1}^{11}}\right),
    \\
    m_{2(K+1-j)-1} &=
    \frac{\heineintegral_{j,j-1}^{01}}{\heineintegral_{jj}^{10} \heineintegral_{j-1,j-1}^{10}}
    \sqrt{\frac{\heineintegral_{j-1,j-1}^{11} \heineintegral_{jj}^{00}}{2}}
    ,
  \end{align}
  for $j=1,\dots,K-1$, together with
  \begin{align}
    x_1 &= \frac12 \ln \left( \frac{2 \, \heineintegral_{K,K-1}^{00}}{\heineintegral_{K-1,K-2}^{11} + \dfrac{\strut 2 b^*_{\infty} L}{M} \, \heineintegral_{K-1,K-1}^{10}} \right)
    \\
    m_1 &=
    \frac{M/L}{\heineintegral_{K-1,K-1}^{10}}
    \sqrt{\frac{\heineintegral_{K,K-1}^{00}}{2} \left( \heineintegral_{K-1,K-2}^{11} + \dfrac{\strut 2 b^*_{\infty} L}{M} \, \heineintegral_{K-1,K-1}^{10} \right)}
  ,
  \end{align}
  where
  \begin{equation}
    \label{eq:LM}
    L = \prod_{i=1}^K \lambda_i
    ,\qquad
    M = \prod_{j=1}^{K-1} \mu_j
    .
  \end{equation}
\end{theorem}

\begin{remark}
  \label{rem:solution-alternative-form}
  Let us write $j' = K+1-j$.
  For example,
  $x_{2j'}$ then corresponds to $x_{2K}$, $x_{2K-2}$, $x_{2K-4}$, \ldots,
  as $j=1$, $2$, $3$, \ldots,
  i.e., it is the position of the $j$th of the even-numbered peakons if
  we count them \emph{from the right}.
  Then we can express the formulas in
  \autoref{thm:inverse-spectral-map-formulas} in the following more compact way:
  \begin{equation}
    \label{eq:solution-alternative-position}
    \begin{aligned}
      \tfrac12 \exp 2 x_{2j'}
      &=
      \begin{cases}
        \heineintegral_{11}^{00} + b_{\infty} \heineintegral_{10}^{00}
        \,
        ,&
        j = 1
        ,\\[1.5ex] \displaystyle
        \frac{\heineintegral_{j,j-1}^{00}}{\heineintegral_{j-1,j-2}^{11}}
        ,&
        j = 2,\dots,K
        ,
      \end{cases}
      \\[1ex]
      \tfrac12 \exp 2 x_{2j'-1}
      &=
      \begin{cases}
        \displaystyle
        \frac{\heineintegral_{jj}^{00}}{\heineintegral_{j-1,j-1}^{11}}
        ,&
        j=1,\dots,K-1
        ,\\[1.5em] \displaystyle
        \frac{\heineintegral_{K,K-1}^{00}}{\heineintegral_{K-1,K-2}^{11} + \dfrac{\strut 2 b^*_{\infty} L}{M} \, \heineintegral_{K-1,K-1}^{10}}
        ,&
        j=K
        ,
      \end{cases}
    \end{aligned}
  \end{equation}
  and
  \begin{equation}
    \label{eq:solution-alternative-amplitude}
    \begin{aligned}
      2 n_{2j'} \exp(-x_{2j'})
      &=
      \begin{cases}
        \displaystyle
        \frac{1}{\heineintegral_{10}^{00}}
        ,&
        j = 1
        ,\\[1.5em] \displaystyle
        \frac{\heineintegral_{j-1,j-2}^{11} \heineintegral_{j-1,j-1}^{10}}{\heineintegral_{j-1,j-2}^{01} \heineintegral_{j,j-1}^{01}}
        ,&
        j = 2,\dots,K
        ,
      \end{cases}
      \\[1ex]
      2 m_{2j'-1} \exp(-x_{2j'-1})
      &=
      \begin{cases}
        \displaystyle
        \frac{\heineintegral_{j-1,j-1}^{11} \heineintegral_{j,j-1}^{01}}{\heineintegral_{jj}^{10} \heineintegral_{j-1,j-1}^{10}}
        ,&
        j=1,\dots,K-1
        ,\\[1.5em] \displaystyle
        \frac{M \, \heineintegral_{K-1,K-2}^{11}}{L \, \heineintegral_{K-1,K-1}^{10}} + 2 b^*_{\infty}
        ,&
        j=K
        .
      \end{cases}
    \end{aligned}
  \end{equation}
\end{remark}

\subsection{The forward and inverse spectral map for~$K=1$}
\label{sec:map-K1}

In the case $K=1$, the correspondence between peakon variables
$(x_1,x_2,m_1,n_2)$ 
and spectral variables $(\lambda_1,a_1,b_{\infty},b_{\infty}^*)$
reduces to equation~(4.51)
in~\cite{lundmark-szmigielski:GX-inverse-problem}:
\begin{equation}
  \label{eq:map-K1}
  \begin{split}
    \tfrac12 e^{2 x_2} &= a_1 b_{\infty},\\
    \tfrac12 e^{-2 x_1} &= (2 a_1)^{-1} \lambda_1b^*_{\infty},\\
    2 n_2 e^{-x_2} &= \frac{1}{a_1}, \\
    2 m_1 e^{x_1} &= \frac{2a_1}{\lambda_1}.
  \end{split}
\end{equation}
These formulas define a bijection from the pure peakon sector
(where $m_1 > 0$, $n_2>0$ and $x_1<x_2$)
to the region where $\lambda_1$, $a_1$,
$b_{\infty}$ and~$b_{\infty}^*$ are all positive (as usual)
and in addition satisfy a nonlinear constraint particular to
the case $K=1$:
\begin{equation}
  \label{eq:nonlinear-constraint-K1}
  1 < 2 \, \lambda_1 \, b_{\infty} \, b_{\infty}^*
  .
\end{equation}
The inverse spectral map is immediately found by solving
\eqref{eq:map-K1} for the peakon variables:
\begin{equation}
  \label{eq:peakon-recovery-K1}
  x_1 = \frac12 \ln \frac{a_1}{\lambda_1 \, b_{\infty}^*}
  ,\qquad
  x_2 = \frac12 \ln 2\, a_1 \, b_{\infty}
  ,\qquad
  m_1 = \sqrt{\frac{a_1 \, b_{\infty}^*}{\lambda_1}}
  ,\qquad
  n_2 = \sqrt{\frac{b_{\infty}}{2\, a_1} }
  .
\end{equation}
See also \autoref{sec:dynamics-1},
in particular \autoref{rem:solution-K1-spectral}.

\section{Time dependence of the spectral variables}
\label{sec:time-spectral}

If the positions~$x_i$ and the amplitudes $m_{2a-1}$ and~$n_{2a}$
depend on time,
then the spectral map described in \autoref{sec:preliminaries},
being a bijection,
will induce a time dependence for the spectral variables
\eqref{eq:spectral-parameters} as well.
The point of the spectral map is that it transforms the
complicated time dependence given by the Geng--Xue peakon ODEs
into a very simple dependence for the spectral variables.

\begin{theorem}
  \label{thm:spectral-variables-ode}
  The ODEs \eqref{eq:GX-peakon-ode} for interlacing peakons
  are equivalent to
  the following linear ODEs for the spectral variables:
  \begin{equation}
    \label{eq:spectral-variables-ode}
    \frac{d \lambda_i}{dt} = 0
    ,\quad
    \frac{d a_i}{dt} = \frac{a_i}{\lambda_i}
    ,\quad
    \frac{d \mu_j}{dt} = 0
    ,\quad
    \frac{d b_j}{dt} = \frac{b_j}{\mu_j}
    ,\quad
    \frac{d b_{\infty}}{dt} = 0
    ,\quad
    \frac{d b^*_{\infty}}{dt} = 0
    .
  \end{equation}
  Hence, the variables $\{ \lambda_i, \mu_j, b_{\infty}, b^*_{\infty} \}$
  are constant,
  while $\{ a_i, b_j \}$ have the time dependence
  \begin{equation}
    \label{eq:residues-time-dependence}
    a_i(t) = a_i(0) \, e^{t/\lambda_i}
    ,\qquad
    b_j(t) = b_j(0) \, e^{t/\mu_j}
    .
  \end{equation}
\end{theorem}

\begin{proof}
  As we carefully verify in \autoref{sec:lax-pair},
  the peakon ODEs \eqref{eq:GX-peakon-ode}
  are equivalent to the Lax equations
  \eqref{eq:laxI-x} and \eqref{eq:laxI-t}
  with $u$ and~$v$ given by the interlacing $K+K$ peakon ansatz
  \eqref{eq:GX-interlacing-peakons}.
  For $x < x_1$ we have
  \begin{equation*}
    u = u_x = M_- \, e^x
    ,\qquad \text{where} \quad
    M_- = \sum_{a=1}^K m_{2a-1} \, e^{-x_{2a-1}}
    ,
  \end{equation*}
  and
  \begin{equation*}
    v = v_x = N_- \, e^x
    ,\qquad \text{where} \quad
    N_- = \sum_{a=1}^K n_{2a} \, e^{-x_{2a}}
    .
  \end{equation*}
  This makes both sides of \eqref{eq:laxI-t} vanish
  when $(\psi_1,\psi_2,\psi_3)=(e^x,0,e^x)$,
  showing that this boundary condition (which was imposed when
  defining the spectral data; see~\eqref{eq:psi-left})
  is indeed consistent with the time evolution.

  For $x > x_{2K}$ we get instead
  \begin{equation*}
    u = -u_x = M_+ \, e^{-x}
    ,\qquad \text{where} \quad
    M_+ = \sum_{a=1}^K m_{2a-1} \, e^{x_{2a-1}}
    ,
  \end{equation*}
  and
  \begin{equation*}
    v = -v_x = N_+ \, e^{-x}
    ,\qquad \text{where} \quad
    N_+ = \sum_{a=1}^K n_{2a} \, e^{x_{2a}}
    .
  \end{equation*}
  Inserting this into \eqref{eq:laxI-t}
  together with the expression \eqref{eq:psi-left} for
  $(\psi_1,\psi_2,\psi_3)$,
  identifying coefficients of the linearly
  independent functions $(e^x,1,e^{-x})$,
  and setting $\lambda = -z^2$,
  we find
  \begin{equation*}
    \frac{\partial A}{\partial t}(\lambda) = 0
    ,\qquad
    \frac{\partial B}{\partial t}(\lambda) = \frac{B(\lambda) - A(\lambda) \, M_+}{\lambda}
    ,\qquad
    \frac{\partial C}{\partial t}(\lambda) = \frac{2 \bigl( B(\lambda) - A(\lambda) \, M_+ \bigr) \, N_+}{\lambda}
    .
  \end{equation*}
  This shows that the polynomial~$A(\lambda)$ is constant in time,
  hence so are its zeros $\lambda_1,\dots,\lambda_K$.
  The time evolution of the Weyl function
  defined in \eqref{eq:W-Wtwin-parfrac},
  \begin{equation*}
    W(\lambda)
    = - \frac{B(\lambda)}{A(\lambda)}
    = \sum_{i=1}^K \frac{a_i}{\lambda - \lambda_i}
    ,
  \end{equation*}
  is
  \begin{equation*}
    \frac{\partial W}{\partial t}(\lambda)
    = \frac{\partial}{\partial t}\left( - \frac{B(\lambda)}{A(\lambda)}\right)
    = -\frac{\frac{\partial B}{\partial t}(\lambda)}{A(\lambda)} 
    = \frac{A(\lambda) \, M_+ - B(\lambda)}{\lambda \, A(\lambda)}
    = \frac{M_+ + W(\lambda)}{\lambda}
    ,
  \end{equation*}
  and taking residues of both sides of this equality
  at $\lambda = \lambda_i$ we obtain
  \begin{equation*}
    \frac{d a_i}{dt} = \frac{a_i}{\lambda_i}
    .
  \end{equation*}
  (As an aside, the residue at $\lambda=0$ is zero, since
  $W(0)=-M_+$; see equation (B.16) in~\cite{lundmark-szmigielski:GX-inverse-problem}.)

  An entirely similar computation, using the other pair of Lax equations
  \eqref{eq:laxII-x} and \eqref{eq:laxII-t},
  shows that
  \begin{equation*}
    \frac{\partial \twin A}{\partial t}(\lambda) = 0
    ,\qquad
    \frac{\partial \twin B}{\partial t}(\lambda) = \frac{\twin B(\lambda) - \twin A(\lambda) \, N_+}{\lambda}
    ,\qquad
    \frac{\partial \twin C}{\partial t}(\lambda) = \frac{2 \bigl( \twin B(\lambda) - \twin A(\lambda) \, N_+ \bigr) \, M_+}{\lambda}
    .
  \end{equation*}
  It follows that
  $\mu_1,\dots,\mu_{K-1}$ are constant in time,
  and that
  \begin{equation*}
    \twin W(\lambda)
    = - \frac{\twin B(\lambda)}{\twin A(\lambda)}
    = -b_{\infty} + \sum_{j=1}^{K-1} \frac{b_j}{\lambda - \mu_j}
  \end{equation*}
  satisfies
  \begin{equation*}
    \frac{\partial \twin W}{\partial t}(\lambda)
    = \frac{N_+ + \twin W(\lambda)}{\lambda}
    ,
  \end{equation*}
  so that
  $\frac{d b_j}{dt} = \frac{b_j}{\mu_j}$.
  Taking the limit $\lambda \to \infty$,
  we may also deduce from this that
  $\frac{d b_{\infty}}{dt} = 0$,
  i.e., $b_{\infty}$ is a constant of motion.
  However, this was already noticed in \autoref{rem:b-infty-and-star-constants-of-motion},
  where we also saw that $b^*_{\infty}$ is a constant of motion.
\end{proof}

\begin{corollary}[Solution formulas for interlacing Geng--Xue peakons]
  \label{cor:peakon-solution-formulas}
  The formulas in \autoref{thm:inverse-spectral-map-formulas}
  give the general solution to the peakon ODEs
  \eqref{eq:GX-peakon-ode} in the interlacing $K+K$ case
  (with $K \ge 2$, and all amplitudes $m_{2a-1}$ and $n_{2a}$ positive),
  if we let the parameters
  \begin{equation*}
    \{ \lambda_i, \mu_j, b_{\infty}, b^*_{\infty} \}
  \end{equation*}
  be constants with
  \begin{equation*}
    0 < \lambda_1 < \dotsb < \lambda_K
    ,\qquad
    0 < \mu_1 < \dotsb < \mu_{K-1}
    ,\qquad
    b_{\infty} > 0
    ,\qquad
    b^*_{\infty} > 0
    ,
  \end{equation*}
  and let $\{ a_i, b_j \}$ have the time dependence
  \begin{equation}
    \label{eq:residues-time-dependence-again}
    a_i(t) = a_i(0) \, e^{t/\lambda_i}
    ,\qquad
    b_j(t) = b_j(0) \, e^{t/\mu_j}
    ,
  \end{equation}
  where $\{ a_i(0), b_j(0) \}$ are positive constants.
  These solutions are globally defined, and satisfy
  $x_1(t) < \dots < x_{2K}(t)$ for all~$t \in \R$.
\end{corollary}

\begin{remark}
  The solution for the case $K=1$
  is derived in \autoref{sec:dynamics-1} below.
\end{remark}

\section{Dynamics of $1+1$ peakon solutions}
\label{sec:dynamics-1}

With all the preliminaries out of the way,
we can finally begin analyzing the properties of
the $K+K$ interlacing peakon solutions of the Geng--Xue equation.
The governing ODEs are \eqref{eq:GX-peakon-ode},
and we have explicit formulas for the general solution for any~$K$,
as described in \autoref{cor:peakon-solution-formulas}.
However, these formulas are fairly involved,
so we will warm up by first looking at the case $K=1$
(which is somewhat exceptional) in this section,
and the case $K=2$ in \autoref{sec:dynamics-2}.
The general case $K \ge 2$ will be treated in \autoref{sec:dynamics-K}.

The ODEs governing $1+1$ peakon solutions
\begin{equation*}
  u(x,t) = m_1(t) \, e^{-\abs{x-x_1(t)}}
  ,\qquad
  v(x,t) = n_2(t) \, e^{-\abs{x-x_2(t)}}
\end{equation*}
with $x_1(t) < x_2(t)$
are
\begin{equation}
  \label{eq:GX-peakon-ode-K1}
  \dot x_1 = \dot x_2
  = \frac{\dot m_1}{m_1} = - \frac{\dot n_2}{n_2}
  = m_1 n_2 E_{12}
  ,
\end{equation}
where $E_{12} = e^{x_1-x_2}$.
It is immediately verified by differentiation that
$m_1 n_2 E_{12}$ is a constant of motion;
denote its value by~$c$.
(If we impose the pure peakon assumption that $m_1>0$ and $n_2>0$,
then obviously $c>0$.)
Direct integration then gives the solution
\begin{equation}
  \label{eq:solution-K1}
  \begin{aligned}
    x_1(t) &= x_1(0) + ct
    ,\\
    x_2(t) &= x_2(0) + ct
    ,\\
    m_1(t) &= m_1(0) \, e^{ct}
    ,\\
    n_2(t) &= n_2(0) \, e^{-ct}
    ,
  \end{aligned}
\end{equation}
with constants
\begin{equation*}
  x_1(0) < x_2(0)
  ,\qquad
  m_1(0) > 0
  ,\qquad
  n_2(0) > 0
  .
\end{equation*}
So the two peakons travel with the same (constant) velocity
\begin{equation*}
  c = m_1(0) \, m_2(0) \, e^{x_1(0)-x_2(0)}
  .
\end{equation*}
As $t \to \infty$, the amplitude~$m_1$ of the left peakon tends to infinity
and the amplitude~$n_2$ of the right peakon tends to zero,
in such a way that the product $m_1 n_2$ stays constant,
and the other way around as $t \to -\infty$.

The actual peakon wave profiles are
\begin{equation*}
  u(x,t) = m_1(t) \, e^{-\abs{x-x_1(t)}}
  =
  \begin{cases}
    m_1(0) \, e^{x - x_1(0)}
    ,&
    x \le x_1(0) + ct
    ,
    \\
    m_1(0) \, e^{2ct + x_1(0) - x}
    ,&
    x \ge x_1(0) + ct
    ,
  \end{cases}
\end{equation*}
and
\begin{equation*}
  v(x,t) = n_2(t) \, e^{-\abs{x-x_2(t)}}
  =
  \begin{cases}
    n_2(0) \, e^{-2ct + x - x_2(0)}
    ,&
    x \le x_2(0) + ct
    ,
    \\
    n_2(0) \, e^{x_2(0) - x}
    ,&
    x \ge x_2(0) + ct
    .
  \end{cases}
\end{equation*}
This means that the function $u(x,t) \, v(x,t)$ will
be a stump-shaped travelling wave with velocity~$c$:
\begin{equation*}
  u(x,t) \, v(x,t)
  =
  \begin{cases}
    c \, e^{2(x-ct - x_1(0))}
    ,&
    x \le x_1(0) + ct
    ,
    \\
    c
    ,&
    x_1(0) + ct \le x \le x_2(0) + ct
    ,
    \\
    c \, e^{-2(x-ct - x_2(0))}
    ,&
    x \ge x_2(0) + ct
    .
  \end{cases}
\end{equation*}

\begin{remark}
  \label{rem:1+1-peakon-antipeakon}

  Here in the case $K=1$, it is not really necessary to assume that $m_1$ and~$n_2$
  are positive. If we allow negative amplitudes (antipeakons),
  the solution will still be given by the same formulas and it will be
  globally defined; the only difference is that the constant of motion
  $c = m_1 n_2 e^{x_1-x_2}$ will be negative if $m_1$ and~$n_2$
  have opposite signs.

  However, for $K \ge 2$, mixed peakon--antipeakon solutions
  may exhibit collisions and finite-time blowup; see
  \autoref{sec:2+2-peakon-antipeakon}.
\end{remark}

\begin{remark}
  For Camassa--Holm and Degasperis--Procesi peakons,
  the simplest integral of motion is $\sum_{i=1}^N m_i$,
  corresponding to the conserved quantity
  $\int_{\R} m \, dx$.
  For Novikov peakons, the simplest integral of motion is
  \begin{equation*}
    \sum_{i,j=1}^N m_i m_j E_{ij}
    = \sum_{i,j=1}^N m_i m_j e^{-\abs{x_i-x_j}}
    ,
  \end{equation*}
  and Geng--Xue peakons admit two analogous integrals of motion:
  \begin{equation*}
    \begin{split}
      {\mathcal M} &=
      \sum_{1 \le i<j \le N} m_i n_j E_{ij}
      \\ &
      = m_1 n_2 E_{12} + m_1 n_4 E_{14} + \dots + m_1 n_{2K} E_{1,2K}
      + \dots + m_{2K-1} n_{2K} E_{2K-1,2K}
    \end{split}
  \end{equation*}
  and
  \begin{equation*}
    \begin{split}
      \twin{\mathcal M} &=
      \sum_{1 < j < i < N} n_j m_i E_{ji}
      \\ &
      = n_2 m_3 E_{23} + n_2 m_5 E_{25} + \dots + n_2 m_{2K-1} E_{2,2K -1}
      + \dots + n_{2K-2} m_{2K-1} E_{2K-2,2K-1}
      ,
    \end{split}
  \end{equation*}
  In the case $K=1$, clearly $\mathcal{M}$ reduces to the constant of
  motion $m_1 n_2 E_{12}$ that we saw above, while
  $\twin{\mathcal M}$ is identically zero.
\end{remark}

\begin{remark}
  \label{rem:solution-K1-spectral}
  As we just saw, the $1+1$ peakon solution is easy to find
  directly, but it is reassuring to check that the inverse spectral map
  from \autoref{sec:map-K1} works correctly here as well.
  Inserting the time dependence $a_1(t) = a_1(0) \, e^{t/\lambda_1}$ into
  \eqref{eq:peakon-recovery-K1} yields
  \begin{equation}
    \label{eq:peakon-solution-K1-spectral}
    \begin{aligned}
      x_1(t) &= \frac{t}{2\lambda_1} +\frac12 \ln \frac{a_1(0)}{\lambda_1 \, b_{\infty}^*}
      ,\\
      x_2(t) &= \frac{t}{2\lambda_1} +\frac12 \ln 2\, a_1(0) \, b_{\infty}
      ,\\
      m_1(t) &= \sqrt{\frac{a_1(0) \, b_{\infty}^*}{\lambda_1}} \, \exp\left( \frac{t}{2\lambda_1} \right)
      ,\\
      n_2(t) &= \sqrt{\frac{b_{\infty}}{2\, a_1(0)} } \, \exp\left( -\frac{t}{2\lambda_1} \right)
      ,
    \end{aligned}
  \end{equation}
  which is just another way of writing the
  solution~\eqref{eq:solution-K1};
  in particular, the velocity~$c$ corresponds to~$(2 \lambda_1)^{-1}$.
  The distance
  \begin{equation*}
    x_2(t) - x_1(t) = \tfrac12 \ln 2\, \lambda_1 \, b_{\infty} \, b_{\infty}^*
  \end{equation*}
  is positive due to the constraint~\eqref{eq:nonlinear-constraint-K1}.
\end{remark}

\section{Dynamics of $1+1$ shockpeakon solutions}
\label{sec:dynamics-shock-1}

We will continue the study of interlacing peakon solutions
in \autoref{sec:dynamics-2} below,
but first we will show how to integrate the $1+1$ shockpeakon ODEs
which were given in \autoref{ex:shockpeakons-1+1},
and are repeated here for convenience:
\begin{equation}
  \label{eq:shockpeakon-odes-1+1}
  \begin{aligned}
    \dot x_1 &= m_1 (n_2+r_2) E_{12}
    ,\\
    \dot x_2 &= (m_1-s_1) n_2 E_{12}
    ,\\
    \dot m_1 &= (m_1^2 - m_1 s_1 + s_1^2) (n_2+r_2) E_{12}
    ,\\
    \dot n_2 &= -(m_1-s_1) (n_2^2 + n_2 r_2 + r_2^2) E_{12}
    ,\\
    \dot s_1 &= s_1 (2m_1-s_1) (n_2+r_2) E_{12}
    ,\\
    \dot r_2 &= -r_2 (m_1-s_1) (2n_2+r_2) E_{12}
    ,
  \end{aligned}
\end{equation}
where $E_{12} = e^{-\abs{x_1-x_2}} = e^{x_1-x_2}$.
We will restrict ourselves to looking for solutions with
$s_1 \ge 0$ and $r_2 \ge 0$
(cf.~\autoref{rem:entropy-condition}).
Also recall that in writing down the ODEs,
we assumed that $x_1(t) < x_2(t)$.
If we take initial data $x_1(0) < x_2(0)$, this assumption will at least
continue to hold in some neighbourhood of $t=0$,
but in general not globally (see \autoref{rem:overlapping-at-collision});
the solution formulas derived below will only be valid until the time of the
first collision $x_1(t) = x_2(t)$.

To begin with, it is straightforward to verify, simply by differentiating
the expressions with respect to~$t$
and using the ODEs~\eqref{eq:shockpeakon-odes-1+1},
that the quantities
\begin{equation}
  K = (m_1-s_1)(n_2+r_2) E_{12}
  ,\qquad
  M = (m_1+s_1) e^{-x_1}
  ,\qquad
  N = (n_2-r_2) e^{x_2}
\end{equation}
are constants of motion,
and that the quantities
\begin{equation}
  S = s_1 e^{-x_1}
  ,\qquad
  R = r_2 e^{x_2}
\end{equation}
have the time dependence
\begin{equation*}
  \dot S = KS
  ,\qquad
  \dot R = -KR
  .
\end{equation*}
Since $K$ is constant, this means that
\begin{equation}
  S(t) = S(0) \, e^{Kt}
  ,\qquad
  R(t) = R(0) \, e^{-Kt}
  .
\end{equation}
In the same way one may check that the quantities
\begin{equation}
  X = (m_1-s_1) e^{x_1}
  ,\qquad
  Y = (n_2+r_2) e^{-x_2}
\end{equation}
satisfy
\begin{equation*}
  \dot X = 2KX
  ,\qquad
  \dot Y = -2KY
  ,
\end{equation*}
so that
\begin{equation}
  X(t) = X(0) \, e^{2Kt}
  ,\qquad
  Y(t) = Y(0) \, e^{-2Kt}
  .
\end{equation}

\begin{subequations}
  \label{eq:shock-solution-generic}
  The fact that $X(t)$ keeps its sign
  means that if $m_1(0) - s_1(0) \neq 0$, then $m_1(t) - s_1(t) \neq 0$
  for all~$t$ such that the solution remains valid,
  i.e., up until the first collision.
  For these~$t$ we then get
  \begin{equation*}
    \begin{split}
      e^{2 x_1(t)}
      &
      = \frac{\bigl( m_1(t) - s_1(t) \bigr) \, e^{x_1(t)}}{\bigl( m_1(t) - s_1(t) \bigr) \, e^{-x_1(t)}}
      = \frac{X(t)}{M - 2 S(t)}
      \\ &
      = \frac{\bigl( m_1(0) - s_1(0) \bigr) \, e^{x_1(0)} \, e^{2Kt}}{\bigl( m_1(0) + s_1(0) \bigr) \, e^{-x_1(0)} - 2 s_1(0) \, e^{-x_1(0)} \, e^{Kt}}
      \\ &
      = \frac{e^{2 x_1(0)} \, \bigl( m_1(0) - s_1(0) \bigr) \, e^{Kt}}{\bigl( m_1(0) + s_1(0) \bigr) \, e^{-Kt} - 2 s_1(0)}
      ,
    \end{split}
  \end{equation*}
  and hence
  \begin{equation}
    \label{eq:shock-solution-x1}
    x_1(t)
    = x_1(0)
    + \frac{Kt}{2}
    + \frac12 \ln\frac{m_1(0) - s_1(0)}{\bigl( m_1(0) + s_1(0) \bigr) \, e^{-Kt} - 2 s_1(0)}
    .
  \end{equation}
  Similarly, if $n_2(0) + r_2(0) \neq 0$, we find
  \begin{equation*}
    e^{-2 x_2(t)}
    = \frac{\bigl( n_2(t) + r_2(t) \bigr) \, e^{-x_2(t)}}{\bigl( n_2(t) + r_2(t) \bigr) \, e^{x_2(t)}}
    = \frac{Y(t)}{N + 2 R(t)}
    = \frac{e^{-2 x_2(0)} \, \bigl( n_2(0) + r_2(0) \bigr) \, e^{-Kt}}{\bigl( n_2(0) - r_2(0) \bigr) \, e^{Kt} + 2 r_2(0)}
    ,
  \end{equation*}
  so that
  \begin{equation}
    \label{eq:shock-solution-x2}
    x_2 (t)
    = x_2(0)
    + \frac{Kt}{2}
    - \frac12 \ln\frac{n_2(0) + r_2(0)}{\bigl( n_2(0) - r_2(0) \bigr) \, e^{Kt} + 2 r_2(0)}
    .
  \end{equation}
  From this we obtain
  \begin{equation}
    \label{eq:shock-solution-s1}
    \begin{split}
      s_1(t)
      &
      = S(t) \, e^{x_1(t)}
      = s_1(0) \, e^{-x_1(0)} \, e^{Kt} \, e^{x_1(t)}
      \\ &
      = s_1(0) \,  \, e^{3Kt/2}
      \sqrt{\frac{\bigl( m_1(0) - s_1(0) \bigr)}{\bigl( m_1(0) + s_1(0) \bigr) \, e^{-Kt} - 2 s_1(0)}}
    \end{split}
  \end{equation}
  and
  \begin{equation}
    \label{eq:shock-solution-r2}
    \begin{split}
      r_2(t)
      &
      = R(t) \, e^{-x_2(t)}
      = r_2(0) \, e^{x_2(0)} \, e^{-Kt} \, e^{-x_2(t)}
      \\ &
      = r_2(0) \,  \, e^{-3Kt/2}
      \sqrt{\frac{\bigl( n_2(0) + r_2(0) \bigr)}{\bigl( n_2(0) - r_2(0) \bigr) \, e^{Kt} + 2 r_2(0)}}
      .
    \end{split}
  \end{equation}
  Finally,
  \begin{equation}
    \label{eq:shock-solution-m1}
    \begin{split}
      m_1(t)
      &
      = M \, e^{x_1(t)} - s_1(t)
      = \bigl( m_1(0) + s_1(0) \bigr) \, e^{-x_1(0)} \, e^{x_1(t)} - s_1(t)
      \\ &
      = 
      \Bigl( \bigl( m_1(0) + s_1(0) \bigr) e^{Kt/2}
      - s_1(0) \,  \, e^{3Kt/2} \Bigr)
      \sqrt{\frac{\bigl( m_1(0) - s_1(0) \bigr)}{\bigl( m_1(0) + s_1(0) \bigr) \, e^{-Kt} - 2 s_1(0)}}
    \end{split}
  \end{equation}
  and
  \begin{equation}
    \label{eq:shock-solution-n2}
    \begin{split}
      n_2(t)
      &
      = N \, e^{-x_2(t)} + r_2(t)
      = \bigl( n_2(0) - r_2(0) \bigr) \, e^{x_2(0)} \, e^{-x_2(t)} + r_2(t)
      \\ &
      = 
      \Bigl( \bigl( n_2(0) - r_2(0) \bigr) e^{-Kt/2}
      + r_2(0) \,  \, e^{-3Kt/2} \Bigr)
      \sqrt{\frac{\bigl( n_2(0) + r_2(0) \bigr)}{\bigl( n_2(0) - r_2(0) \bigr) \, e^{Kt} + 2 r_2(0)}}
      .
    \end{split}
  \end{equation}
\end{subequations}
The formulas \eqref{eq:shock-solution-generic}
give the solution of~\eqref{eq:GX-shockpeakon-ode} in the generic case
$m_1(0) - s_1(0) \neq 0$ and $n_2(0) + r_2(0) \neq 0$.

\begin{subequations}
  \label{eq:shock-solution-caseA}
  If $m_1(0)-s_1(0)=0$,
  we get instead $X(0)=0$,
  hence $X(t)=0$, meaning that $m_1(t) = s_1(t)$ for all~$t$.
  According to the ODEs~\eqref{eq:GX-shockpeakon-ode},
  this immediately implies that $\dot x_2 = \dot n_2 = \dot r_2 = 0$,
  so the second shockpeakon doesn't move or change at all.
  For the first shockpeakon we then have
  \begin{equation*}
    \begin{aligned}
      \dot x_1(t) &= m_1(t) \, \bigl( n_2(0) + r_2(0) \bigr) e^{x_1(t)} e^{-x_2(0)}
      = Y(0) \, m_1(t) \, e^{x_1(t)}
      ,\\
      \dot m_1(t) &= m_1(t)^2 \, \bigl( n_2(0) + r_2(0) \bigr) e^{x_1(t)} e^{-x_2(0)}
      = Y(0) \, m_1(t)^2 \, e^{x_1(t)}
      .
    \end{aligned}
  \end{equation*}
  With $m_1=s_1$,
  the constant of motion $M$ reduces to $M = 2 m_1 \, e^{-x_1}$,
  which is positive because we are assuming that $s_1 \ge 0$
  and $(m_1,s_1) \neq (0,0)$.
  Thus,
  \begin{equation*}
    \dot m_1(t)
    = \frac{2 Y(0)}{M} \, m_1(t)^3
    = \frac{A}{2 m_1(0)^2} \, m_1(t)^3
    ,
  \end{equation*}
  where
  \begin{equation}
    \label{eq:shock-solution-A-caseA}
    A
    = \frac{4 Y(0) \, m_1(0)^2}{M}
    = 2 m_1(0) \, \bigl( n_2(0) + r_2(0) \bigr) \, e^{x_1(0)-x_2(0)}
    ,
  \end{equation}
  which gives
  \begin{equation}
    \label{eq:shock-solution-m1-s1-caseA}
    m_1(t) = s_1(t)
    = \frac{m_1(0)}{\sqrt{1 - A t}}
    ,
  \end{equation}
  and consequently
  \begin{equation*}
    e^{x_1(t)}
    = \frac{2 m_1(t)}{M}
    = \frac{2 m_1(t)}{2 m_1(0) \, e^{-x_1(0)}}
    = \frac{e^{x_1(0)}}{\sqrt{1 - A t}}
    ,
  \end{equation*}
  so that
  \begin{equation}
    \label{eq:shock-solution-x1-caseA}
    x_1(t) = x_1(0) - \tfrac12 \ln(1 - A t)
    .
  \end{equation}
\end{subequations}

\begin{subequations}
  \label{eq:shock-solution-caseB}
  In the case $n_2(0)+r_2(0)=0$,
  a similar computation shows that $\dot x_1 = \dot m_1 = \dot s_1 = 0$,
  \begin{equation}
    \label{eq:shock-solution-n2-r2-caseB}
    n_2(t) = -r_2(t) = \frac{n_2(0)}{\sqrt{1 + Bt}}
    ,
  \end{equation}
  and
  \begin{equation}
    \label{eq:shock-solution-x2-caseB}
    x_2(t) = x_2(0) + \tfrac12 \ln(1 + B t)
  \end{equation}
  where
  \begin{equation}
    \label{eq:shock-solution-B-caseB}
    B = 2 \bigl( m_1(0) - s_1(0) \bigr) \, n_2(0) \, e^{x_1(0)-x_2(0)}
    .
  \end{equation}
\end{subequations}

\begin{remark}
  Note that if both conditions
  $m_1(0)-s_1(0)=0$ and $n_2(0)+r_2(0)=0$ hold at the same time,
  then the solution is completely time-independent:
  \begin{equation*}
    \dot x_1 = \dot m_1 = \dot s_1 = \dot x_2 = \dot n_2 = \dot r_2 = 0
    .
  \end{equation*}
\end{remark}

\begin{remark}
  We may also point out that $\dot x_2 = K$ if $s_1=0$,
  i.e., if the first shockpeakon is in fact an ordinary peakon,
  then the second shockpeakon travels with constant speed.
  Similarly, $\dot x_1 = K$ if $r_2=0$.
  As a special case, when $s_1=r_2=0$, we recover the result from the
  \autoref{sec:dynamics-1} that in the $1+1$ peakon solution,
  both peakons travel with equal and constant speed.
\end{remark}

\begin{remark}
  \label{rem:overlapping-at-collision}
  The formulas above show
  that in many cases there will be a collision $x_1(t)=x_2(t)$
  after finite time,
  with $(m_1,s_1) \neq (0,0)$ and $(n_2,r_2) \neq (0,0)$
  at the instant of collision.
  For example, if $B < 0$ in \eqref{eq:shock-solution-x2-caseB},
  then $x_2(t) \to -\infty$ as $t \to (1/B)^-$,
  so by continuity there must be a collision time $t_0 \in (0,1/B)$
  such that $x_2(t_0)$ equals the constant $x_1(t)=x_1(0)=x_1(t_0)$.
  Thus, the non-overlapping condition is \emph{not} preserved
  at shockpeakon--shockpeakon collisions.
  (This is in contrast to the peakon--antipeakon collision scenario
  described in~\autoref{sec:2+2-peakon-antipeakon}.)
  It is not clear to us at present if, and in that case how,
  such a solution can be continued past the collision.
\end{remark}

\section{Dynamics of $2+2$ interlacing pure peakon solutions}
\label{sec:dynamics-2}

In this section we leave shockpeakons and return to the interlacing
pure peakon solutions.
We take a detailed look at the case $K=2$,
as a preparation for the general case $K \ge 2$ treated in
\autoref{sec:dynamics-K}.
Except for \autoref{sec:2+2-peakon-antipeakon},
we will only consider pure peakon solutions,
i.e., we will assume that the amplitudes $m_1$, $n_2$, $m_3$, $n_4$
are all positive.

\subsection{Governing ODEs and explicit solution formulas}

The $2+2$ interlacing peakon solutions of the Geng--Xue equation
are governed by the ODEs
\begin{equation}
  \label{eq:GX-peakon-ode-K2}
  \begin{split}
    \dot x_1 &= (m_1 + m_3 E_{13}) (n_2 E_{12} + n_4 E_{14})
    , \\
    \dot x_2 &= (m_1 E_{12} + m_3 E_{23}) (n_2 + n_4 E_{24})
    , \\
    \dot x_3 &= (m_1 E_{13} + m_3) (n_2 E_{23} + n_4 E_{34})
    , \\
    \dot x_4 &= (m_1 E_{14} + m_3 E_{34}) (n_2 E_{24} + n_4)
    , \\
    \frac{\dot m_1}{m_1} &= (m_1 + m_3 E_{13}) (n_2 E_{12} + n_4 E_{14}) - 2 m_3 E_{13} (n_2 E_{12} + n_4 E_{14})
    , \\
    \frac{\dot n_2}{n_2} &= (-m_1 E_{12} + m_3 E_{23}) (n_2 + n_4 E_{24}) - 2 (m_1 E_{12} + m_3 E_{23}) n_4 E_{24}
    , \\
    \frac{\dot m_3}{m_3} &= (m_1 E_{13} + m_3) (-n_2 E_{23} + n_4 E_{34}) + 2 m_1 E_{13} (n_2 E_{23} + n_4 E_{34})
    , \\
    \frac{\dot n_4}{n_4} &= (-m_1 E_{14} - m_3 E_{34}) (n_2 E_{24} + n_4) + 2 (m_1 E_{14} + m_3 E_{34}) n_2 E_{24}
    ,
  \end{split}
\end{equation}
where $E_{ij} = e^{-\abs{x_i-x_j}} = e^{x_i-x_j}$ for $i<j$.
It is apparent that for $K=2$ the equations are already much more
complicated than for $K=1$;
cf. equation~\eqref{eq:GX-peakon-ode-K1}.
There seems to be little hope of integrating this system
by any direct methods,
so we proceed without further ado to the general solution formulas
derived by inverse spectral methods (\autoref{cor:peakon-solution-formulas}).
These formulas are stated in terms of
the sums $\heineintegral_{nm}^{rs}$ defined in
\autoref{sec:inverse-spectral-map},
but here in the case $K=2$
the expressions are still small enough to
allow us to write out in detail
what these sums actually are.
Then the sought quantities
\begin{equation*}
  x_1(t), \, x_2(t), \, x_3(t), \, x_4(t), \,
  m_1(t), \, n_2(t), \, m_3(t), \, n_4(t)
\end{equation*}
will be directly expressed in terms of the spectral variables
\begin{equation*}
  \lambda_1, \, \lambda_2, \, \mu_1, \,
  a_1, \, a_2, \, b_1, \, b_{\infty}, \, b_{\infty}^*,
\end{equation*}
where $\lambda_1$, $\lambda_2$, $\mu_1$, $b_{\infty}$ and~$b_{\infty}^*$
are arbitrary positive constants with $\lambda_1 < \lambda_2$,
and where $a_1$, $a_2$ and~$b_1$ have the time dependence
\begin{equation}
  \label{eq:solution-K2-timedependence}
  a_1(t) = a_1(0) \, e^{t/\lambda_1}
  ,\qquad
  a_2(t) = a_2(0) \, e^{t/\lambda_2}
  ,\qquad
  b_1(t) = b_1(0) \, e^{t/\mu_1}
  ,
\end{equation}
with arbitrary positive constants $a_1(0)$, $a_2(0)$ and~$b_1(0)$.
(This time dependence will mostly be suppressed in the notation;
whenever we write just $a_k$, we mean $a_k(t)$.)

In terms of the quantities from \autoref{rem:solution-alternative-form},
the general solution of \eqref{eq:GX-peakon-ode-K2} is then given
(completely explicitly in terms of elementary functions) by
\begin{subequations}
  \label{eq:solution-K2}
\begin{equation}
  \label{eq:solution-K2-position}
  \begin{split}
    \tfrac12 e^{2 x_4} &= I_{00} + b_{\infty} \alpha_0 = \frac{a_1 b_1}{\lambda_1+\mu_1} + \frac{a_2 b_1}{\lambda_2+\mu_1} + b_{\infty} (a_1 + a_2),\\
    \tfrac12 e^{2 x_3} &= \frac{\heineintegral_{11}^{00}}{\heineintegral_{00}^{11}} = \frac{I_{00}}{1} = \frac{a_1 b_1}{\lambda_1+\mu_1} + \frac{a_2 b_1}{\lambda_2+\mu_1},\\
    \tfrac12 e^{2 x_2} &= \frac{\heineintegral_{21}^{00}}{\heineintegral_{10}^{11}} = \frac{\dfrac{\bigl( \lambda_1-\lambda_2 \bigr)^2}{(\lambda_1+\mu_1)(\lambda_2+\mu_1)} a_1 a_2 b_1}{\lambda_1 a_1 + \lambda_2 a_2},\\
    \tfrac12 e^{2 x_1} &= \frac{\heineintegral_{21}^{00}}{\heineintegral_{10}^{11} + \dfrac{\strut 2 b^*_{\infty} L}{M} \, \heineintegral_{11}^{10}} \\ &= \frac{\dfrac{\bigl( \lambda_1-\lambda_2 \bigr)^2}{(\lambda_1+\mu_1)(\lambda_2+\mu_1)} a_1 a_2 b_1}{\lambda_1 a_1 + \lambda_2 a_2 + \dfrac{2 b_{\infty}^* \lambda_1 \lambda_2}{\mu_1} \left( \dfrac{\lambda_1 a_1 b_1}{\lambda_1+\mu_1} + \dfrac{\lambda_2 a_2 b_1}{\lambda_2+\mu_1} \right)}
  \end{split}
\end{equation}
and
\begin{equation}
  \label{eq:solution-K2-amplitude}
  \begin{split}
    2 n_4 e^{-x_4} &= \frac{1}{\alpha_0} = \frac{1}{a_1 + a_2}, \\
    2 m_3 e^{-x_3} &= \frac{\heineintegral_{00}^{11} \heineintegral_{10}^{01}}{\heineintegral_{11}^{10} \heineintegral_{00}^{10}} = \frac{1 \cdot \heineintegral_{10}^{01}}{\heineintegral_{11}^{10} \cdot 1} = \frac{a_1 + a_2}{\dfrac{\lambda_1 a_1 b_1}{\lambda_1+\mu_1} + \dfrac{\lambda_2 a_2 b_1}{\lambda_2+\mu_1}}, \\
    2 n_2 e^{-x_2} &= \frac{\heineintegral_{10}^{11} \heineintegral_{11}^{10}}{\heineintegral_{10}^{01} \heineintegral_{21}^{01}} = \frac{\left( \lambda_1 a_1 + \lambda_2 a_2 \right) \left( \dfrac{\lambda_1 a_1 b_1}{\lambda_1+\mu_1} + \dfrac{\lambda_2 a_2 b_1}{\lambda_2+\mu_1} \right)}{\left( a_1 + a_2 \right) \dfrac{\mu_1 \bigl( \lambda_1-\lambda_2 \bigr)^2}{(\lambda_1+\mu_1)(\lambda_2+\mu_1)} a_1 a_2 b_1}, \\
    2 m_1 e^{-x_1} &= \frac{M \, \heineintegral_{10}^{11}}{L \, \heineintegral_{11}^{10}} + 2 b^*_{\infty} = \frac{\mu_1 \bigl( \lambda_1 a_1 + \lambda_2 a_2 \bigr)}{\lambda_1 \lambda_2 \left( \dfrac{\lambda_1 a_1 b_1}{\lambda_1+\mu_1} + \dfrac{\lambda_2 a_2 b_1}{\lambda_2+\mu_1} \right)} + 2 b_{\infty}^*.
  \end{split}
\end{equation}
\end{subequations}

\subsection{Asymptotics as $t \to +\infty$}

We will now derive formulas for the large time
asymptotics of the $2+2$ interlacing peakon solution~\eqref{eq:solution-K2}.
We remind the reader
that these features were illustrated with graphics
in \autoref{ex:GX-interlacing-2+2},
and it may be helpful to revisit that example at this point.

We begin with the case $t \to +\infty$.
In this case,
the factors $a_1(t)$, $a_2(t)$ and~$b_1(t)$ will all diverge to~$+\infty$,
and $a_1(t)$ will grow much faster than $a_2(t)$
since we are assuming that $0 < \lambda_1 < \lambda_2$;
indeed, setting
\begin{equation*}
  \delta = \frac{1}{\lambda_1} - \frac{1}{\lambda_2} > 0
\end{equation*}
we have
\begin{equation*}
  \frac{a_2(t)}{a_1(t)}
  = \frac{a_2(0) \, e^{t/\lambda_2}}{a_1(0) \, e^{t/\lambda_1}}
  = \frac{a_2(0)}{a_1(0)} \, e^{-\delta t}
  \to 0
  ,\qquad
  t \to +\infty
  .
\end{equation*}
Factoring out the dominant terms, we get
\begin{equation}
  \label{eq:solution-K2-position-factor-out}
  \begin{split}
    \tfrac12 e^{2 x_4} &=
    a_1 b_1 \Biggl(
    \frac{1}{\lambda_1+\mu_1} + \frac{\tfrac{a_2}{a_1}}{\lambda_2+\mu_1}
    + b_{\infty} \bigl( \tfrac{1}{b_1} + \tfrac{a_1}{a_2 b_1} \bigr)
    \Biggr)
    ,\\
    \tfrac12 e^{2 x_3} &=
    a_1 b_1 \Biggl(
    \frac{1}{\lambda_1+\mu_1} + \frac{\tfrac{a_2}{a_1}}{\lambda_2+\mu_1}
    \Biggr)
    ,\\
    \tfrac12 e^{2 x_2} &=
    a_2 b_1 \cdot
    \frac{\dfrac{\bigl( \lambda_1-\lambda_2 \bigr)^2}{(\lambda_1+\mu_1)(\lambda_2+\mu_1)}}{\lambda_1 + \lambda_2 \tfrac{a_2}{a_1}}
    ,\\
    \tfrac12 e^{2 x_1} &=
    a_2 \cdot
    \frac{\dfrac{\bigl( \lambda_1-\lambda_2 \bigr)^2}{(\lambda_1+\mu_1)(\lambda_2+\mu_1)}}{\lambda_1 \tfrac{1}{b_1} + \lambda_2 \tfrac{a_2}{a_1 b_1} + \dfrac{2 b_{\infty}^* \lambda_1 \lambda_2}{\mu_1} \left( \dfrac{\lambda_1}{\lambda_1+\mu_1} + \dfrac{\lambda_2 \frac{a_2}{a_1}}{\lambda_2+\mu_1} \right)}
  \end{split}
\end{equation}
and
\begin{equation}
  \label{eq:solution-K2-amplitude-factor-out}
  \begin{split}
    2 n_4 e^{-x_4} &=
    \frac{1}{a_1} \cdot \frac{1}{1 + \tfrac{a_2}{a_1}}
    , \\
    2 m_3 e^{-x_3} &=
    \frac{1}{b_1} \cdot
    \frac{1 + \tfrac{a_2}{a_1}}{\dfrac{\lambda_1}{\lambda_1+\mu_1} + \dfrac{\lambda_2 \tfrac{a_2}{a_1}}{\lambda_2+\mu_1}}
    , \\
    2 n_2 e^{-x_2} &=
    \frac{1}{a_2} \cdot
    \frac{\left( \lambda_1 + \lambda_2 \tfrac{a_2}{a_1} \right) \left( \dfrac{\lambda_1}{\lambda_1+\mu_1} + \dfrac{\lambda_2 \tfrac{a_2}{a_1}}{\lambda_2+\mu_1} \right)}{\left( 1 + \tfrac{a_2}{a_1} \right) \dfrac{\mu_1 \bigl( \lambda_1-\lambda_2 \bigr)^2}{(\lambda_1+\mu_1)(\lambda_2+\mu_1)}}
    , \\
    2 m_1 e^{-x_1} &=
    \frac{1}{b_1} \cdot 
    \frac{\mu_1 \bigl( \lambda_1 + \lambda_2 \tfrac{a_2}{a_1} \bigr)}{\lambda_1 \lambda_2 \left( \dfrac{\lambda_1}{\lambda_1+\mu_1} + \dfrac{\lambda_2 \tfrac{a_2}{a_1}}{\lambda_2+\mu_1} \right)} + 2 b_{\infty}^*
    .
  \end{split}
\end{equation}
Writing $o(1)$ for terms which tend to zero
(exponentially fast, actually),
we therefore have,
as $t \to +\infty$,
\begin{equation}
  \label{eq:asymptotics-K2-position-plus-infty}
  \begin{split}
    x_4(t) &=
    \frac{t}{2} \left( \frac{1}{\lambda_1} + \frac{1}{\mu_1} \right)
    + \frac12 \ln \frac{2 \, a_1(0) \, b_1(0)}{\lambda_1 + \mu_1} + o(1)
    ,\\
    x_3(t) &=
    \frac{t}{2} \left( \frac{1}{\lambda_1} + \frac{1}{\mu_1} \right)
    + \frac12 \ln \frac{2 \, a_1(0) \, b_1(0)}{\lambda_1 + \mu_1} + o(1)
    ,\\
    x_2(t) &=
    \frac{t}{2} \left( \frac{1}{\lambda_2} + \frac{1}{\mu_1} \right)
    + \frac12 \ln \frac{2 \, a_2(0) \, b_1(0)}{\lambda_2+\mu_1}
    + \frac12 \ln \frac{\bigl( \lambda_1-\lambda_2 \bigr)^2}{\lambda_1 (\lambda_1+\mu_1)} + o(1)
    ,\\
    x_1(t) &=
    \frac{t}{2} \left( \frac{1}{\lambda_2} \right)
    + \frac12 \ln \frac{\mu_1 \, a_2(0)}{\lambda_1 \lambda_2 \, b^*_{\infty}}
    + \frac12 \ln \frac{\bigl( \lambda_1-\lambda_2 \bigr)^2}{\lambda_1 (\lambda_2 + \mu_1)} + o(1)
  \end{split}
\end{equation}
and
\begin{equation}
  \label{eq:asymptotics-K2-amplitude-plus-infty}
  \begin{split}
    n_4(t) &=
    \bigl( 1+o(1) \bigr)
    \,
    \sqrt{\frac{b_1(0)}{2 a_1(0) (\lambda_1 + \mu_1)}}
    \,
    \exp\left( \frac{t}{2} \left( \frac{1}{\mu_1} - \frac{1}{\lambda_1} \right) \right)
    ,\\
    m_3(t) &=
    \bigl( 1+o(1) \bigr)
    \,
    \frac{1}{\lambda_1} \sqrt{\frac{a_1(0) \, (\lambda_1 + \mu_1)}{2 b_1(0)}}
    \,
    \exp\left( \frac{t}{2} \left( \frac{1}{\lambda_1} - \frac{1}{\mu_1} \right) \right)
    ,\\
    n_2(t) &=
    \bigl( 1+o(1) \bigr)
    \,
    \frac{1}{\mu_1} \sqrt{\frac{b_1(0) \, (\lambda_2 + \mu_1) \, \lambda_1^3}{2 a_2(0) \, \bigl( \lambda_1-\lambda_2 \bigr)^2 (\lambda_1 + \mu_1)}}
    \,
    \exp\left( \frac{t}{2} \left( \frac{1}{\mu_1} - \frac{1}{\lambda_2} \right) \right)
    ,\\
    m_1(t) &=
    \bigl( 1+o(1) \bigr)
    \,
    \sqrt{\frac{b^*_{\infty} a_2(0) \, \mu_1 \bigl( \lambda_1-\lambda_2 \bigr)^2}{\lambda_1^2 \lambda_2 (\lambda_2 + \mu_1)}}
    \,
    \exp\left( \frac{t}{2} \left( \frac{1}{\lambda_2} \right) \right)
    .
  \end{split}
\end{equation}
In other words, what \eqref{eq:asymptotics-K2-position-plus-infty}
says is that the peakons asymptotically travel along straight lines
in $(x,t)$ space as $t \to +\infty$,
with the asymptotic velocities
\begin{equation}
  \label{eq:asymptotics-K2-speed-plus-infty}
  \begin{aligned}
    \dot x_1 & \sim
    c_3 =
    \frac{1}{2} \left( \frac{1}{\lambda_2} \right)
    ,\\
    \dot x_2 & \sim
    c_2 =
    \frac{1}{2} \left( \frac{1}{\lambda_2} + \frac{1}{\mu_1} \right)
    ,\\
    \dot x_3, \, \dot x_4 & \sim
    c_1 =
    \frac{1}{2} \left( \frac{1}{\lambda_1} + \frac{1}{\mu_1} \right)
  \end{aligned}
\end{equation}
satisfying
$0 < c_3 < c_2 < c_1$.
So the two rightmost peakons asymptotically have the same velocity~$c_1$,
and in fact we see from \eqref{eq:asymptotics-K2-position-plus-infty}
that $x_4(t) - x_3(t) = o(1)$,
i.e., the distance between them approaches zero as~$t \to +\infty$.

Looking at \eqref{eq:asymptotics-K2-amplitude-plus-infty},
we note that the amplitude of the leftmost peakon
always diverges:
$m_1(t) \to \infty$ as $t \to +\infty$.
The other amplitudes decay exponentially to zero,
or grow exponentially to infinity,
or tend to some positive constant value,
depending on the relative sizes of the eigenvalues $\lambda_i$
and~$\mu_j$.
It is easier to visualize the situation if we take logarithms
in~\eqref{eq:asymptotics-K2-amplitude-plus-infty}:
\begin{equation}
  \label{eq:asymptotics-K2-log-amplitude-plus-infty}
  \begin{split}
    -\ln n_4(t) &=
    \frac{t}{2} \left( \frac{1}{\lambda_1} - \frac{1}{\mu_1} \right)
    - \frac12 \ln \frac{b_1(0)}{2 a_1(0) (\lambda_1 + \mu_1)}
    + o(1)
    ,\\
    \ln m_3(t) &=
    \frac{t}{2} \left( \frac{1}{\lambda_1} - \frac{1}{\mu_1} \right)
    - \ln \lambda_1
    + \frac12 \ln \frac{a_1(0) \, (\lambda_1 + \mu_1)}{2 b_1(0)}
    + o(1)
    ,\\
    -\ln n_2(t) &=
    \frac{t}{2} \left( \frac{1}{\lambda_2} - \frac{1}{\mu_1} \right)
    + \ln \mu_1
    - \frac12 \ln \frac{b_1(0) \, (\lambda_2 + \mu_1)}{2 a_2(0)}
    \\ & \quad
    + \frac12 \ln \frac{\bigl( \lambda_1-\lambda_2 \bigr)^2 (\lambda_1 + \mu_1)}{\lambda_1^3}
    + o(1)
    ,\\
    \ln m_1(t) &=
    \frac{t}{2} \left( \frac{1}{\lambda_2} \right)
    + \frac12 \ln \frac{b^*_{\infty} a_2(0) \, \mu_1}{\lambda_1 \lambda_2}
    + \frac12 \ln \frac{\bigl( \lambda_1-\lambda_2 \bigr)^2}{\lambda_1 (\lambda_2 + \mu_1)}
    + o(1)
    .
  \end{split}
\end{equation}
Setting
\begin{equation}
  \label{eq:logamp-slopes-dk}
  d_1 =
  \frac{1}{2} \left( \frac{1}{\lambda_1} - \frac{1}{\mu_1} \right)
  , \qquad
  d_2 =
  \frac{1}{2} \left( \frac{1}{\lambda_2} - \frac{1}{\mu_1} \right)
  , \qquad
  d_3 =
  \frac{1}{2} \left( \frac{1}{\lambda_2} \right)
  ,
\end{equation}
we thus have
\begin{equation}
  \label{eq:asymptotics-K2-logamp-speed-plus-infty}
  \begin{split}
    -\ln n_4(t) &=
    d_1 t + \text{constant} + o(1)
    ,\\
    \ln m_3(t) &=
    d_1 t + \text{constant} + o(1)
    ,\\
    -\ln n_2(t) &=
    d_2 t + \text{constant} + o(1)
    ,\\
    \ln m_1(t) &=
    d_3 t + \text{constant} + o(1)
    ,
  \end{split}
\end{equation}
so the graphs of the quantities on the left-hand side
approach straight lines as $t \to +\infty$.
Note that $d_3 > 0$ always, but $d_1$ and~$d_2$ may be
positive, negative, or zero.

\subsection{Asymptotics as $t \to -\infty$}

As $t \to -\infty$, we can carry out similar calculations, but
now the roles are reversed; 
the factors $a_1(t)$, $a_2(t)$ and~$b_1(t)$ will all tend to zero,
and $a_1(t)$ does so much faster than~$a_2(t)$\,:
\begin{equation*}
  \frac{a_1(t)}{a_2(t)} =
  \frac{a_1(0)}{a_2(0)} \, e^{\delta t}
  \to 0
  ,\qquad
  t \to -\infty
  .
\end{equation*}
So now the dominant terms to be factored out
are not the same as in
\eqref{eq:solution-K2-position-factor-out}
and~\eqref{eq:solution-K2-amplitude-factor-out};
we get
\begin{equation*}
  \begin{split}
    \tfrac12 e^{2 x_4} &=
    a_2 \Biggl(
    \frac{\tfrac{a_1}{a_2} \, b_1}{\lambda_1+\mu_1}
    + \frac{b_1}{\lambda_2+\mu_1}
    + b_{\infty} \bigl( \tfrac{a_2}{a_1} + 1 \bigr)
    \Biggr)
    ,\\
    \tfrac12 e^{2 x_3} &=
    a_2 b_1 \Biggl(
    \frac{\tfrac{a_1}{a_2}}{\lambda_1+\mu_1}
    + \frac{1}{\lambda_2+\mu_1}
    \Biggr)
    ,
  \end{split}
\end{equation*}
and so on.
Taking this into account, we obtain the following
asymptotics as $t \to -\infty$\,:
\begin{equation}
  \label{eq:asymptotics-K2-position-minus-infty}
  \begin{split}
    x_4(t) &=
    \frac{t}{2} \left( \frac{1}{\lambda_2} \right)
    + \frac12 \ln \bigl( 2 \, a_2(0) \, b_{\infty} \bigr) + o(1)
    ,\\
    x_3(t) &=
    \frac{t}{2} \left( \frac{1}{\lambda_2} + \frac{1}{\mu_1} \right)
    + \frac12 \ln \frac{2 \, a_2(0) \, b_1(0)}{\lambda_2 + \mu_1} + o(1)
    ,\\
    x_2(t) &=
    \frac{t}{2} \left( \frac{1}{\lambda_1} + \frac{1}{\mu_1} \right)
    + \frac12 \ln \frac{2 \, a_1(0) \, b_1(0)}{\lambda_1+\mu_1}
    + \frac12 \ln \frac{\bigl( \lambda_1-\lambda_2 \bigr)^2}{\lambda_2 (\lambda_2+\mu_1)} + o(1)
    ,\\
    x_1(t) &=
    \frac{t}{2} \left( \frac{1}{\lambda_1} + \frac{1}{\mu_1} \right)
    + \frac12 \ln \frac{2 \, a_1(0) \, b_1(0)}{\lambda_1+\mu_1}
    + \frac12 \ln \frac{\bigl( \lambda_1-\lambda_2 \bigr)^2}{\lambda_2 (\lambda_2+\mu_1)} + o(1)
  \end{split}
\end{equation}
and
\begin{equation}
  \label{eq:asymptotics-K2-amplitude-minus-infty}
  \begin{split}
    n_4(t) &=
    \bigl( 1+o(1) \bigr)
    \,
    \sqrt{\frac{b_{\infty}}{2 a_2(0)}}
    \,
    \exp\left( -\frac{t}{2} \left(\frac{1}{\lambda_2} \right) \right)
    ,\\
    m_3(t) &=
    \bigl( 1+o(1) \bigr)
    \,
    \frac{1}{\lambda_2} \sqrt{\frac{a_2(0) \, (\lambda_2 + \mu_1)}{2 b_1(0)}}
    \,
    \exp\left( -\frac{t}{2} \left( \frac{1}{\mu_1} - \frac{1}{\lambda_2} \right) \right)
    ,\\
    n_2(t) &=
    \bigl( 1+o(1) \bigr)
    \,
    \frac{1}{\mu_1} \sqrt{\frac{b_1(0) \, (\lambda_1 + \mu_1) \, \lambda_2^3}{2 a_1(0) \, \bigl( \lambda_1-\lambda_2 \bigr)^2 (\lambda_2 + \mu_1)}}
    \,
    \exp\left( - \frac{t}{2} \left( \frac{1}{\lambda_1} - \frac{1}{\mu_1} \right) \right)
    ,\\
    m_1(t) &=
    \bigl( 1+o(1) \bigr)
    \,
    \frac{\mu_1}{\lambda_1}
    \sqrt{\frac{a_1(0) \, \bigl( \lambda_1-\lambda_2 \bigr)^2 (\lambda_2 + \mu_1)}{2 b_1(0) \, (\lambda_1 + \mu_1) \, \lambda_2^3}}
    \,
    \exp\left( - \frac{t}{2} \left( \frac{1}{\mu_1} - \frac{1}{\lambda_1} \right) \right)
    .
  \end{split}
\end{equation}
Taking logarithms, we can write
\eqref{eq:asymptotics-K2-amplitude-minus-infty} as
\begin{equation}
  \label{eq:asymptotics-K2-log-amplitude-minus-infty}
  \begin{split}
    -\ln n_4(t) &=
    \frac{t}{2} \left( \frac{1}{\lambda_2} \right)
    - \frac12 \ln \frac{b_{\infty}}{2 a_2(0)}
    + o(1)
    ,\\
    \ln m_3(t) &=
    \frac{t}{2} \left( \frac{1}{\lambda_2} - \frac{1}{\mu_1} \right)
    - \ln \lambda_2
    + \frac12 \ln \frac{a_2(0) \, (\lambda_2 + \mu_1)}{2 b_1(0)}
    + o(1)
    ,\\
    -\ln n_2(t) &=
    \frac{t}{2} \left( \frac{1}{\lambda_1} - \frac{1}{\mu_1} \right)
    + \ln \mu_1
    - \frac12 \ln \frac{b_1(0) \, (\lambda_1 + \mu_1)}{2 a_2(0)}
    \\ & \quad
    + \frac12 \ln \frac{\bigl( \lambda_1-\lambda_2 \bigr)^2 (\lambda_2 + \mu_1)}{\lambda_2^3}
    + o(1)
    ,\\
    \ln m_1(t) &=
    \frac{t}{2} \left( \frac{1}{\lambda_1} - \frac{1}{\mu_1} \right)
    + \ln \frac{\mu_1}{\lambda_1}
    + \frac12 \ln \frac{a_1(0)}{2 b_1(0) \, (\lambda_1 + \mu_1)}
    \\ & \qquad
    + \frac12 \ln \frac{\bigl( \lambda_1-\lambda_2 \bigr)^2 (\lambda_2 + \mu_1)}{\lambda_2^3}
    + o(1)
    .
  \end{split}
\end{equation}

From \eqref{eq:asymptotics-K2-position-minus-infty}
we conlude that 
the peakons asymptotically travel along straight lines
in $(x,t)$ space also in the case $t \to -\infty$,
with the same asymptotic velocities~$c_k$
(defined in equation~\eqref{eq:asymptotics-K2-speed-plus-infty})
as in the case $t \to +\infty$,
but in the opposite order:
\begin{equation}
  \label{eq:asymptotics-K2-speed-minus-infty}
  \begin{aligned}
    \dot x_1, \, \dot x_2 & \sim
    c_1 =
    \frac{1}{2} \left( \frac{1}{\lambda_1} + \frac{1}{\mu_1} \right)
    ,\\
    \dot x_3 & \sim
    c_2 =
    \frac{1}{2} \left( \frac{1}{\lambda_2} + \frac{1}{\mu_1} \right)
    ,\\
    \dot x_4 & \sim
    c_3 =
    \frac{1}{2} \left( \frac{1}{\lambda_2} \right)
    .
  \end{aligned}
\end{equation}
Now it is the two leftmost peakons that
asymptotically have the same velocity,
and from \eqref{eq:asymptotics-K2-position-minus-infty}
we get $x_2(t) - x_1(t) = o(1)$,
i.e., the distance between them approaches zero as~$t \to -\infty$.

Moreover,
the formulas \eqref{eq:asymptotics-K2-log-amplitude-minus-infty}
have the structure
\begin{equation}
  \label{eq:asymptotics-K2-logamp-speed-minus-infty}
  \begin{split}
    -\ln n_4(t) &=
    d_3 t + \text{constant} + o(1)
    ,\\
    \ln m_3(t) &=
    d_2 t + \text{constant} + o(1)
    ,\\
    -\ln n_2(t) &=
    d_1 t + \text{constant} + o(1)
    ,\\
    \ln m_1(t) &=
    d_1 t + \text{constant} + o(1)
    ,
  \end{split}
\end{equation}
which is similar to what
we had in the case $t \to +\infty$;
see~\eqref{eq:asymptotics-K2-logamp-speed-plus-infty}.
The constants $d_k$ defined in~\eqref{eq:logamp-slopes-dk}
appear here too,
but in the opposite order.
The amplitude of the rightmost peakon
always diverges:
$n_4(t) \to \infty$ as $t \to -\infty$,
since $d_3>0$.
The behaviour of the other amplitudes
depends on the relative sizes of the eigenvalues $\lambda_i$
and~$\mu_j$.

\subsection{Phase shifts in positions}
\label{sec:phaseshift-K2-positions}

As usual in soliton theory, we can identify ``phase shifts''
in the positions of solitons,
by comparing the asympotics as $t \to +\infty$
and $t \to -\infty$.

Thus, we may compare the straight line approached by the curves
$x=x_3(t)$ and $x=x_4(t)$ as $t \to +\infty$,
\begin{equation*}
  x = c_1 t
  + \frac12 \ln \frac{2 \, a_1(0) \, b_1(0)}{\lambda_1 + \mu_1}
  ,
\end{equation*}
to the parallel line approached by the curves $x=x_1(t)$ and $x=x_2(t)$
as $t \to -\infty$,
\begin{equation*}
  x = c_1 t
  + \frac12 \ln \frac{2 \, a_1(0) \, b_1(0)}{\lambda_1+\mu_1}
  + \frac12 \ln \frac{\bigl( \lambda_1-\lambda_2 \bigr)^2}{\lambda_2 (\lambda_2+\mu_1)}
  ,
\end{equation*}
and say that during the complete course of the evolution,
as $t$ runs from $-\infty$ to~$+\infty$,
``the pair of fast peakons with asymptotic velocity~$c_1$''
experiences a shift in the $x$ direction of size
\begin{equation}
  \label{eq:phaseshift-K2-fastest}
  \lim_{t \to +\infty} \bigl( x_{3,4}(t) - c_1 t \bigr)
  - \lim_{t \to -\infty} \bigl( x_{1,2}(t) - c_1 t \bigr)
  =
  - \frac12 \ln \frac{\bigl( \lambda_1-\lambda_2 \bigr)^2}{\lambda_2 (\lambda_2+\mu_1)}
  .
\end{equation}
Similarly, comparing the line approached by the curve
$x=x_2(t)$ as $t \to +\infty$,
\begin{equation*}
  x = c_2 t
  + \frac12 \ln \frac{2 \, a_2(0) \, b_1(0)}{\lambda_2+\mu_1}
  + \frac12 \ln \frac{\bigl( \lambda_1-\lambda_2 \bigr)^2}{\lambda_1 (\lambda_1+\mu_1)}
  ,
\end{equation*}
to the line approached by the curve
$x=x_3(t)$ as $t \to -\infty$,
\begin{equation*}
  x = c_2 t
  + \frac12 \ln \frac{2 \, a_2(0) \, b_1(0)}{\lambda_2 + \mu_1}
  ,
\end{equation*}
we see that ``the peakon with asymptotic velocity~$c_2$''
experiences the phase shift
\begin{equation}
  \label{eq:phaseshift-K2-middle}
  \lim_{t \to +\infty} \bigl( x_2(t) - c_2 t \bigr)
  - \lim_{t \to -\infty} \bigl( x_3(t) - c_2 t \bigr)
  =
  \frac12 \ln \frac{\bigl( \lambda_1-\lambda_2 \bigr)^2}{\lambda_1 (\lambda_1+\mu_1)}
  .
\end{equation}
And finally, comparing the line approached by the curve
$x=x_1(t)$ as $t \to +\infty$,
\begin{equation*}
  x = c_3 t
  + \frac12 \ln \frac{\mu_1 \, a_2(0)}{\lambda_1 \lambda_2 \, b^*_{\infty}}
  + \frac12 \ln \frac{\bigl( \lambda_1-\lambda_2 \bigr)^2}{\lambda_1 (\lambda_2 + \mu_1)}
  ,
\end{equation*}
to the line approached by the curve
$x=x_4(t)$ as $t \to -\infty$,
\begin{equation*}
  x = c_3 t
  + \frac12 \ln \bigl( 2 \, a_2(0) \, b_{\infty} \bigr)
  ,
\end{equation*}
it's clear that the phase shift of
``the slow peakon with asymptotic velocity~$c_3$'' is
\begin{multline}
  \label{eq:phaseshift-K2-slowest}
  \lim_{t \to +\infty} \bigl( x_1(t) - c_3 t \bigr)
  - \lim_{t \to -\infty} \bigl( x_4(t) - c_3 t \bigr)
  \\ =
  \frac12 \ln \frac{\mu_1}{2 \, b_{\infty} \, b^*_{\infty} \, \lambda_1 \lambda_2}
  + \frac12 \ln \frac{\bigl( \lambda_1-\lambda_2 \bigr)^2}{\lambda_1 (\lambda_2 + \mu_1)}
  .
\end{multline}

\subsection{Phase shifts in logarithms of amplitudes}
\label{sec:phaseshift-K2-amplitudes}

Geng--Xue peakons also exhibit a ``phase shift''
of a more unusual type,
not seen in systems where the amplitudes simply tend
to constant values as $t \to \pm \infty$.
Here, the amplitudes instead have exponental growth or decay
as $t \to \pm\infty$,
so their logarithms asymptotically behave like
$d_k t + \text{constant}$,
with the same coefficients $d_k$ appearing at $+\infty$ and $-\infty$
(see \eqref{eq:asymptotics-K2-log-amplitude-plus-infty}/\eqref{eq:asymptotics-K2-logamp-speed-plus-infty}
and~\eqref{eq:asymptotics-K2-log-amplitude-minus-infty}/\eqref{eq:asymptotics-K2-logamp-speed-minus-infty}).
Therefore we can make a similar comparison as we did
for positions above.
This gives the following formulas:
\begin{multline}
  \lim_{t \to +\infty} \Bigl( -\ln n_4(t) - d_1 t \Bigl)
  - \lim_{t \to -\infty} \Bigl( \ln m_1(t) - d_1 t \Bigr)
  \\
  =
  \ln \frac{2 \lambda_1 (\lambda_1 + \mu_1)}{\mu_1}
  - \frac12 \ln \frac{\bigl( \lambda_1 - \lambda_2 \bigr)^2 (\lambda_2 + \mu_1)}{\lambda_2^3}
  ,
\end{multline}
\begin{multline}
  \lim_{t \to +\infty} \Bigl( \ln m_3(t) - d_1 t \Bigl)
  - \lim_{t \to -\infty} \Bigl( -\ln n_2(t) - d_1 t \Bigr)
  = - \ln \frac{2 \lambda_1 \mu_1}{\lambda_1 + \mu_1}
  ,
  \hfill  % Ugly hack, to make this single-line equation aligned with the multlines...
\end{multline}
\begin{multline}
  \lim_{t \to +\infty} \Bigl( -\ln n_2(t) - d_2 t \Bigl)
  - \lim_{t \to -\infty} \Bigl( \ln m_3(t) - d_2 t \Bigr)
  \\
  =
    \ln \frac{2 \lambda_2 \mu_1}{\lambda_2 + \mu_1}
    + \frac12 \ln \frac{\bigl( \lambda_1 - \lambda_2 \bigr)^2 (\lambda_1 + \mu_1)}{\lambda_1^3}
  ,
\end{multline}
and
\begin{multline}
  \lim_{t \to +\infty} \Bigl( \ln m_1(t) - d_3 t \Bigl)
  - \lim_{t \to -\infty} \Bigl( -\ln n_4(t) - d_3 t \Bigr)
  \\
  = 
  \frac12 \ln \frac{\mu_1}{2 \lambda_1 \lambda_2}
  + \frac12 \ln \bigl( b_{\infty} \, b_{\infty}^* \bigr)
  + \frac12 \ln \frac{\bigl( \lambda_1 - \lambda_2 \bigr)^2}{\lambda_1 (\lambda_2 + \mu_1)}
  .
\end{multline}

\section{Shock formation in a $2+2$ mixed peakon--antipeakon case}
\label{sec:2+2-peakon-antipeakon}

Throughout the article, except in this section, we are assuming
that all the peakon amplitudes $m_k$ and $n_k$ are positive,
i.e., we are only considering what is known as ``pure peakon solutions''.
The experience from other peakon equations
(Camassa--Holm, Degasperis--Procesi, Novikov)
is that pure peakon solutions are globally defined, whereas mixed
peakon--antipeakon solutions,
with some amplitudes positive and some negative,
lead to collisions, finite-time blowup, and subtle questions of how
to continue the solution past singularities
\cite{bressan-constantin:global-conservative-CH,
  bressan-constantin:global-dissipative-CH,
  holden-raynaud:global-conservative-multipeakon-CH,
  holden-raynaud:global-conservative-CH-Lagrangian,
  holden-raynaud:global-dissipative-multipeakon-CH,
  lundmark:shockpeakons}.
% \TODO{Add Kardell-Lundmark here too, once it's finished.}

For the Geng--Xue equation, negative amplitudes cause no problems as
long as all the peakons in~$u$ have the same sign and
all the peakons in~$v$ have the same sign.
In fact, from the governing ODEs \eqref{eq:GX-peakon-ode}
it is immediate that
if
\begin{equation*}
  x_k = \xi_k(t)
  ,\qquad
  m_k = \mu_k(t)
  ,\qquad
  n_k = \nu_k(t)
\end{equation*}
is a pure peakon solution,
then
\begin{equation*}
  x_k = \xi_k(t)
  ,\qquad
  m_k = -\mu_k(t)
  ,\qquad
  n_k = -\nu_k(t)
\end{equation*}
is a pure antipeakon solution, while
\begin{equation*}
  x_k = \xi_k(-t)
  ,\qquad
  m_k = -\mu_k(-t)
  ,\qquad
  n_k = \nu_k(-t)
\end{equation*}
is a solution with antipeakons in~$u$ and peakons in~$v$,
and the other way around for
\begin{equation*}
  x_k = \xi_k(-t)
  ,\qquad
  m_k = \mu_k(-t)
  ,\qquad
  n_k = -\nu_k(-t)
  .
\end{equation*}
This is the reason why there was nothing remarkable
about the $1+1$ interlacing peakon--antipeakon case
considered in \autoref{rem:1+1-peakon-antipeakon}.

But when we mix peakons and antipeakons within one component
of a solution, there will indeed be complications.
Just to get an idea of what may happen,
we will spend the rest of this section
looking at the $2+2$ interlacing case with
$m_1$, $n_2$, $m_3$ positive and $n_4$ negative.
It will emerge that it is possible for the solution
$\bigl( u(x,t), v(x,t) \bigr)$
to form a jump discontinuity after finite time,
meaning that one is forced to consider shockpeakons
in order to provide a meaningful continuation past the singularity;
this is similar to what happens for the
Degasperis--Procesi equation~\cite{lundmark:shockpeakons}.

When negative amplitudes are involved,
the spectral variables will not lie in the usual positive
sector~\eqref{eq:spectral-parameters},
so we must begin by investigating their signs.
From \eqref{eq:ABC-jump-product} we compute
\begin{equation*}
  \begin{aligned}    
    A(\lambda) &
    = 1
    - 2 \lambda \bigl( m_1 n_2 E_{12} + m_1 n_4 E_{14} + m_3 n_4 E_{34} \bigr)
    + 4 \lambda^2 m_1 n_2 m_3 n_4 E_{12} (1-E_{23}^2) E_{34}
    ,\\
    B(\lambda) &
    = e^{x_3} \Bigl(
    (m_1 E_{13} + m_3)
    - 2 \lambda m_1 n_2 m_3 E_{12} (1-E_{23}^2)
    \Bigr)
    ,
  \end{aligned}
\end{equation*}
while \eqref{eq:twin-ABC-jump-product} gives
\begin{equation*}
  \begin{aligned}
    \twin A(\lambda) &
    = 1
    - 2 \lambda n_2 m_3 E_{23}
    ,\\
    \twin B(\lambda) &
    = e^{x_4} \Bigl(
    (n_2 E_{24} + n_4)
    - 2 \lambda n_2 m_3 n_4 E_{23} (1-E_{34}^2)
    \Bigr)
    .
  \end{aligned}
\end{equation*}
Since the polynomial~$A(\lambda) = (1-\lambda/\lambda_1)(1-\lambda/\lambda_2)$
has negative $\lambda^2$-coefficient,
its zeros will be of opposite sign, say $\lambda_1 < 0 < \lambda_2$,
and the single zero of $\twin A(\lambda)$
is $\mu_1 = (2 n_2 m_3 E_{23})^{-1} > 0$.
From \eqref{eq:b-infty-star-directly} and \eqref{eq:b-infty-directly}
we see that
$b_{\infty} < 0$ and~$b^*_{\infty} > 0$.

For the sake of simplicity, we will now consider some concrete
numerical values.

\begin{example}
  The following initial data are designed to give simple spectral data:
  \begin{equation*}
    x_1(0) = 0
    ,\quad
    x_2(0) = \tfrac12 \ln\tfrac{9}{4}
    ,\quad
    x_3(0) = \tfrac12 \ln\tfrac{7}{2}
    ,\quad
    x_4(0) = \tfrac12 \ln\tfrac{11}{2}
  \end{equation*}
  and
  \begin{equation*}
    m_1(0) = \tfrac{9}{5}
    ,\qquad
    n_2(0) = \tfrac{5}{6}
    ,\qquad
    m_3(0) = \tfrac{1}{5} \sqrt{\tfrac{7}{2}}
    ,\qquad
    n_4(0) = -\tfrac{1}{2} \sqrt{\tfrac{11}{2}}
    .
  \end{equation*}
  The Weyl functions at time $t=0$ are
  \begin{equation*}
    \omega(\lambda;0)
    = - \frac{B(\lambda;0)}{A(\lambda;0)}
    = - \frac{\tfrac52 - \tfrac12 \lambda}{1 + \tfrac12 \lambda - \tfrac12 \lambda^2}
    = \frac{-2}{\lambda-(-1)} + \frac{1}{\lambda-2}
    = \frac{a_1(0)}{\lambda-\lambda_1} + \frac{a_2(0)}{\lambda-\lambda_2}
  \end{equation*}
  and
  \begin{equation*}
    \twin\omega(\lambda;0)
    = - \frac{\twin B(\lambda;0)}{\twin A(\lambda;0)}
    = - \frac{-\frac32 + \frac12 \lambda}{1 - \tfrac12 \lambda}
    = 1 + \frac{-1}{\lambda - 2}
    = -b_{\infty} + \frac{b_1(0)}{\lambda - \mu_1}
    ,
  \end{equation*}
  so that we have
  \begin{equation*}
    \lambda_1 = -1
    ,\quad
    \lambda_2 = 2
    ,\quad
    \mu_1 = 2
    ,\quad
    a_1(0) = -2
    ,\quad
    a_2(0) = 1
    ,\quad
    b_1(0) = -1
    ,\quad
    b_{\infty} = -1
    ,
  \end{equation*}
  and also $b^*_{\infty} = 1$ from \eqref{eq:b-infty-star-directly}.
  Then, since $a_1 = a_1(t) = a_1(0) \, e^{t/\lambda_1} = -2 e^{-t}$,
  and so on,
  the quantities appearing in the peakon solution
  formulas~\eqref{eq:solution-K2} are
  \begin{equation*}
    \begin{split}
      a_1 + a_2 &
      = -2 e^{-t} + e^{t/2}
      < 0
      \iff
      t < \tfrac23 \ln 2
      ,\\
      \lambda_1 a_1 + \lambda_2 a_2 &
      = 2 e^{-t} + 2 e^{t/2}
      > 0
      ,\\
      \frac{a_1 b_1}{\lambda_1 + \mu_1} + \frac{a_2 b_1}{\lambda_2 + \mu_1} &
      = 2 e^{-t/2} - \tfrac14 e^{t}
      > 0
      \iff
      t < 2 \ln 2
      ,\\
      \frac{\lambda_1 a_1 b_1}{\lambda_1 + \mu_1} + \frac{\lambda_2 a_2 b_1}{\lambda_2 + \mu_1} &
      = -2 e^{-t/2} - \tfrac12 e^{t}
      < 0
      ,\\
      \frac{(\lambda_1-\lambda_2)^2}{(\lambda_1+\mu_1)(\lambda_2+\mu_1)} a_1 a_2 b_1 &
      = \tfrac92
      > 0
      .
    \end{split}
  \end{equation*}
  Two of these quantities change their sign,
  and the first one to become zero is
  $a_1 + a_2$, which changes sign from negative to positive
  when
  \begin{equation*}
    t = t_0 := \tfrac23 \ln 2
    .
  \end{equation*}
  The fact that the solution formulas really satisfy the peakon ODEs
  is purely algebraic, and does not depend on the signs of the
  spectral variables,
  as long as everything is defined and the ordering $x_1 \le x_2 \le x_3 \le x_4$
  is preserved.
  From \eqref{eq:solution-K2-position} we obtain
  \begin{equation*}
    \begin{aligned}
      \tfrac12 e^{2 x_4} - \tfrac12 e^{2 x_3} &
      = b_{\infty} (a_1 + a_2)
      \\ &
      -2 e^{-t} + e^{t/2}
      ,\\[1ex]
      \tfrac12 e^{2 x_3} - \tfrac12 e^{2 x_2} &
      = \frac{a_1+a_2}{\lambda_1 a_1 + \lambda_2 a_2}
      \left(
        \frac{\lambda_1 a_1 b_1}{\lambda_1+\mu_1} + \frac{\lambda_2 a_2 b_1}{\lambda_2+\mu_1}
      \right)
      \\ &
      = \frac{\bigl( -2 e^{-t} + e^{t/2} \bigr) \, \bigl( -2 e^{-t/2} - \tfrac12 e^{t} \bigr)}{2 e^{-t} + 2 e^{t/2}}
      ,\\[1ex]
      2 e^{-2 x_1} - 2 e^{-2 x_2} &
      = \frac{2 b_{\infty}^* \lambda_1 \lambda_2}{\mu_1} \left( \dfrac{\lambda_1 a_1 b_1}{\lambda_1+\mu_1} + \dfrac{\lambda_2 a_2 b_1}{\lambda_2+\mu_1} \right)
      \\ &
      = -2 \, \bigl( -2 e^{-t/2} - \tfrac12 e^{t} \bigr)
      ,
    \end{aligned}
  \end{equation*}
  so the ordering is indeed preserved up until $t=t_0$,
  when a (triple) collision
  \begin{equation*}
    \tfrac12 \ln\tfrac{3}{2^{1/3}+2^{2/3}} =
    x_1(t_0) < x_2(t_0) = x_3(t_0) = x_4(t_0)
    = \tfrac12 \ln\tfrac{3}{2^{1/3}}
  \end{equation*}
  occurs,
  and looking at where the factor $a_1+a_2$ occurs in~\eqref{eq:solution-K2-amplitude}
  we also see that
  \begin{equation*}
    n_2(t) \to \infty
    ,\qquad
    m_3(t) \to 0
    ,\qquad
    n_4(t) \to -\infty
    ,\qquad
    \text{as $t \to t_0^-$}.
  \end{equation*}
  The component $u(x,t)$ simply converges to $m_1(t_0) \, e^{-\abs{x-x_1(t_0)}}$,
  where
  \begin{equation*}
    m_1(t_0) = \bigl(\tfrac32 (1+2^{1/3}) \bigr)^{1/2}
    ,
  \end{equation*}
  i.e.,
  \begin{equation}
    \label{eq:limiting-profile-u}
    \lim_{t \to t_0^-} u(x,t) 
    =
    \bigl(\tfrac32 (1+2^{1/3}) \bigr)^{1/2}
    \, e^{-\abs{x-\tfrac12 \ln\tfrac{3}{2^{1/3}+2^{2/3}}}}
    .
  \end{equation}
  The behaviour of~$v(x,t)$ is more subtle.
  The solution formulas~\eqref{eq:solution-K2} imply that
  \begin{equation*}
    \begin{split}
      v(x_2(t),t) &
      = n_2 + n_4 E_{24}
      \\ &
      = \tfrac{1}{\sqrt2} \, \bigl( \tfrac12 e^{2 x_2} \bigr)^{1/2} \,
      \bigl( 2 n_2 e^{-x_2} + 2 n_4 e^{-x_4} \bigr)
      \\ &
      =
      \frac{1}{\sqrt2} \,
      \Biggl(
        \dfrac{\tfrac{(\lambda_1-\lambda_2)^2}{(\lambda_1+\mu_1)(\lambda_2+\mu_1)} a_1 a_2 b_1}{\lambda_1 a_1 + \lambda_2 a_2}
      \Biggr)^{1/2}
      \\ & \qquad
      \times
      \frac{(\lambda_1+\mu_1)(\lambda_2+\mu_1)}{(\lambda_1 - \lambda_2)^2 \mu_1 a_1 a_2}
      \left( \frac{\lambda_1^2 a_1}{\lambda_1+\mu_1} + \frac{\lambda_2^2 a_2}{\lambda_2+\mu_1} \right)
      \\ &
      =
      \frac{1}{\sqrt2} \,
      \Biggl(
        \dfrac{9/2}{2 e^{-t} + 2 e^{t/2}}
      \Biggr)^{1/2}
      \frac{-e^{t/2}}{9}
      \bigl( -2 e^{-t} + e^{t/2} \bigr)
      .
    \end{split}
  \end{equation*}
  Here it is important that the expression
  $2 n_2 e^{-x_2} + 2 n_4 e^{-x_4}$
  has been simplified by cancelling a common factor $a_1+a_2$
  from the numerator and the denominator.
  The factor $a_1+a_2$ likewise cancels when computing
  \begin{equation*}
    \begin{split}
      v(x_4(t),t) &
      = n_2 E_{24} + n_4
      \\ &
      = \tfrac{1}{\sqrt2} \, \bigl( \tfrac12 e^{2 x_4} \bigr)^{-1/2} \,
      \bigl( 2 n_2 e^{-x_2} \cdot \tfrac12 e^{2 x_2} + 2 n_4 e^{-x_4} \cdot \tfrac12 e^{2 x_4} \bigr)
      \\ &
      =
      \frac{1}{\sqrt2} \,
      \Biggl(
      \frac{a_1 b_1}{\lambda_1+\mu_1} + \frac{a_2 b_1}{\lambda_2+\mu_1} + b_{\infty} (a_1 + a_2)
      \Biggr)^{-1/2}
      \Biggl(
      \frac{b_1}{\mu_1} + b_{\infty}
      \Biggr)
      \\ &
      =
      \frac{1}{\sqrt2} \,
      \bigl(
      2 e^{-t/2} - \tfrac14 e^{t} + 2 e^{-t} - e^{t/2}
      \bigr)^{-1/2}
      \bigl(
      - \tfrac12 e^{t/2} - 1
      \bigr)
      .
    \end{split}
  \end{equation*}
  Because of this cancellation, these two expressions have finite limits
  as $t \to t_0^-$,
  \begin{equation*}
    \lim_{t \to t_0^-} v(x_2(t),t) = 0
    ,\qquad
    \lim_{t \to t_0^-} v(x_4(t),t) = - \frac{3 (1+2^{2/3})}{4 \sqrt2} < 0
    ,
  \end{equation*}
  which means that $v(x,t)$ converges to a shockpeakon-type wave profile
  at the instant of collision,
  as claimed:
  \begin{equation}
    \label{eq:limiting-profile-v}
    \begin{split}
      \lim_{t \to t_0^-} v(x,t) &
      =
      \begin{cases}
        \displaystyle
        \Biggl(
        \lim_{t \to t_0^-} v(x_2(t),t)
        \Biggr)
        \, e^{x-x_0}
        , &
        x < x_0
        \\[1em]
        \displaystyle
        \Biggl(
        \lim_{t \to t_0^-} v(x_4(t),t)
        \Biggr)
        \, e^{x_0-x}
        , &
        x > x_0
      \end{cases}
      \\ &
      =
      \begin{cases}
        0
        , &
        x < x_0
        \\
        - \frac{3 (1+2^{2/3})}{4 \sqrt2} \, e^{x_0 - x}
        , &
        x > x_0
        ,
      \end{cases}
    \end{split}
  \end{equation}
  where $x_0 = x_2(t_0) = x_3(t_0) = x_4(t_0) = \tfrac12 \ln\tfrac{3}{2^{1/3}}$
  is the site of the triple collision.
  (The fact that $v$ becomes identically zero for $x < x_0$ is clearly
  a bit of a coincidence;
  it's just that for our particular spectral data,
  the factor
  $\frac{\lambda_1^2 a_1}{\lambda_1+\mu_1} + \frac{\lambda_2^2 a_2}{\lambda_2+\mu_1}$
  happens to equal $a_1+a_2$,
  and hence it vanishes at the collision.)

  Since the second peakon in $u$ vanishes at the instant of the collision,
  the non-overlapping condition is actually preserved automatically,
  and we can continue the solution for $t \ge t_0$
  by taking the limiting wave profiles
  \eqref{eq:limiting-profile-u} and~\eqref{eq:limiting-profile-v}
  as initial data $u(x,t_0)$ and~$v(x,t_0)$ for a
  $1+1$-shockpeakon solution starting at $t=t_0$,
  i.e., a solution of the type in \autoref{sec:dynamics-shock-1},
  with
  \begin{equation*}
    \begin{gathered}
      x_1(t_0) = \tfrac12 \ln\tfrac{3}{2^{1/3}+2^{2/3}}
      ,\qquad
      x_2(t_0) = \tfrac12 \ln\tfrac{3}{2^{1/3}}
      ,\\
      m_1(t_0) = \bigl(\tfrac32 (1+2^{1/3}) \bigr)^{1/2}
      ,\quad
      s_1(t_0) = 0
      ,\quad
      n_2(t_0) = -r_2(t_0) = - \frac{3 (1+2^{2/3})}{8 \sqrt2}
      .
    \end{gathered}
  \end{equation*}
\end{example}

\begin{remark}
  Apart from collisions,
  another technical complication with mixed peakon--antipeakon solutions
  is that there may be resonant cases where some $\lambda_i + \mu_j$
  vanishes.
  In such cases there will be division by zero in
  the usual solution formulas, so they will not be valid at all,
  not even before the blowup.
  However, this can be handled by limiting arguments.
  Just to give one example, consider the initial data
  \begin{equation*}
    x_1(0) = 0
    ,\quad
    x_2(0) = \ln 2
    ,\quad
    x_3(0) = \ln 4
    ,\quad
    x_4(0) = \ln 8
  \end{equation*}
  and
  \begin{equation*}
    m_1(0) = n_2(0) = m_3(0) = 1
    ,\qquad
    n_4(0) = -1
    .
  \end{equation*}
  Then
  \begin{equation*}
    \omega(\lambda;0)
    = - \frac{B(\lambda)}{A(\lambda)}
    = - \frac{5 - 3 \lambda}{1 + \tfrac14 \lambda - \tfrac34 \lambda^2}
    = \frac{-32/7}{\lambda-(-1)} + \frac{4/7}{\lambda-\tfrac43}
    = \frac{a_1(0)}{\lambda-\lambda_1} + \frac{a_2(0)}{\lambda-\lambda_2}
  \end{equation*}
  and
  \begin{equation*}
    \twin\omega(\lambda;0)
    = - \frac{\twin B(\lambda)}{\twin A(\lambda)}
    = - \frac{-6 + 6 \lambda}{1 - \lambda}
    = 6 + \frac{0}{\lambda - 1}
    = -b_{\infty} + \frac{b_1(0)}{\lambda - \mu_1}
    ,
  \end{equation*}
  so that $\lambda_1 = -1$, $\lambda_2 = \tfrac43$,
  $\mu_1 = 1$, and in particular $\lambda_1 + \mu_1 = 0$.
  However, this is balanced in the peakon solution
  formulas~\eqref{eq:solution-K2}
  by the fact that $b_1(0) = 0$ also.
  In fact, if we take $n_4(0)=-(1+\epsilon)$ instead,
  we find
  \begin{equation*}
    \lambda_1 = \frac{1+5\epsilon-\sqrt{49+58\epsilon+25\epsilon^2}}{6(1+\epsilon)}
    ,\qquad
    \mu_1 = 1
    ,\qquad
    b_1(0) = -2 \epsilon
    ,
  \end{equation*}
  so that
  \begin{equation*}
    \frac{b_1(t)}{\lambda_1+\mu_1}
    = \frac{b_1(0) \, e^{t/\mu_1}}{\lambda_1+\mu_1}
    \to
    - \frac74 \, e^t
    ,\qquad
    \text{as $\epsilon \to 0$}
    .
  \end{equation*}
  Thus the correct solution formulas in this limiting case
  are given by replacing the expression $b_1(t)/(\lambda_1+\mu_1)$
  with $-\tfrac74 e^t$ wherever it occurs
  in~\eqref{eq:solution-K2}.

\end{remark}

\section{Dynamics of $K+K$ interlacing pure peakon solutions}
\label{sec:dynamics-K}

With the detailed examples of the previous sections
(in particular \autoref{sec:dynamics-2}) under our belt,
we are now ready to tackle the asymptotics of the
general $K+K$ interlacing peakon solution
(described in \autoref{cor:peakon-solution-formulas}).
Since the exceptional case $K=1$ has been treated
in \autoref{sec:dynamics-1},
we will assume in this section that $K \ge 2$.
Moreover, we will only consider pure peakon solutions,
i.e., we will assume that all the amplitudes $m_k$ and~$n_k$
are positive.

\subsection{Preparations}

First we establish the large time asymptotic 
behaviour of the sum~$\heineintegral_{nm}^{rs}$
defined by equation~\eqref{eq:heine-integral-as-sum}.
We recall the definition here, for convenience:
\begin{equation*}
  \heineintegral_{nm}^{rs} =
  \sum_{I \in \binom{[K]}{n}} \sum_{J \in \binom{[K-1]}{m}}
  \Psi_{IJ} \, \lambda_I^r a_I \, \mu_J^s b_J.
\end{equation*}
Determining the asymptotics of this sum is quite easy,
since the dominant contribution comes from a
single term, with all other terms being exponentially small
in comparison.
(If we write the sum with the index sets $I$ and~$J$ in
lexicographic order, then the first term dominates as $t \to +\infty$,
and the last term dominates as $t \to -\infty$.)

We remind the reader that our notation was defined
in \autoref{sec:inverse-spectral-map};
in particular, $[k]$ denotes the integer interval
$\{ 1,2, \dots, k \}$ if $k \ge 1$, and the empty set if~$k=0$.

\begin{lemma}
  \label{lem:asymptotics-J}
  Suppose $0 \le n \le K$ and $0 \le m \le K-1$.
  Given the time evolution of the
  spectral data described by \autoref{thm:spectral-variables-ode},
  the leading long-time behaviour of~$\heineintegral_{nm}^{rs}$
  is given by
  \begin{equation}
    \heineintegral_{nm}^{rs}
    =
    \bigl( 1+o(1) \bigr) \,
    \Psi_{AB} \lambda_A^r \mu_B^s a_A b_B
    ,\qquad
    t \to +\infty
    ,
  \end{equation}
  where $A=[n]$ and~$B=[m]$, and by
  \begin{equation}
    \heineintegral_{nm}^{rs}
    =
    \bigl( 1+o(1) \bigr) \,
    \Psi_{CD} \lambda_C^r \mu_D^s a_C b_D
    ,\qquad
    t \to -\infty
    ,
  \end{equation}
  where $C = [K] \setminus [K-n]$
  and $D = [K-1] \setminus [K-1-m]$.
  (Thus, $C = \{ K-n+1,\dots,K \}$ if $n \ge 1$
  and $C = \emptyset$ if $n=0$, and similarly for~$D$.)
\end{lemma}

\begin{proof}
  Remember (see \eqref{eq:spectral-parameters})
  that we always label the eigenvalues so that they are ordered:
  \begin{equation*}
    0 < \lambda_1 < \lambda_2 < \dots < \lambda_K
    , \qquad
    0 < \mu_1 < \mu_2 < \dots < \mu_{K-1}.
  \end{equation*}
  This implies that
  $a_1(t) \gg a_2(t) \gg \dots \gg a_K(t)$
  and
  $b_1(t) \gg b_2(t) \gg \dots \gg b_{K-1}(t)$
  as $t \to +\infty$.
  What we mean by this is that if $i>j$, then
  \begin{equation*}
    \frac{a_i(t)}{a_j(t)}
    = \frac{a_i(0)}{a_j(0)} \, \exp\left( t \left( \frac{1}{\lambda_i} - \frac{1}{\lambda_j} \right) \right)
    \to 0
    ,\qquad
    \frac{b_i(t)}{b_j(t)}
    = \frac{b_i(0)}{b_j(0)} \, \exp\left( t \left( \frac{1}{\mu_i} - \frac{1}{\mu_j} \right) \right)
    \to 0
    ,
  \end{equation*}
  as $t \to +\infty$.
  Then it is clear that the term $a_A b_B$,
  which contains the smallest indices, will be dominant as $t \to +\infty$\,;
  if we factor it out from the sum
  \eqref{eq:heine-integral-as-sum},
  what remains is the constant
  $\Psi_{AB} \lambda_A^r \mu_B^s$ plus terms which tend to zero
  (exponentially fast)
  as $t \to +\infty$.

  The case $t \to -\infty$ is similar,
  taking into account that
  $a_1(t) \ll a_2(t) \ll \dots \ll a_K(t)$
  and
  $b_1(t) \ll b_2(t) \ll \dots \ll b_{K-1}(t)$
  as $t \to -\infty$.
\end{proof} 

\subsection{Asymptotics for positions}

We can now easily derive the general formulas for the asymptotics of the
positions~$x_k$ and their derivatives~$\dot x_k$.
In order to make these formulas more readable,
we use the notation $c_1,\dots,c_{2K-1}$
for the asymptotic velocities,
and we also introduce abbreviations for certain other combinations
of the eigenvalues $\lambda_i$ and~$\mu_j$.

\begin{definition}
  \label{def:notation-asymptotics-positions}
  Define
  $c_1 > c_2 > \dots > c_{2K-1} > 0$
  by
  \begin{equation}
    \label{eq:asymptotic-velocities}
    \begin{aligned}
      c_{2j}
      &=
      \frac12 \left( \frac{1}{\lambda_{j+1}} + \frac{1}{\mu_{j}} \right)
      ,\quad
      j = 1, \dots, K-1
      ,\\[1.5ex]
      c_{2j-1}
      &=
      \begin{cases}
        \displaystyle
        \frac12 \left( \frac{1}{\lambda_{j}} + \frac{1}{\mu_{j}} \right)
        ,&
        j = 1, \dots, K-1
        ,\\[2ex] \displaystyle
        \frac12 \left( \frac{1}{\lambda_{K}} \right)
        ,&
        j = K
        .
      \end{cases}
    \end{aligned}
  \end{equation}
  Moreover, let
  \begin{equation}
    \begin{aligned}
      R_{rj}'
      &=
      \frac12 \ln \frac{(\lambda_r - \lambda_j)^2}{\lambda_r (\lambda_r + \mu_j)}
      ,\\
      R_{rj}''
      &=
      \frac12 \ln \frac{(\lambda_r - \lambda_j)^2}{\lambda_r (\lambda_r + \mu_{j-1})}
      ,\\
      S_{sj}'
      &=
      \frac12 \ln \frac{(\mu_s - \mu_j)^2}{(\lambda_j + \mu_s) \mu_s}
      ,\\
      S_{sj}''
      &=
      \frac12 \ln \frac{(\mu_s - \mu_{j-1})^2}{(\lambda_j + \mu_s) \mu_s}
      .
    \end{aligned}
  \end{equation}
\end{definition}

\begin{theorem}[Asymptotics for positions and velocities]
  \label{thm:asymptotics-generalK-positions}
  \begin{subequations}
  In terms of the abbreviations of \autoref{def:notation-asymptotics-positions},
  the positions and velocities in the
  $K+K$ interlacing Geng--Xue peakon solution with $K \ge 2$
  satisfy the following asymptotic formulas
  (where empty sums, such as $\sum_{s=j}^{K-1}$ when $j=K$,
  should be interpreted as zero).

  Asymptotic velocities as $t \to -\infty$\,:
  \begin{equation}
    \dot x_i \sim
    \begin{cases}
      c_1
      ,&
      i = 1
      ,\\
      c_{i-1}
      ,&
      i = 2, 3, \dots, 2K
      .
    \end{cases}
  \end{equation}

  Asymptotic velocities as $t \to +\infty$\,:
  \begin{equation}
    \dot x_{2K+1-i} \sim
    \begin{cases}
      c_1
      ,&
      i = 1
      ,\\
      c_{i-1}
      ,&
      i = 2, 3, \dots, 2K
      .
    \end{cases}
  \end{equation}

  Asymptotics for positions as $t \to -\infty$\,:
  \begin{equation}
    \label{eq:asy-K-pos-minus-odd}
    x_{2j-1}(t) =
    \begin{cases}
      \displaystyle
      c_1 t
      + \frac12 \ln \frac{2 \, a_1(0) \, b_1(0)}{\lambda_1 + \mu_1}
      \\[2ex] \displaystyle \qquad +
      \sum_{r=2}^K R_{r1}'
      +
      \sum_{s=2}^{K-1} S_{s1}'
      + o(1)
      ,&
      j=1
      ,\\[2em] \displaystyle
      c_{2j-2} t
      +\frac12 \ln \frac{2 \, a_j(0) \, b_{j-1}(0)}{\lambda_j + \mu_{j-1}}
      \\[2ex] \displaystyle \qquad +
      \sum_{r=j+1}^K R_{rj}''
      +
      \sum_{s=j}^{K-1} S_{sj}''
      + o(1)
      ,&
      j = 2,\dots,K
      ,
    \end{cases}
  \end{equation}
  and
  \begin{equation}
    \label{eq:asy-K-pos-minus-even}
    x_{2j}(t) =
    \begin{cases}
      \displaystyle
      c_{2j-1} t
      + \frac12
      \ln \frac{2 \, a_j(0) \, b_j(0)}{\lambda_j + \mu_j}
      \\[2ex] \displaystyle \qquad +
      \sum_{r=j+1}^K R_{rj}'
      +
      \sum_{s=j+1}^{K-1} S_{sj}'
      + o(1)
      ,&
      j=1,\dots,K-1
      ,\\[2em] \displaystyle
      c_{2K-1} t
      + \frac12 \ln \bigl( 2 \, a_K(0) \, b_{\infty} \bigr)
      +o(1)
      ,&
      j= K
      .
    \end{cases}
  \end{equation}
  Asymptotics for positions as $t \to +\infty$\,:
  \begin{equation}
    \label{eq:asy-K-pos-plus-even}
    x_{2(K+1-j)}(t)
    =
    \begin{cases}
      \displaystyle
      c_1 t
      + \frac12 \ln \frac{2\, a_1(0)b_1(0)}{\lambda_1+\mu_1} + o(1)
      ,&
      j = 1
      ,\\[2em] \displaystyle
      c_{2j-2} t
      +\frac12 \ln \frac{2 \, a_j(0) \, b_{j-1}(0)}{\lambda_j + \mu_{j-1}}
      \\[2ex] \displaystyle \qquad +
      \sum_{r=1}^{j-1} R_{rj}''
      +
      \sum_{s=1}^{j-2} S_{sj}''
      + o(1)
      ,&
      j = 2,\dots,K
      ,
    \end{cases}
  \end{equation}
  and
  \begin{equation}
    \label{eq:asy-K-pos-plus-odd}
    x_{2(K+1-j)-1}(t)
    =
    \begin{cases}
      \displaystyle
      c_{2j-1} t
      + \frac12
      \ln \frac{2 \, a_j(0) \, b_j(0)}{\lambda_j + \mu_j}
      \\[2ex] \displaystyle \qquad +
      \sum_{r=1}^{j-1} R_{rj}'
      +
      \sum_{s=1}^{j-1} S_{sj}'
      + o(1)
      ,&
      j=1,\dots,K-1
      ,\\[2em] \displaystyle
      c_{2K-1} t
      + \frac12 \ln \frac{M \, a_K(0)}{\strut L \, b^*_{\infty}}
      \\[2ex] \displaystyle \qquad +
      \frac12 \sum_{r=1}^{K-1} \ln \frac{(\lambda_r - \lambda_K)^2}{\lambda_r (\lambda_K + \mu_r)}
      + o(1)
      ,&
      j= K
      .
    \end{cases}
  \end{equation}
  \end{subequations}
\end{theorem}

\begin{remark}
  The formulas in \autoref{thm:asymptotics-generalK-positions}
  are somewhat involved, so let us say a few words about their structure
  before we delve into the details of the proof.
  (Cf. the discussion of the case $K=2$ in \autoref{sec:dynamics-2}.)

  The essential feature is that each curve $x = x_k(t)$ approaches a certain
  straight line $x=At+B$ as $t \to -\infty$,
  and another line $x=Ct+D$ as $t \to +\infty$.
  The coefficients $A$ and~$C$ are the asymptotic velocities of the peakon
  in question, and they belong to the set $\{ c_1,c_2,\dots,c_{2K-1} \}$.
  So even though there are $2K$ peakons, there are only $2K-1$ asymptotic
  velocities, which are numbered in decreasing order,
  \begin{equation*}
    c_1 > c_2 > \dots > c_{2K-1}
    ,
  \end{equation*}
  and only depend on the eigenvalues
  (see \autoref{def:notation-asymptotics-positions}).
  The coefficients $B$ and~$D$ are given by more complicated expressions
  involving all the spectral variables.

  As $t \to -\infty$, the two leftmost peakons both asymptotically have the 
  fastest velocity~$c_1$, and in fact we see from the cases $j=1$
  in \eqref{eq:asy-K-pos-minus-odd} and~\eqref{eq:asy-K-pos-minus-even}
  that the curves $x=x_1(t)$ and $x=x_2(t)$
  approach the same asymptotic line
  \begin{equation*}
    x =
    c_1 t
    +
    B
    ,\qquad
    \text{where}
    \quad
    B =
    \frac12 \ln \frac{2 \, a_1(0) \, b_1(0)}{\lambda_1 + \mu_1}
    +
    \sum_{r=2}^K R_{r1}'
    +
    \sum_{s=2}^{K-1} S_{s1}'
    ,
  \end{equation*}
  i.e., the distance $x_2(t)-x_1(t)$ tends to zero.
  The third peakon has asymptotic velocity $c_2$,
  the fourth one~$c_3$, and so on.

  Similarly, as $t \to +\infty$,
  we see from the cases $j=1$ in
  in \eqref{eq:asy-K-pos-plus-odd} and~\eqref{eq:asy-K-pos-plus-even}
  that the curves $x=x_{2K-1}(t)$ and $x=x_{2K}(t)$ have the same asymptote
  \begin{equation*}
    x = c_1 t + D
    ,\qquad
    \text{where}
    \quad
    D =
    \frac12 \ln \frac{2 \, a_1(0) \, b_1(0)}{\lambda_1 + \mu_1}
    ,
  \end{equation*}
  so the two rightmost peakons approach each other and both asymptotically
  have the fastest velocity~$c_1$.
  The third peakon from the right (number $2K-2$) has asymptotic velocity~$c_2$,
  and so on.

  Recall that $x_1(t)$ and $x_{2K}(t)$ are given by formulas which look
  different from the ones for the other positions~$x_j(t)$;
  in particular,
  the spectral variable $b_{\infty}$ only enters in the formula for~$x_{2K}$,
  and $b^*_{\infty}$ only affects~$x_1$
  (see \autoref{thm:inverse-spectral-map-formulas}).
  This is the reason for the division into cases in the asymptotic formulas here.
  If we look at the curve $x_1(t)$
  and its asymptote $x = At+B$ as $t \to -\infty$,
  we see that both its velocity $A=c_1$ and the coefficient~$B$
  deviate completely from the pattern followed by the other odd-numbered
  peakons~$x_{2j-1}$.
  However, in the asymptote $x = Ct+D$ for $x_1(t)$ as $t \to +\infty$
  (the case $j=K$ in~\eqref{eq:asy-K-pos-plus-odd}),
  it is only $D$ which is exceptional; the velocity $C=c_{2K-1}$ follows
  the general pattern $\dot x_{2(K+1-j)-1} \sim c_{2j-1}$.
  Similar remarks apply to $x_{2K}(t)$: it has the expected velocity $c_{2K-1}$
  as $t \to -\infty$,
  but apart from that, its asymptotics
  deviate from the pattern followed by the other even-numbered peakons.
\end{remark}

\begin{proof}[Proof of \autoref{thm:asymptotics-generalK-positions}]
  We will work out the details for $x_{2(K+1-j)}$ with $2 \le j \le K$
  as~$t \to +\infty$.
  The proofs for the other cases are entirely similar and will be omitted.
  The formulas for velocities follow from the ones from positions,
  because the $o(1)$ terms and their derivatives are actually
  bounded
  by factors of the form $e^{\mp\delta t}$ with $\delta > 0$,
  as $t \to \pm\infty$;
  see the proof of \autoref{lem:asymptotics-J}.

  From \autoref{thm:inverse-spectral-map-formulas}
  we have the exact formula for the solution:
  \begin{equation*}
    x_{2(K+1-j)}(t) = \frac12 \ln \left( \frac{2 \, \heineintegral_{j,j-1}^{00}(t)}{\heineintegral_{j-1,j-2}^{11}(t)} \right)
    .
  \end{equation*}
  The asymptotic behaviour of the factors as $t \to +\infty$
  is given by \autoref{lem:asymptotics-J}:
  \begin{equation*}
    \heineintegral_{j,j-1}^{00}(t)
    =
    \bigl( 1+o(1) \bigr) \,
    \Psi_{[j]\, [j-1]} \, a_{[j]}(t) \, b_{[j-1]}(t)
  \end{equation*}
  and
  \begin{equation*}
    \heineintegral_{j-1,j-2}^{11}
    =
    \bigl( 1+o(1) \bigr) \,
    \Psi_{[j-1]\, [j-2]} \, \lambda_{[j-1]} \, \mu_{[j-2]} \, a_{[j-1]}(t) \, b_{[j-2]}(t)
    .
  \end{equation*}
  It follows that
  \begin{equation*}
    \begin{split}
      x_{2(K+1-j)}(t)
      &=
      \frac12 \ln \left( \bigl( 1+o(1) \bigr) \, \frac{2 \, \Psi_{[j][j-1]} \, a_{[j]}(t) \, b_{[j-1]}(t)}{\Psi_{[j-1][j-2]} \, \lambda_{[j-1]} \, \mu_{[j-2]} \, a_{[j-1]}(t) \, b_{[j-2]}(t)} \right)
      \\
      &=
      \frac12 \ln \left( \frac{2 \, \Psi_{[j][j-1]} \, a_{j}(t) \, b_{j-1}(t)}{\Psi_{[j-1][j-2]} \, \lambda_{[j-1]} \, \mu_{[j-2]}} \right)
      + \frac12 \ln \bigl( 1+o(1) \bigr)
      \\
      &=
      \frac12 \ln \left( \frac{2 \, \Psi_{[j][j-1]} \, a_{j}(0) \, b_{j-1}(0) \, e^{t/\lambda_j} \, e^{t/\mu_{j-1}}}{\Psi_{[j-1][j-2]} \, \lambda_{[j-1]} \, \mu_{[j-2]}} \right)
      + o(1)
      \\
      &=
      \frac{t}{2} \left( \frac{1}{\lambda_j} + \frac{1}{\mu_{j-1}} \right)
      +\frac12 \ln \frac{2 \Psi_{[j][j-1]} \, a_j(0) \, b_{j-1}(0)}{\Psi_{[j-1][j-2]} \, \lambda_{[j-1]} \, \mu_{[j-2]}}
      + o(1)
      .
    \end{split}
  \end{equation*}
  We now obtain the claimed formula by expanding the definition
  of $\Psi_{IJ}$ and
  cancelling all common factors from the ratios involving $\Delta_I^2$,
  $\twin \Delta_J^2$ and~$\Gamma_{IJ}$:
  \begin{equation*}
    \begin{split}
      &
      \frac{\Psi_{[j][j-1]}}{\Psi_{[j-1][j-2]} \, \lambda_{[j-1]} \, \mu_{[j-2]}}
      \\
      &=
      \frac{\Delta_{[j]}^2}{\Delta_{[j-1]}^2}
      \times
      \frac{\twin\Delta_{[j-1]}^2}{\twin\Delta_{[j-2]}^2}
      \times
      \frac{\Gamma_{[j-1][j-2]}}{\Gamma_{[j][j-1]}}
      \times
      \frac{1}{\lambda_{[j-1]} \, \mu_{[j-2]}}
      \\
      &=
      \frac{
        \displaystyle
        \left( \prod_{r=1}^{j-1} (\lambda_r - \lambda_j)^2 \right)
        \left( \prod_{s=1}^{j-2} (\mu_s - \mu_{j-1})^2 \right)
      }{
        \displaystyle
        \left( \prod_{r=1}^{j-1} (\lambda_r + \mu_{j-1}) \right)
        \left( \prod_{s=1}^{j-2} (\lambda_j + \mu_s) \right)
        (\lambda_j + \mu_{j-1})
        \left( \prod_{r=1}^{j-1} \lambda_r \right)
        \left( \prod_{s=1}^{j-2} \mu_s \right)
      }
      \\
      &=
      \left( \prod_{r=1}^{j-1} \frac{(\lambda_r - \lambda_j)^2}{\lambda_r (\lambda_r + \mu_{j-1})} \right)
      \left( \prod_{s=1}^{j-2} \frac{(\mu_s - \mu_{j-1})^2}{(\lambda_j + \mu_s) \mu_s} \right)
      \frac{1}{\lambda_j + \mu_{j-1}}
      .
    \end{split}
  \end{equation*}
  (The empty product appearing when $j=2$ should be read as having the value~$1$.)
\end{proof}

\begin{corollary}[Phase shifts for positions]
  \label{cor:phaseshift-positions}
  The following formulas hold for the $K+K$ interlacing Geng--Xue
  peakon solution with $K \ge 2$:
  \begin{subequations}
    \begin{equation}
      \label{eq:phaseshift-pos-fastest}
      \begin{split}
        &
        \lim_{t \to +\infty} \bigl( x_a(t) - c_1 t \bigr)
        - \lim_{t \to -\infty} \bigl( x_b(t) - c_1 t \bigr)
        \\ &
        \qquad
        =
        - \sum_{r=2}^K R_{r1}'
        - \sum_{s=2}^{K-1} S_{s1}'
        ,\qquad
        \text{for $a \in \bigl\{ 2K-1, \, 2K \bigr\}$ and $b \in \bigl\{ 1, \, 2 \bigr\}$}
        ,
      \end{split}
    \end{equation}
    \begin{equation}
      \label{eq:phaseshift-pos-odd}
      \begin{split}
        &
        \lim_{t \to +\infty} \bigl( x_{2(K+1-j)-1}(t) - c_{2j-1} t \bigr)
        - \lim_{t \to -\infty} \bigl( x_{2j}(t) - c_{2j-1} t \bigr)
        \\ &
        \qquad
        =
        \sum_{r=1}^{j-1} R_{rj}'
        - \sum_{r=j+1}^K R_{rj}'
        + \sum_{s=1}^{j-1} S_{sj}'
        - \sum_{s=j+1}^{K-1} S_{sj}'
        ,\qquad
        \text{for $j=1,\dots,K-1$}
        ,
      \end{split}
    \end{equation}
    \begin{equation}
      \label{eq:phaseshift-pos-even}
      \begin{split}
        &
        \lim_{t \to +\infty} \bigl( x_{2(K+1-j)}(t) - c_{2j-2} t \bigr)
        - \lim_{t \to -\infty} \bigl( x_{2j-1}(t) - c_{2j-2} t \bigr)
        \\ &
        \qquad
        =
        \sum_{r=1}^{j-1} R_{rj}''
        - \sum_{r=j+1}^K R_{rj}''
        + \sum_{s=1}^{j-2} S_{sj}''
        - \sum_{s=j}^{K-1} S_{sj}''
        ,\qquad
        \text{for $j=2,\dots,K-1$}
        ,
      \end{split}
    \end{equation}
    and
    \begin{equation}
      \label{eq:phaseshift-pos-slowest}
      \begin{split}
        &
        \lim_{t \to +\infty} \bigl( x_{1}(t) - c_{2K-1} t \bigr)
        - \lim_{t \to -\infty} \bigl( x_{2K}(t) - c_{2K-1} t \bigr)
        \\ &
        \qquad
        =
        \frac12 \ln \frac{M}{2L}
        - \frac12 \ln \bigl( b_{\infty} \, b_{\infty}^* \bigr)
        + \frac12 \sum_{r=1}^{K-1} \ln \frac{(\lambda_r - \lambda_K)^2}{\lambda_r (\lambda_K + \mu_r)}
        .
      \end{split}
    \end{equation}
  \end{subequations}
\end{corollary}

\begin{remark}
  This way of writing the formulas is slightly redundant,
  since the case $j=1$ of~\eqref{eq:phaseshift-pos-odd}
  is already included in~\eqref{eq:phaseshift-pos-fastest}
  as the case $a=2K-1$, $b=2$.
  The purpose of including $j=1$ in~\eqref{eq:phaseshift-pos-odd}
  is to show that the pattern persists,
  and the purpose of~\eqref{eq:phaseshift-pos-fastest}
  is to emphasize that the curves $x=x_1(t)$ and $x=x_2(t)$
  have a common asymptote as $t \to -\infty$,
  which is parallel to the
  common asymptote of $x=x_{2K-1}(t)$ and $x=x_{2K}(t)$ as $t \to +\infty$.
\end{remark}

\subsection{Asymptotics for amplitudes}

Next, we turn to the asymptotics of the
amplitudes $m_k$ and~$n_k$,
which exhibit exponential growth or decay (or tend to constant
values in borderline cases),
so that their logarithms asymptotically behave like straight lines.
To obtain concise formulas, we make definitions similar to
\autoref{def:notation-asymptotics-positions} above.

\begin{definition}
  \label{def:notation-asymptotics-amplitudes}
  Let
  \begin{equation}
    \label{eq:asymptotic-amplitude-slopes}
    \begin{aligned}
      d_{2j}
      &=
      \frac12 \left( \frac{1}{\lambda_{j+1}} - \frac{1}{\mu_{j}} \right)
      ,\quad
      j = 1, \dots, K-1
      ,\\[1.5ex]
      d_{2j-1}
      &=
      \begin{cases}
        \displaystyle
        \frac12 \left( \frac{1}{\lambda_{j}} - \frac{1}{\mu_{j}} \right)
        ,&
        j = 1, \dots, K-1
        ,\\[2ex] \displaystyle
        \frac12 \left( \frac{1}{\lambda_{K}} \right)
        ,&
        j = K
        ,
      \end{cases}
    \end{aligned}
  \end{equation}
  and
  \begin{equation}
    \begin{aligned}
      P_{rj}' &= \frac12 \ln \frac{\bigl( \lambda_r - \lambda_j \bigr)^2 (\lambda_r + \mu_j)}{\lambda_r^3}
      ,\\
      P_{rj}'' &= \frac12 \ln \frac{\bigl( \lambda_r - \lambda_j \bigr)^2 (\lambda_r + \mu_{j-1})}{\lambda_r^3}
      ,\\
      Q_{sj}' &= \frac12 \ln \frac{\bigl( \mu_s - \mu_j \bigr)^2 (\lambda_j + \mu_s)}{\mu_s^3}
      ,\\
      Q_{sj}'' &= \frac12 \ln \frac{\bigl( \mu_s - \mu_{j-1} \bigr)^2 (\lambda_j + \mu_s)}{\mu_s^3}
      .
    \end{aligned}
  \end{equation}
\end{definition}

\begin{theorem}[Asymptotics for amplitudes]
  \label{thm:asymptotics-generalK-amplitudes}
  \begin{subequations}
  In terms of the abbreviations of \autoref{def:notation-asymptotics-amplitudes},
  the amplitudes in the $K+K$ interlacing Geng--Xue peakon solution
  with $K \ge 2$
  satisfy the following asymptotic formulas.

  Asymptotics for amplitudes as $t \to -\infty$\,:
  \begin{equation}
    \ln m_{2j-1}(t) =
    \begin{cases}
      \displaystyle
      d_1 t
      + \ln \frac{\mu_1}{\lambda_1}
      + \frac12 \ln \frac{a_1(0)}{2 \, b_1(0) \, (\lambda_1 + \mu_1)}
      & \\ \qquad \displaystyle
      + \sum_{r=2}^K P_{r1}'
      - \sum_{s=2}^{K-1} Q_{s1}'
      + o(1)
      ,&
      j = 1
      ,\\[2em] \displaystyle
      d_{2j-2} \, t
      - \ln \lambda_j
      + \frac12 \ln \frac{a_j(0) \, (\lambda_j + \mu_{j-1})}{2 \, b_{j-1}(0)}
      & \\ \qquad \displaystyle
      + \sum_{r=j+1}^K P_{rj}''
      - \sum_{s=j}^{K-1} Q_{sj}''
      + o(1)
      ,&
      j = 2,\dots,K
      ,
    \end{cases}
  \end{equation}
  and
  \begin{equation}
    -\ln n_{2j}(t) =
    \begin{cases}
      \displaystyle
      d_{2j-1} \, t
      + \ln \mu_j
      + \frac12 \ln \frac{2 \, a_j(0)}{b_j(0) \, (\lambda_j + \mu_j)}
      & \\ \qquad \displaystyle
      + \sum_{r=j+1}^K P_{rj}'
      - \sum_{s=j+1}^{K-1} Q_{sj}'
      + o(1)
      ,&
      j=1,\dots,K-1
      ,\\[2em] \displaystyle
      d_{2K-1} \, t
      + \frac12 \ln \frac{2 a_K(0)}{b_{\infty}}
      + o(1)
      ,&
      j= K
      ,
    \end{cases}
  \end{equation}
  Asymptotics for amplitudes as $t \to +\infty$\,:
  \begin{equation}
    -\ln n_{2(K+1-j)}(t) =
    \begin{cases}
      \displaystyle
      d_1 t
      + \frac12 \ln \frac{2 \, a_1(0) \, (\lambda_1 + \mu_1)}{b_1(0)}
      + o(1)
      ,&
      j = 1
      ,\\[2em] \displaystyle
      d_{2j-2} \, t
      +\ln \mu_{j-1}
      + \frac12 \ln \frac{2 \, a_j(0)}{b_{j-1}(0) \, (\lambda_j + \mu_{j-1})}
      & \\ \qquad \displaystyle
      + \sum_{r=1}^{j-1} P_{rj}''
      - \sum_{s=1}^{j-2} Q_{sj}''
      + o(1)
      ,&
      j = 2,\dots,K
      ,
    \end{cases}
  \end{equation}
  and
  \begin{equation}
    \ln m_{2(K+1-j)-1}(t) =
    \begin{cases}
      \displaystyle
      d_{2j-1} \, t
      - \ln \lambda_j
      + \frac12 \ln \frac{a_j(0) \, (\lambda_j + \mu_j)}{2 \, b_j(0)}
      & \\ \qquad \displaystyle
      + \sum_{r=1}^{j-1} P_{rj}'
      - \sum_{s=1}^{j-1} Q_{sj}'
      + o(1)
      ,&
      j=1,\dots,K-1
      ,\\[2em] \displaystyle
      d_{2K-1} \, t
      + \frac12 \ln \left( b_{\infty}^* \, a_K(0) \, \frac{M}{L} \right)
      & \\[1em] \qquad \displaystyle
      + \frac12 \sum_{r=1}^{K-1} \ln \frac{\bigl( \lambda_r - \lambda_K \bigr)^2}{\lambda_r (\lambda_K + \mu_r)}
      ,&
      j= K
      .
    \end{cases}
  \end{equation}
  \end{subequations}
\end{theorem}

\begin{proof}
  Calculations very similar to those in the proof of
  \autoref{thm:asymptotics-generalK-positions}.
  We omit the details.
\end{proof}

\begin{corollary}[Phase shifts for amplitudes]
  \label{cor:phaseshift-amplitudes}
  \begin{subequations}
  The following formulas hold for the $K+K$ interlacing Geng--Xue
  peakon solution with $K \ge 2$:
  \begin{equation}
    \begin{split}
      &
      \lim_{t \to +\infty} \Bigl( -\ln n_{2K}(t) - d_1 t \Bigl)
      - \lim_{t \to -\infty} \Bigl( \ln m_1(t) - d_1 t \Bigr)
      \\ & \qquad
      =
      \ln \frac{2 \lambda_1 (\lambda_1 + \mu_1)}{\mu_1}
      - \sum_{r=2}^K P_{r1}'
      + \sum_{s=2}^{K-1} Q_{s1}'
      ,
    \end{split}
  \end{equation}
  \begin{equation}
    \begin{split}
      &
      \lim_{t \to +\infty} \Bigl( -\ln n_{2(K+1-j)}(t) - d_{2j-2} \, t \Bigl)
      - \lim_{t \to -\infty} \Bigl( \ln m_{2j-1}(t) - d_{2j-2} \, t \Bigr)
      \\ & \qquad
      =
      \ln \frac{2 \lambda_j \mu_{j-1}}{\lambda_j + \mu_{j-1}}
      + \sum_{r=1}^{j-1} P_{rj}''
      - \sum_{r=j+1}^K P_{rj}''
      \\ & \qquad\quad
      - \sum_{s=1}^{j-2} Q_{sj}''
      + \sum_{s=j}^{K-1} Q_{sj}''
      ,\qquad
      \text{for $j=2,\dots,K$}
      ,
    \end{split}
  \end{equation}
  \begin{equation}
    \begin{split}
      &
      \lim_{t \to +\infty} \Bigl( \ln m_{2(K+1-j)-1}(t) - d_{2j-1} \, t \Bigr)
      - \lim_{t \to -\infty} \Bigl( -\ln n_{2j}(t) - d_{2j-1} \, t \Bigl)
      \\ & \qquad
      =
      - \ln \frac{2 \lambda_j \mu_j}{\lambda_j + \mu_j}
      + \sum_{r=1}^{j-1} P_{rj}'
      - \sum_{r=j+1}^K P_{rj}'
      \\ & \qquad\quad
      - \sum_{s=1}^{j-1} Q_{sj}'
      + \sum_{s=j+1}^{K-1} Q_{sj}'
      ,\qquad
      \text{for $j=1,\dots,K-1$}
      ,
    \end{split}
  \end{equation}
  and
  \begin{equation}
    \begin{split}
      &
      \lim_{t \to +\infty} \Bigl( \ln m_1(t) - d_{2K-1} \, t \Bigr)
      - \lim_{t \to -\infty} \Bigl( -\ln n_{2K}(t) - d_{2K-1} \, t \Bigl)
      \\ & \qquad
      =
      \frac12 \ln \frac{M}{2L}
      + \frac12 \ln \bigl( b_{\infty} \, b_{\infty}^* \bigr)
      + \frac12 \sum_{r=1}^{K-1} \ln \frac{\bigl( \lambda_r - \lambda_K \bigr)^2}{\lambda_r (\lambda_K + \mu_r)}
      .
    \end{split}
  \end{equation}
  \end{subequations}
\end{corollary}

\section{Concluding remarks}
\label{sec:conclusions}

In this article we have given a fairly complete treatment
of interlacing pure $(K+K)$-peakon solutions of the two-component
Geng--Xue equation.
To our knowledge, this is the first multi-component peakon equation
for which the peakon ODEs have been solved explicitly.
The third-order inverse spectral problem used for deriving the explicit
solution formulas involves two Lax pairs and correspondingly
two spectra~\cite{lundmark-szmigielski:GX-inverse-problem}.
(There are some similarities to the $3\times 3$ inverse problems
studied by Kaup and collaborators,
for example in~\cite{kaup-vangorder:nondegenerate-3x3-operator},
and it might be interesting to investigate whether there are
any deeper connections.)
This is the first application to peakon equations
of the theory of Cauchy biorthogonal polynomials
\cite{bertola-gekhtman-szmigielski:cubicstring,
  bertola-gekhtman-szmigielski:cauchy}
in its full generality, where the polynomials are
biorthogonal with respect to two \emph{independent} spectral measures.

The interlacing peakon solutions display the Toda-like
asymptotic properties typical of other peakon equations,
with the peakons having the same asymptotic velocities
when $t\to+\infty$ as when $t\to-\infty$,
but in the opposite order,
and the velocities depending only on the eigenvalues in the two spectra.
A noteworthy feature, however, is that the peakons don't scatter completely;
instead, the two fastest peakons have the same asymptotic velocity,
and in fact the distance between them tends to zero.
Another interesting phenomenon is that the amplitudes of the peakons
grow or decay exponentially, instead of just approaching constant
values as is usually the case for peakons,
and that the logarithms of the amplitudes exhibit similar
scattering and phase shifts as the positions.

In the case of non-interlacing peakon configurations,
to be studied in a separate article,
we anticipate that there will be even less scattering,
and more peakons clustering together with the same asymptotic velocity.

We have also reported on the discovery that the Geng--Xue equation admits
discontinuous shockpeakon solutions,
like the Degasperis--Procesi equation.
Although we were able to integrate the $1+1$ shockpeakon ODEs,
there are still several questions open for further investigation,
concerning for example
the status of the Lax pairs in the context of shockpeakons,
formation of shockpeakons at peakon--antipeakon collisions,
continuation of solutions past singularities,
and the possibility of allowing overlapping peakon or shockpeakon
solutions.

\appendix

\section{Distributional interpretation of the Lax pair} 
\label{sec:lax-pair}

The aim of this appendix is to derive the Geng--Xue
shockpeakon ODEs~\eqref{eq:GX-shockpeakon-ode},
which of course include the peakon ODEs~\eqref{eq:GX-peakon-ode}
as a special case,
and to verify that the Lax pairs
\eqref{eq:laxI} and~\eqref{eq:laxII} are valid not only for smooth solutions,
but also for peakon solutions. In other words, we want to show that these Lax pairs,
if given an appropriate distributional interpretation, really are equivalent to the
peakon ODEs~\eqref{eq:GX-peakon-ode} when $u(x,t)$ and~$v(x,t)$
are given by the peakon ansatz~\eqref{eq:GXpeakons}
with non-overlapping peakons.
The analysis is similar to the one done for Novikov peakons in
Appendix~B of our previous article~\cite{hone-lundmark-szmigielski:novikov}.

\subsection{Notation}
\label{sec:notation-distributions}

The functions $u$ and~$v$ considered in this article, as well as
all their partial derivatives in the classical sense,
belong to a certain family of piecewise smooth functions
which we will denote by ~$\piecewiseclass$.
We say that a function $f(x,t)$ belongs to~$\piecewiseclass$
if there are finitely many smooth curves
$\{ x=x_k(t), \, t \in \R \}_{k=1}^N$
such that $x_1(t) < \dots < x_N(t)$,
so that they divide the $(x,t)$ plane into $N+1$ open regions $\{ \Omega_k \}_{k=0}^N$,
and if the restriction of~$f$ to each $\Omega_k$
is a smooth function of $x$ and~$t$ which can be extended to a smooth function
on a neighbourhood of the closure~$\overline{\Omega}_k$.
We do not require that $f(x,t)$ is defined on the curves $x=x_k(t)$.
However, the assumptions imply that the left and right limits of~$f$
exist at every~$x$, and they will be denoted by $f(x^-)$ and~$f(x^+)$,
respectively.
Hence we can define the jump function and the average function,
\begin{equation}
  \label{eq:jump-avg-notation}
  \jump{f}(x) := f(x^+) - f(x^-)
  ,\qquad
  \avg{f}(x) := \frac{f(x^+) + f(x^-)}{2},
\end{equation}
which belong to~$\piecewiseclass$, 
are defined everywhere,
and satisfy $\avg{f}=f$ and $\jump{f}=0$
away from the curves $x=x_k(t)$.
In addition, we have the product rules
\begin{equation}
  \label{eq:jump-avg-products}
  \jump{fg} = \avg{f} \jump{g} + \jump{f} \avg{g},
  \qquad
  \avg{fg} = \avg{f} \avg{g} + \tfrac14 \jump{f} \jump{g}.
\end{equation}
The class~$\piecewiseclass$ is closed under partial differentiation
in the classical sense, and
we will use subscripts to denote such partial derivatives,
for example $f_x$ or~$f_{xxt}$.

On the other hand, we can also interpret functions in~$\piecewiseclass$
as distributions:
for each fixed $t$, the function $x \mapsto f(x,t)$ defines a regular distribution
in the class~$\spaceDprime$,
by acting on test functions $\varphi(x)$ in the usual way,
\begin{equation*}
  \langle f, \varphi \rangle
  = \int_{\R} f(x,t) \, \varphi(x) \, dx
  .
\end{equation*}
The notation $D_x f$ will denote the distributional derivative:
\begin{equation*}
  \langle D_x f, \varphi \rangle
  = - \langle f, \varphi_x \rangle
  .
\end{equation*}
Note that we view $f$ and~$D_x f$ as distributions with respect to the variable~$x$,
depending only parametrically on~$t$.
Therefore the time derivative $D_t f$ is defined differently,
as a limit in the space~$\spaceDprime$:
\begin{equation}
  \label{eq:def-Dt}
  D_t f(\cdot,t) = \lim_{\tau \to 0} \frac{f(\cdot,t+\tau) - f(\cdot,t)}{\tau}
  .
\end{equation}
This limit, provided it exists, commutes with~$D_x$ by the continuity of~$D_x$ on~$\spaceDprime$.

If $f\in \piecewiseclass$ with (possible) jump discontinuities
at $\{ x = x_k(t)\}_{k=1}^N$, then 
\begin{equation*}
  D_x f(\cdot,t) = f_x(\cdot,t) + \sum_{k=1}^N \jump{f}(\cdot,x_k(t)) \, \delta_{x_k(t)}
  ,
\end{equation*}
or
\begin{equation}
  \label{eq:Dx-jump}
  D_x f = f_x + \sum_{k=1}^N \jump{f}(x_k) \, \delta_{x_k}
\end{equation}
for short.
(Here, of course, $\delta_a$ denotes the Dirac delta at the point $x=a$.)
Moreover,
\begin{equation}
  \label{eq:Dt-jump}
  D_t f = f_t - \sum_{k=1}^n \dot x_k \, \jump{f}(x_k) \, \delta_{x_k},
\end{equation}
where $\dot x_k = dx_k / dt$.

We also note that
\begin{equation*}
  \tfrac{d}{dt} f(x_k(t)^{\pm},t)
  = f_x(x_k(t)^{\pm},t) \, \dot x_k(t) + f_t(x_k(t)^{\pm},t)
  ,
\end{equation*}
which gives
\begin{equation}
  \label{eq:chainrule-jump-avg}
  \begin{split}
    \tfrac{d}{dt} \bigl( \jump{f}(x_k) \bigr)
    &=
    \jump{f_x}(x_k) \, \dot x_k + \jump{f_t}(x_k)
    , \\
    \tfrac{d}{dt} \bigl( \avg{f}(x_k) \bigr)
    &=
    \avg{f_x}(x_k) \, \dot x_k + \avg{f_t}(x_k)
    .
  \end{split}
\end{equation}
Finally, we remark that the discussion above 
generalizes easily to the case of matrix-valued functions with entries
in~$\piecewiseclass$.
For example, if $A$ and~$B$ are two matrices 
with entries from $\piecewiseclass$, and the matrix product $AB$ is defined,
then 
\begin{equation*}
  \jump{AB} = \avg{A} \jump{B} + \jump{A} \avg{B},
  \qquad
  \avg{AB} = \avg{A} \avg{B} + \tfrac14 \jump{A} \jump{B}.
\end{equation*}
Likewise, equation \eqref{eq:chainrule-jump-avg} generalizes to matrices.

\subsection{Derivation of the shockpeakon ODEs}
\label{sec:proof-shockpeakon-odes}

In this section we prove \autoref{thm:GX-shockpeakons},
which says that the shockpeakon ansatz~\eqref{eq:GXshockpeakons}
is a solution of the distributional Geng--Xue equation
\eqref{eq:GX-distributional}
if and only if it is non-overlapping and satisfies the
shockpeakon ODEs~\eqref{eq:GX-shockpeakon-ode}.

\begin{proof}[Proof of \autoref{thm:GX-shockpeakons}.]
  A function $u$ given by the shockpeakon ansatz
  is piecewise of the form
  \begin{equation*}
    A \, e^x + B \, e^{-x}
    ,
  \end{equation*}
  which implies that $\frac12 u^2$
  is piecewise of the form
  \begin{equation*}
    \tfrac12 A^2 \, e^{2x} + AB + \tfrac12 B^2 \, e^{-2x}
    ,
  \end{equation*}
  which lies in the kernel of the differential operator
  $(4-D_x^2) D_x$.
  Thus, the expression
  $(4-D_x^2) D_x (\tfrac12 u^2)$
  vanishes identically away from the points $x=x_k$,
  and will therefore be a purely singular distribution:
  a linear combination of
  $\{ \delta_{x_k}, \delta'_{x_k}, \delta''_{x_k} \}_{k=1}^N$
  resulting from differentiating the jump discontinuities of $\frac12 u^2$
  at those points.

  Using the notations for jump~$\jump{u}$ and average~$\avg{u}$
  as in~\eqref{eq:jump-avg-notation},
  we immediately have
  \begin{equation}
    \label{eq:u-ux-jump}
    \jump{u}(x_k) = -2 s_k
    ,\qquad
    \jump{u_x}(x_k) = -2 m_k
    ,
  \end{equation}
  and we also recall from
  \eqref{eq:uv-shorthand-shock} and~\eqref{eq:u-ux-are-averages}
  that we defined the notation $u(x_k)$ and~$u_x(x_k)$
  simply as abbreviations for the averages,
  \begin{equation}
    u(x_k) := \avg{u}(x_k)
    ,\qquad
    u_x(x_k) := \avg{u_x}(x_k)
    .
  \end{equation}
  From the rules~\eqref{eq:jump-avg-products}, we then find the jump in~$\frac12 u^2$
  at $x=x_k$:
  \begin{equation*}
    \begin{split}
      \jump{\tfrac12 u^2}(x_k)
      &
      = \jump{u}(x_k) \cdot \avg{u}(x_k)
      \\ &
      = (-2 s_k) \cdot u(x_k)
      .
    \end{split}
  \end{equation*}
  Each such jump contributes a Dirac delta to the distributional
  derivative,
  so the singular part of the distribution
  $D_x (\tfrac12 u^2)$
  is
  \begin{equation*}
    -2 \sum_{k=1}^N s_k \, u(x_k) \, \delta_{x_k}
    ,
  \end{equation*}
  and the regular part is just
  the function $uu_x$ (the classical partial derivative of $\frac12 u^2$,
  defined away from the points $x=x_k$).
  Thus, as a distribution,
  \begin{equation}
    \label{eq:shock-first-derivative}
    D_x (\tfrac12 u^2)
    = uu_x - 2 \sum_{k=1}^N s_k \, u(x_k) \, \delta_{x_k}
    .
  \end{equation}
  In the next step, we need
  \begin{equation*}
    \begin{split}
      \jump{u u_x}(x_k)
      &
      = \jump{u}(x_k) \cdot \avg{u_x}(x_k)
      + \avg{u}(x_k) \cdot \jump{u_x}(x_k)
      \\ &
      = (-2 s_k) \cdot u_x(x_k) + u(x_k) \cdot (- 2 m_k)
    ,
    \end{split}
  \end{equation*}
  together with the fact that $u_{xx}=u$ for $x \neq x_k$
  (since $u$  is piecewise of the form $A \, e^x + B \, e^{-x}$).
  Using this, we find
  when differentiating \eqref{eq:shock-first-derivative} that,
  as a distribution,
  \begin{equation}
    \label{eq:shock-second-derivative}
    \begin{split}
      D_x^2 (\tfrac12 u^2)
      &
      = u_x^2 + uu_{xx}
      + \sum_{k=1}^N \jump{uu_x}(x_k) \, \delta_{x_k}
      - 2 \sum_{k=1}^N s_k \, u(x_k) \, \delta'_{x_k}
      \\ &
      = u^2 + u_x^2
      - 2 \sum_{k=1}^N \bigl( m_k \, u(x_k) + s_k \, u_x(x_k) \bigr) \, \delta_{x_k}
      - 2 \sum_{k=1}^N s_k \, u(x_k) \, \delta'_{x_k}
      .
    \end{split}
  \end{equation}
  For the final differentiation,
  we can reuse the result~\eqref{eq:shock-first-derivative} for
  $D_x (u^2) = 2 \, D_x (\frac12 u^2)$,
  and we also need
  \begin{equation*}
    \begin{split}
      \jump{u_x^2}(x_k)
      &
      = 2 \jump{u_x}(x_k) \cdot \avg{u_x}(x_k)
      \\ &
      = 2 \cdot (- 2 m_k) \cdot u_x(x_k)
    .
    \end{split}
  \end{equation*}
  Upon differentiating \eqref{eq:shock-second-derivative},
  this gives
  \begin{equation}
    \label{eq:shock-third-derivative}
    \begin{split}
      D_x^3 (\tfrac12 u^2)
      &
      = 2 \, D_x (\tfrac12 u^2)
      + 2 u_x u_{xx}
      + \sum_{k=1}^N \jump{u_x^2}(x_k) \, \delta_{x_k}
      \\ & \quad
      - 2 \sum_{k=1}^N \bigl( m_k \, u(x_k) + s_k \, u_x(x_k) \bigr) \, \delta'_{x_k}
      - 2 \sum_{k=1}^N s_k \, u(x_k) \, \delta''_{x_k}
      \\ &
      = 4 uu_x 
      - 4 \sum_{k=1}^N \bigl( s_k \, u(x_k) + m_k \, u_x(x_k) \bigr) \, \delta_{x_k}
      \\ & \quad
      - 2 \sum_{k=1}^N \bigl( m_k \, u(x_k) + s_k \, u_x(x_k) \bigr) \, \delta'_{x_k}
      - 2 \sum_{k=1}^N s_k \, u(x_k) \, \delta''_{x_k}
      .
    \end{split}
  \end{equation}
  Combining \eqref{eq:shock-first-derivative}
  and~\eqref{eq:shock-third-derivative},
  we get
  \begin{equation}
    \label{eq:DP-operator}
    \begin{split}
      (4-D_x^2) D_x (\tfrac12 u^2)
      &
      = 4 D_x (\tfrac12 u^2) - D_x^3 (\tfrac12 u^2)
      \\ &
      =
      4 \biggl( uu_x - 2 \sum_{k=1}^N s_k \, u(x_k) \, \delta_{x_k} \biggr)
      \\ & \quad
      -
      \biggl( 4 uu_x - 4 \sum_{k=1}^N \bigl( s_k \, u(x_k) + m_k \, u_x(x_k) \bigr) \, \delta_{x_k}
      \\ & \qquad
      - 2 \sum_{k=1}^N \bigl( m_k \, u(x_k) + s_k \, u_x(x_k) \bigr) \, \delta'_{x_k}
      - 2 \sum_{k=1}^N s_k \, u(x_k) \, \delta''_{x_k}
      \biggr)
      \\ &
      =
      4 \sum_{k=1}^N \bigl( -s_k \, u(x_k) + m_k \, u_x(x_k) \bigr) \, \delta_{x_k}
      \\ & \quad
      + 2 \sum_{k=1}^N \bigl( m_k \, u(x_k) + s_k \, u_x(x_k) \bigr) \, \delta'_{x_k}
      + 2 \sum_{k=1}^N s_k \, u(x_k) \, \delta''_{x_k}
      .
    \end{split}
  \end{equation}
  (Notice that the regular parts cancel out, as predicted.)

  Indices~$k$ for which $m_k = s_k = 0$ give no contribution
  to the sums here.
  Let $\mathcal{K} \subset \{ 1,2,\dots,N \}$
  be the set of the remaining indices (those that do contribute);
  then we can replace $\sum_{k=1}^N$ with $\sum_{k \in \mathcal{K}}$ above.
  Because of the non-overlapping assumption,
  $v$ is smooth near $x_k$ for $k \in \mathcal{K}$,
  so the values $v(x_k)$, $v_x(x_k)$ and $v_{xx}(x_k) = v(x_k)$ exist
  for $k \in \mathcal{K}$
  (and they coincide with the averages $\avg{v}(x_k)$ and $\avg{v_x}(x_k)$,
  so there is no conflicting notation).
  Thus, multiplying \eqref{eq:DP-operator} by $v(x,t)$ according to the
  rules
  \begin{equation*}
    \begin{aligned}
      v(x) \, \delta_a &= v(a) \, \delta_a
      ,\\
      v(x) \, \delta'_a &= v(a) \, \delta'_a - v'(a) \, \delta_a
      ,\\
      v(x) \, \delta''_a &= v(a) \, \delta''_a - 2 v'(a) \, \delta'_a +  v''(a) \, \delta_a
      ,
    \end{aligned}
  \end{equation*}
  we obtain
  \begin{equation}
    \label{eq:v-times-DP-operator}
    \begin{split}
      v \cdot (4-D_x^2) D_x (\tfrac12 u^2)
      &
      =
      4 \sum_{k \in \mathcal{K}} \bigl( -s_k \, u(x_k) + m_k \, u_x(x_k) \bigr) \, v(x_k) \, \delta_{x_k}
      \\ & \quad
      + 2 \sum_{k \in \mathcal{K}} \bigl( m_k \, u(x_k) + s_k \, u_x(x_k) \bigr) \,
      \bigl( v(x_k) \, \delta'_{x_k} - v_x(x_k) \, \delta_{x_k} \bigr)
      \\ & \quad
      + 2 \sum_{k \in \mathcal{K}} s_k \, u(x_k) \,
      \bigl( v(x_k) \, \delta''_{x_k} - 2 v_x(x_k) \, \delta'_{x_k} + v_{xx}(x_k) \, \delta_{x_k} \bigr)
      \\ &
      =
      2 \sum_{k \in \mathcal{K}} \Bigl( \bigl( -s_k \, u(x_k) + 2m_k \, u_x(x_k) \bigr) \, v(x_k)
      \\ & \qquad\qquad
      - \bigl( m_k \, u(x_k) + s_k \, u_x(x_k) \bigr) \, v_x(x_k) \Bigr) \, \delta_{x_k}
      \\ & \quad
      + 2 \sum_{k \in \mathcal{K}} \Bigl( \bigl( m_k \, u(x_k) + s_k \, u_x(x_k) \bigr) \, v(x_k)
      \\ & \qquad\qquad
      - 2 s_k \, u(x_k) \, v_x(x_k) \Bigr) \, \delta'_{x_k}
      \\ & \quad
      + 2 \sum_{k \in \mathcal{K}} s_k \, u(x_k) \, v(x_k) \, \delta''_{x_k}
      .
    \end{split}
  \end{equation}
  The requirement of the distributional Geng--Xue equation~\eqref{eq:GX-distributional}
  is that this should equal $-D_t m$, where $m = u - D_x^2 u$.
  From
  \begin{equation*}
    D_x u = u_x + \sum_{k \in \mathcal{K}} \jump{u}(x_k) \, \delta_{x_k} 
    = u_x - 2 \sum_{k \in \mathcal{K}} s_k \, \delta_{x_k}
  \end{equation*}
  and
  \begin{equation*}
    \begin{split}
      D_x^2 u
      &
      = u_{xx} + \sum_{k \in \mathcal{K}} \jump{u_x}(x_k) \, \delta_{x_k}
      - 2 \sum_{k \in \mathcal{K}} s_k \, \delta'_{x_k}
      \\ &
      = u 
      - 2 \sum_{k \in \mathcal{K}} m_k \, \delta_{x_k}
      - 2 \sum_{k \in \mathcal{K}} s_k \, \delta'_{x_k}
    \end{split}
  \end{equation*}
  we get
  \begin{equation*}
    m =
    2 \sum_{k \in \mathcal{K}} m_k \, \delta_{x_k}
    + 2 \sum_{k \in \mathcal{K}} s_k \, \delta'_{x_k}
    ,
  \end{equation*}
  and thus
  \begin{equation}
    \label{eq:m-t}
    -D_t m =
    - 2 \sum_{k \in \mathcal{K}} \dot m_k \, \delta_{x_k}
    - 2 \sum_{k \in \mathcal{K}} \dot s_k \, \delta'_{x_k}
    + 2 \sum_{k \in \mathcal{K}} m_k \, \dot x_k \, \delta'_{x_k}
    + 2 \sum_{k \in \mathcal{K}} s_k \, \dot x_k \, \delta''_{x_k}
    .
  \end{equation}
  Identifying coefficients in \eqref{eq:v-times-DP-operator} and~\eqref{eq:m-t},
  we find, for $k \in \mathcal{K}$,
  \begin{equation*}
    \begin{aligned}
      -\dot m_k
      &
      =
      \bigl( - s_k \, u(x_k) + 2m_k \, u_x(x_k) \bigr) \, v(x_k)
      - \bigl( m_k \, u(x_k) + s_k \, u_x(x_k) \bigr) \, v_x(x_k)
      ,\\
      - \dot s_k + m_k \, \dot x_k
      &
      =
      \bigl( m_k \, u(x_k) + s_k \, u_x(x_k) \bigr) \, v(x_k)
      - 2 s_k \, u(x_k) \, v_x(x_k)
      ,\\
      s_k \, \dot x_k
      &
      =
      s_k \, u(x_k) \, v(x_k)
      .
    \end{aligned}
  \end{equation*}
  If $s_k \neq 0$, the third equation implies that
  $\dot x_k = u(x_k) \, v(x_k)$,
  and then the second equation reduces to
  $-\dot s_k = s_k \bigl( u_x(x_k) \, v(x_k) - 2 u(x_k) \, v_x(x_k) \bigr)$.
  If $s_k = 0$, then $m_k \neq 0$ (since $k \in \mathcal{K}$),
  and the second equation shows that
  $\dot x_k = u(x_k) \, v(x_k)$.
  So in either case, we can simplify the equations to
  \begin{equation}
    \begin{aligned}
      \dot x_k
      &
      =
      u(x_k) \, v(x_k)
      ,\\
      \dot m_k
      &
      =
      \bigl( s_k \, u(x_k) - 2m_k \, u_x(x_k) \bigr) \, v(x_k)
      + \bigl( m_k \, u(x_k) + s_k \, u_x(x_k) \bigr) \, v_x(x_k)
      \\ &
      =
      m_k \bigl( u(x_k) \, v_x(x_k)  - 2 u_x(x_k) \, v(x_k) \bigr)
      + s_k \bigl( u(x_k) \, v(x_k) + u_x(x_k) \, v_x(x_k) \bigr) 
      ,\\
      \dot s_k
      &
      =
      s_k \bigl( 2 u(x_k) \, v_x(x_k) - u_x(x_k) \, v(x_k) \bigr)
      ,
    \end{aligned}
  \end{equation}
  for $k \in \mathcal{K}$,
  in agreement with the claimed shockpeakon ODEs~\eqref{eq:GX-shockpeakon-ode}.
  (And for $k \notin \mathcal{K}$, we can include the same equations
  for $m_k$ and $s_k$ if we like, since they are consistent with $m_k=s_k=0$.)

  By symmetry (simply interchanging the roles of $u$ and~$v$),
  we immediately obtain the corresponding equations
  for $x_k$, $n_k$, $r_k$ with $k \notin K$.
\end{proof}

\begin{remark}
  Using \eqref{eq:u-ux-jump},
  one can write the shockpeakon ODEs~\eqref{eq:GX-shockpeakon-ode}
  for $k \in \mathcal{K}$ as
  \begin{equation}
    \begin{aligned}
    \tfrac{d}{dt} x_k &= \avg{u}v 
    ,\\
    \tfrac{d}{dt} {[u]} &= [u] \, \Bigl( 2 \avg{u}v_x - \avg{u_x} v \Bigr)
    ,\\
    \tfrac{d}{dt}{[u_x]} &=
    [u] \, \Bigl( \avg{u} v + \avg{u_x} v_x \Bigr)
    +[u_x] \, \Bigl( \avg{u} v_x - 2 \avg{u_x} v \Bigr)
    ,
    \end{aligned}
  \end{equation}
  where all evaluations of jumps and averages are carried out at~$x_k$,
  and similarly for $k \notin \mathcal{K}$
  but with $u$ and~$v$ interchanged.
  For comparison, the Degasperis--Procesi shockpeakon ODEs are
  \begin{equation}
    \begin{aligned}
      \tfrac{d}{dt} x_k &= \avg{u}
      ,\\
      \tfrac{d}{dt} {[u]} &= -[u] \avg{u_x}
      ,\\
      \tfrac{d}{dt} [u_x] &= 2 [u] \avg{u} - 2 [u_x] \avg{u_x}
      .
    \end{aligned}
  \end{equation}
\end{remark}

\subsection{Verification of the Lax pair for peakon solutions}

The peakon solutions considered in this article are of the form
\begin{equation}
  \label{eq:weak-uv}
  u(x,t) = \sum_{k=1}^K m_k(t) \, e^{-\abs{x-x_k(t)}}
  , \qquad
  v(x,t) = \sum_{k=1}^K n_k(t) \, e^{-\abs{x-y_k(t)}}
  . 
\end{equation}
(Here we have found it convenient to change the notation from the main text,
and use $x_k$ and~$y_k$ instead of $x_{2k-1}$ and~$x_{2k}$,
and $m_k$ and~$n_k$ instead of $m_{2k-1}$ and~$n_{2k}$.)

The functions $u$ and~$v$ both belong to the piecewise smooth class~$\piecewiseclass$.
They are continuous and satisfy
\begin{equation*}
  \begin{split}
    D_x u &= u_x = \sum_{k=1}^K m_k \, \sgn(x_k-x) \, e^{-\abs{x-x_k}},\\
    D_x^2 u &= D_x(u_x) = u_{xx} + \sum_{k=1}^K \jump{u_x}(x_k) \, \delta_{x_k}
    = u + \sum_{k=1}^K (-2m_k) \, \delta_{x_k},
  \end{split}
\end{equation*}
with analogous formulas for $v$.  
These formulas imply
\begin{equation}
  \label{eq:weak-mn}
  m = u - D_x^2 u = 2 \sum_{k=1}^K m_k \, \delta_{x_k}
  , \qquad
  n = v- D_x^2 v= 2 \sum_{k=1}^K n_k \, \delta_{y_k}
  .
\end{equation}
The Lax pair \eqref{eq:laxI-x}--\eqref{eq:laxI-t} will
involve the functions $u$, $v$, $u_x$, $v_x$, as well as the purely
singular distributions $m$ and~$n$.
We will take $\psi_1$, $\psi_2$, $\psi_3$ to be functions in $\piecewiseclass$,
and separate the regular (function) part from the singular
(Dirac delta) part.
Writing $\Psi = (\psi_1,\psi_2,\psi_3)^t$,
the formulation obtained in this way reads
\begin{equation}
  \label{eq:weak-lax}
  D_x \Psi = \widehat{L} \Psi,
  \qquad
  D_t \Psi = \widehat{A} \Psi,
\end{equation}
where
\begin{equation}
  \label{eq:weak-lax-x}
  \begin{gathered}
    \widehat{L} =
    L + 2z \left( \sum_{k=1}^K n_k \, \delta_{y_k} \right) E_{12}+
    2z\left( \sum_{k=1}^K m_k \, \delta_{x_k} \right) E_{23}
    ,\\[1ex]
    L =
    \begin{pmatrix}
      0 & 0 & 1 \\
      0 & 0 & 0 \\
      1 & 0 & 0
    \end{pmatrix},
    \quad
    E_{12}=
    \begin{pmatrix}
      0 & 1 & 0 \\
      0 & 0 & 0 \\
      0 & 0 & 0
    \end{pmatrix},
    \quad 
    E_{23}=
    \begin{pmatrix}
      0 & 0 & 0 \\
      0 & 0 & 1 \\
      0 & 0 & 0
    \end{pmatrix}
    ,
  \end{gathered}
\end{equation}
and
\begin{equation}
  \label{eq:weak-lax-t}
  \begin{gathered}
    \widehat{A} =
    A - 2z \left( \sum_{k=1}^K n_k \, u(y_k) \, v(y_k) \, \delta_{y_k} \right) E_{12}-2z \left( \sum_{k=1}^K m_k \, u(x_k) \, v(x_k) \, \delta_{x_k} \right) E_{23}
    ,\\[1ex]
    A =
    \begin{pmatrix}
      -v_xu & v_x/z & v_xu_x\\
      u /z & v_xu-vu_x-1/z^2 & -u_x/z \\
      -uv& v/z & vu_x
    \end{pmatrix}
    .
  \end{gathered}
\end{equation}
First we make a few general comments.
Note that \eqref{eq:weak-lax} involves multiplying
$(\psi_2,0,0)^t$ by~$\delta_{y_k}$
and multiplying $(0,\psi_3,0)^t$ by~$\delta_{x_k}$.
In the second case there is no problem, since the function
$x \mapsto \psi_3(x,t)$ will automatically be continuous
according the the third component of the vector equation
$D_x \Psi = \widehat{L} \Psi$.
But in the first case, the function
$x \mapsto \psi_2(x,t)$
may have a jump at $x=y_k$, so some value $\psi_2(y_k)$ must be assigned
in order for this operation to be well-defined,
and this assignment must be consistent with $D_x D_t \Psi = D_t D_x \Psi$.
However, this is only a problem if one tries to consider the general case of
overlapping supports.
In this article, the supports of $m$ and~$n$ are assumed to be disjoint, and
as a result $\psi_2$ is continuous at the points of support of $n$,
and no other assignment is needed.

\begin{theorem}
  \label{thm:distributionalLaxpair}
  Let $u$, $v$, $m$, $n$ be given 
  by \eqref{eq:weak-uv}--\eqref{eq:weak-mn} and assume 
  that the supports of $m$ and~$n$ are disjoint.
  Then \eqref{eq:weak-lax-x}--\eqref{eq:weak-lax-t}
  form a weak Lax pair whose compatibility condition is given by the peakon ODEs
  \begin{align*}
    \dot x_k &= u(x_k) \, v(x_k)
    , &
    \dot m_k &= m_k \, \Bigl( u(x_k) \, v_x(x_k) - 2 \avg{u_x}(x_k) \, v(x_k) \Bigr)
    , \\
    \dot y_k &= u(y_k) \, v(y_k)
    , &
    \dot n_k &= n_k \, \Bigl( v(y_k) \, u_x(y_k) - 2 \avg{v_x}(y_k) \, u(y_k) \Bigr)
    .
  \end{align*}
\end{theorem}

\begin{proof}
  An essential simplification is to observe that we can localize our computations to 
  a vicinity of one of the points of support of the measures,
  by considering test functions which are zero except in a small neighbourhood of such a point.
  Since the supports of $m$ and $n$ are disjoint,
  we will only deal with one type of computation;
  either involving $E_{21}$ (for~$n$) or $E_{23}$ (for~$m$).
  Let us localize our computation around~$y_k$.
  Then we can omit the summation, as well as completely ignore the
  contribution coming from~$m$ but not from~$u$.
  (Some of the equalities below are a slight abuse of notation, which should be understood in the light of this remark.)

  We observe that
  $\psi_3$ is continuous at all the points of supports of both
  measures; in particular $\jump{\psi_3}(y_k)=0$. Likewise,
  $\jump{\psi_2}(y_k)=0$, even though $\jump{\psi_2}(x_k)=2z \, m_k \, \psi_3(x_k)$.
  Clearly, $\psi_1$ is not defined at $y_k$. Even though this will not
  impact the computation, we will set $\Psi(y_k)=\avg{\Psi}(y_k)$ for
  the duration of the computation.

  Next, we compute the derivatives of \eqref{eq:weak-lax}:
  \begin{equation*}
    \begin{split}
      D_t(D_x \Psi) &
      = D_t(L \Psi + 2z \, n_k \, \delta_{y_k} E_{12} \Psi)
      \\ &
      = L (\widehat{A} \Psi) + 2z \, E_{12} \tfrac{d}{dt} \bigl( n_k \Psi(y_k) \bigr) \, \delta_{y_k} - 2z \, E_{12} n_k \Psi(y_k) \dot y_k \delta'_{y_k}
      , \\
      D_x(D_t \Psi) &
      = D_x(A\Psi - 2z \, n_k \, u(y_k) \, v(y_k) \, E_{12} \Psi(y_k) \, \delta_{y_k})
      \\ &
      = (A\Psi)_x + \jump{A\Psi}(y_k) \delta_{y_k} - 2z \, n_k \, E_{12} \Psi(y_k) \, u(y_k) \, v(y_k) \, \delta'_{y_k}
      .
    \end{split}
  \end{equation*}
  The regular part of \eqref{eq:weak-lax} gives $\Psi_x = L\Psi$,
  so that $(A\Psi)_x=A_x \Psi + AL\Psi$, and it is easily verified
  that $LA=A_x+AL$ holds identically (since $u_{xx}=u$).
  This implies that the regular parts of the two expressions above are equal,
  and the terms involving $\delta'_{y_k}$ are also equal provided 
  $\dot y_k = u(y_k) \, v(y_k)$.
  Therefore the compatibility condition
  $D_t(D_x \Psi) = D_x(D_t \Psi)$
  reduces to an equality between the coefficients of~$\delta_{y_k}$,
  \begin{equation}
    \label{eq:gazonk}
    -2z \, n_k \, u(y_k) \, v(y_k) \, LE_{12} \Psi(y_k)
    + 2z \, E_{12} \tfrac{d}{dt} \bigl( n_k \Psi(y_k) \bigr)
    = \jump{A\Psi}(y_k).
  \end{equation}
  Using the product rule \eqref{eq:jump-avg-products},
  $\jump{\Psi}(y_k)=2z \, n_k \, E_{12}\Psi(y_k)$ and $\jump{v_x}(y_k) = -2n_k$,
  we find that the right-hand side of \eqref{eq:gazonk} equals
  \begin{multline}
    \avg{A}(y_k) \, 2z \, n_k \, E_{12}\Psi(y_k) + \jump{A}(y_k) \, \avg{\Psi}(y_k)
    =
    \\
    2z\,n_k
    \begin{pmatrix}
      0 & -u \avg{v_x} & 0\\
      0 & u/z & 0 \\
      0 & -uv & 0
    \end{pmatrix}_{\!\! y_k}
    \!\!\! \Psi(y_k)
    +
    2 n_k
    \begin{pmatrix}
      u & -1/z & -u_x \\
      0 & -u & 0\\
      0 & 0 & 0
    \end{pmatrix}_{\!\! y_k}
    \!\!\! \Psi(y_k).
  \end{multline}
  The (3,2) entry $-uv$ in the matrix in the first term
  will cancel against the whole first term on the left-hand side of
  \eqref{eq:gazonk}, since the only nonzero entry of $LE_{12} $ is $(LE_{21})_{32}=1$.
  Likewise, the (2,2) entries sum up to $0$.  
  Thus \eqref{eq:gazonk} is equivalent to
  \begin{equation}
    \label{eq:gazonk2}
    \dot n_k \, E_{12} \Psi(y_k) + n_k \, E_{12} \tfrac{d}{dt} \Psi(y_k)
    =
    n_k
    \begin{pmatrix}
      u/z & -u\avg{v_x}-1/z^2 & -u_x/z \\
      0 & 0 & 0\\
      0 & 0 &0
      \end{pmatrix}_{\!\! y_k}
    \!\!\! \Psi(y_k).
  \end{equation}
  Computing $\frac{d}{dt}\Psi(y_k)$ using \eqref{eq:chainrule-jump-avg},
  $\Psi_x=L\Psi$ and $\Psi_t=A\Psi$, we obtain
  \begin{equation*}
    \begin{split}
      E_{12}\tfrac{d}{dt} \Psi(y_k) &=
      E_{12} \avg{L\Psi}(y_k) \, \dot y_k + E_{12} \avg{A\Psi}(y_k)
      \\
      &= E_{12} \Bigl( L u(y_k)v(y_k)+ \avg{A}(y_k) \Bigr) \, \Psi(y_k) + E_{12} \tfrac14 \jump{A}(y_k) \, \jump{\Psi}(y_k)
      \\
      &=
      \begin{pmatrix}
        u/z & \avg{v_x}u-u_xv-1/z^2 & -u_x/z \\
        0 & 0 & 0 \\
        0 & 0 & 0
      \end{pmatrix}_{\!\! y_k}
      \Psi(y_k)
      \\ & \quad
      + \tfrac14 \underbrace{E_{12} \jump{A(y_k)}E_{12}}_{=0} \, 2z\,n_k \Psi(y_k).
    \end{split}
  \end{equation*}
  After cancelling common terms in \eqref{eq:gazonk2} we arrive at 
   \begin{equation*}
     \Bigl( \dot n_k + n_k \bigl( 2u(y_k) \avg{v_x}(y_k)-u_x(y_k)v(y_k) \bigr) \Bigr) \Psi_2(y_k)
     = 0
     ,
  \end{equation*}
  which gives the claimed equation for~$\dot n_k$.

  To prove the statement for $x_k$ and~$m_k$ one can either repeat an analogous computation 
  for the coefficient of $\delta_{x_k}$, involving $E_{23}$ rather than $E_{12}$,
  or simply use the twin Lax pair which immediately produces the result via
  the symmetry $u \leftrightarrow v$, $m \leftrightarrow n$, $y_k \leftrightarrow x_k$.
\end{proof}

\phantomsection
\addcontentsline{toc}{section}{Acknowledgements}
\section*{Acknowledgements}

Hans Lundmark was supported in part by the Swedish Research Council (Veten\-skaps\-r{\aa}det, grant 2010-5822),
and Jacek Szmigielski by the Natural Sciences and Engineering Research Council of Canada (NSERC).

\bibliographystyle{abbrv-bibulous}
\bibliography{GX-interlacing-peakons}
% rm GX-interlacing-peakons-extract.bib; bibulous.py GX-interlacing-peakons.aux

\end{document}